\documentclass[11pt,reqno]{amsart}
\usepackage{amsmath,amsthm,amssymb,bm,bbm,dsfont}
\usepackage[mathscr]{eucal}
\usepackage{amsaddr}
\usepackage{mathtools}
\usepackage[shortlabels]{enumitem}
\usepackage{centernot}
\mathtoolsset{centercolon}
\usepackage{upgreek}
\usepackage[left=2cm,right=2cm,top=2cm,bottom=2cm]{geometry}
\usepackage{setspace}
\usepackage{graphicx}
\usepackage{verbatim}
\usepackage{float}
\usepackage{placeins}
\usepackage{array}
\usepackage{booktabs}
\usepackage{threeparttable}
\usepackage[update,prepend]{epstopdf}
\usepackage{multirow}
\usepackage{amsfonts,amssymb,dsfont}
\usepackage[abs]{overpic}

\usepackage[usenames,dvipsnames]{color}
\usepackage[hidelinks]{hyperref}
\hypersetup{
	unicode=false,          
	pdftoolbar=true,        
	pdfmenubar=true,        
	pdffitwindow=false,     
	pdfstartview={FitH},    
	pdftitle={My title},    
	pdfauthor={Author},     
	pdfsubject={Subject},   
	pdfcreator={Creator},   
	pdfproducer={Producer}, 
	pdfkeywords={keyword1} {key2} {key3}, 
	pdfnewwindow=true,      
	colorlinks=true,        
	linkcolor=Red,          
	citecolor=ForestGreen,  
	filecolor=Magenta,      
	urlcolor=BlueViolet,    
}
\usepackage{doi}
\usepackage{url}
\usepackage{caption, subcaption}
\usepackage{enumitem}

\makeatletter
\ifx\@NODS\undefined%

\let\mathbb=\mathds
\else%
\fi
\makeatother

\DeclarePairedDelimiter{\ceil}{\lceil}{\rceil}

\DeclareMathOperator*{\argmax}{\arg\max}

\DeclareMathOperator{\Tr}{Tr}

\DeclareMathOperator{\Var}{Var}

\newcommand{\ket}[1]{| #1 \rangle}
\newcommand{\proj}[1]{| #1 \rangle\!\langle #1 |}

\newcommand{\be}{{\mathbf e}}

\def\0{{\mathbf{0}}}
\def\1{{\mathbf{1}}}
\def\2{{\mathbf{2}}}
\def\3{{\mathbf{3}}}
\def\4{{\mathbf{4}}}
\def\5{{\mathbf{5}}}
\def\6{{\mathbf{6}}}

\def\7{{\mathbf{7}}}
\def\8{{\mathbf{8}}}
\def\9{{\mathbf{9}}}

\def\be{\begin{equation}}
\def\ee{\end{equation}}
\def\bea{\begin{eqnarray}}
\def\eea{\end{eqnarray}}

\theoremstyle{plain}
\newtheorem{theo}{Theorem} 
\newtheorem{prop}[theo]{Proposition} 
\newtheorem{lemm}[theo]{Lemma} 
\newtheorem*{prop2}{Proposition~\ref{prop:weak}}
\newtheorem*{prop3}{Proposition~\ref{prop:strong}}
\newtheorem*{prop4}{Proposition~\ref{prop:prop2}}
\newtheorem*{prop5}{Proposition~\ref{prop:prop_h}}
\newtheorem*{prop6}{Proposition~\ref{prop:spCh}}
\newtheorem*{prop7}{Proposition~\ref{prop:prop_b}}

\theoremstyle{definition}

\theoremstyle{remark}
\newtheorem{remark}{Remark}[section]

\numberwithin{equation}{section}

\makeatletter
\newcommand{\opnorm}{\@ifstar\@opnorms\@opnorm}
\newcommand{\@opnorms}[1]{%
	\left|\mkern-1.5mu\left|\mkern-1.5mu\left|
	#1
	\right|\mkern-1.5mu\right|\mkern-1.5mu\right|
}
\newcommand{\@opnorm}[2][]{%
	\mathopen{#1|\mkern-1.5mu#1|\mkern-1.5mu#1|}
	#2
	\mathclose{#1|\mkern-1.5mu#1|\mkern-1.5mu#1|}
}
\makeatother

\begin{document}

\let\origmaketitle\maketitle
\def\maketitle{
	\begingroup
	\def\uppercasenonmath##1{} 
	\let\MakeUppercase\relax 
	\origmaketitle
	\endgroup
}

\title{\bfseries \Large{Moderate Deviation Analysis for Classical-Quantum Channels and Quantum Hypothesis Testing}}

\author{ {Hao-Chung Cheng$^{1,2}$ and Min-Hsiu Hsieh$^1$}}
\address{\small  	
$^{1}$Centre for Quantum Software and Information (UTS:Q$\ket{\text{SI}}$), \\
Faculty of Engineering and Information Technology, University of Technology Sydney, Australia\\
$^{2}$Graduate Institute Communication Engineering, National Taiwan University, Taiwan (R.O.C.) }
\email{\href{mailto:F99942118@ntu.edu.tw}{F99942118@ntu.edu.tw}}
\email{\href{mailto:Min-Hsiu.Hsieh@uts.edu.au}{Min-Hsiu.Hsieh@uts.edu.au}}

\begin{abstract}
In this work, we study the tradeoffs between the error probabilities of classical-quantum channels and the blocklength $n$ when the transmission rates approach the channel capacity at a rate slower than $1/\sqrt{n}$, a research topic known as moderate deviation analysis.  We show that the optimal error probability vanishes under this rate convergence. Our main technical contributions are a tight quantum sphere-packing bound, obtained via Chaganty and Sethuraman's concentration inequality in strong large deviation theory, and asymptotic expansions of error-exponent functions. Moderate deviation analysis for quantum hypothesis testing is also established.  The converse directly follows from our channel coding result, while the achievability relies on a martingale inequality. 

\end{abstract}

\maketitle

\section{Introduction} \label{sec:introduction}

Investigating the interplay between the transmission rate, blocklength and error probability is one of the core problems in information theory. Based on different ranges of the error probability, the analysis of communication performance roughly falls into the following three categories: (i) \emph{large error probability} or \emph{non-vanishing error probability} regime; (ii) \emph{medium error probability} regime; and (iii) \emph{small error probability} regime. In the {non-vanishing error probability} regime, the largest transmission rate, given a coding length $n$ and an error probability no more than $\epsilon$, is one of the main research focuses. Strassen \cite{Str62} first demonstrated that the maximum size of an $n$-blocklength code through a discrete memoryless channel (DMC) $\mathscr{W}$, denoted by $M^*(\mathscr{W}^n,\epsilon)$, yields an asymptotic expansion to the order $\sqrt{n}$, and hence this is called \emph{second-order analysis}:
\begin{align} \label{eq:second1}
\log M^*(W^n,\epsilon) = n C+ \sqrt{n V} \, \Phi^{-1}(\epsilon) + O(\log n),
\end{align}
where the quantities $C$ and $V$ denote the capacity \cite{Sha48} and the dispersion \cite{PPV10} of the channel, and $\Phi$ is the cumulative distribution function of a standard normal random variable. Equivalently, Eq.~\eqref{eq:second1} yields the following relationship between the optimal decoding error with blocklength $n$ and rate $C-A/\sqrt{n}$ for any constant $A$: 
\begin{align} \label{eq:second2}
\lim_{n\to+\infty} \epsilon^*\left(n, C-A/\sqrt{n} \right) = \Phi\left(\frac{A}{\sqrt{V}}\right).
\end{align}
Strassen's result relied on the \emph{Gaussian approximation} or the \emph{central limit theorem (CLT)}. His work was latter refined by Hayashi \cite{Hay09b}, Polyanskiy \textit{et al.}~\cite{PPV10}, and extended to quantum channels \cite{TH13,Li14,TV15,TBR16}. The results for higher-order asymptotics are referred to Refs.~\cite{TT13, Tan14, TT15}.

In the \emph{small error probability} regime, Shannon \cite{Sha59} introduced  the \emph{reliability function} $E(R)$ as the optimal error exponent:
\begin{align} \label{eq:exponent}
\lim_{n\to +\infty} -\frac1n\log \epsilon^* \left( n, R \right)  = E(R),
\end{align}
 for rate $R$ below the channel capacity\footnote{To the best of our knowledge, the reliability function $E(R)$ is only known in the high rate regime, i.e.~at rates above a \emph{critical rate} (see e.g.~\cite[p.~160]{Gal68}).  } $C$. This seminal work entails the \emph{error exponent analysis} of a broad class of channels \cite{Fan61, Gal68, Bla87, HHH07, CK11, CHT16b}. The exponential decay of the error probability in Eq.~\eqref{eq:exponent} is a consequence of the \emph{large deviation principle (LDP)} \cite{DZ98}. In summary, the errors in Eqs.~\eqref{eq:second2} and \eqref{eq:exponent}, respectively, fall into the CLT regime and large-deviation regime. 

Altu\u{g} and Wagner \cite{AW10, AW14b} pioneered the study of the {medium error probability} regime, and investigated the asymptotic behaviour of the optimal decoding error when the coding rate converges to capacity sufficiently slowly. Specifically, they studied under which conditions the error is asymptotically equal to\footnote{We denote $f_n\sim g_n$ if and only if $\lim_{n\to+\infty} \frac{f_n}{g_n} = 1$.}  
\begin{align} \label{eq:moderate}
\epsilon^*\left(n, C-a_n \right) \sim \Phi\left(\frac{\sqrt{n} a_n}{\sqrt{v}}\right) \sim \mathrm{e}^{\frac{-n a_n^2}{2v}},
\end{align}
where the sequence $(a_n)_{n\in\mathbb{N}}$ satisfies 
\begin{align} \label{eq:a_n}
\begin{split}
&\textnormal{(i)} \lim_{n\to+\infty} a_n = 0;\\
&\textnormal{(ii)} \lim_{n\to+\infty} a_n\sqrt{n} = +\infty. 
\end{split}
\end{align}
Evidently, the transmission rate in Eq.~\eqref{eq:moderate} approaches capacity slower than $1/\sqrt{n}$. A DMC with errors satisfying Eq.~\eqref{eq:moderate} possesses a \emph{moderate deviation property (MDP)} \cite[Section 3.7]{DZ98}. The constant $v$ in Eq.~\eqref{eq:moderate} equals the channel dispersion $V$ when both the limit in Eq.~\eqref{eq:second2}  and MDP hold \cite[Theorem 1]{PV10}. We refer the interested readers to Refs.~\cite{PV10, Sas12, AW14b} for further results in classical channel coding.
These three approaches---(i), (ii), and (iii)---all have theoretical significance and practical value, and this paper will focus on the medium error probability regime, which is rarely explored in the quantum scenario.

Our main contribution is, for any classical-quantum (c-q) channel with a non-zero dispersion $V>0$, 
\begin{align} \label{eq:main}
\lim_{n\to+\infty} \frac{\log \epsilon^* (n, C - a_n )}{n a_n^2} = -\frac1{2V},
\end{align}
where $(a_n)_{n\in\mathbb{N}}$ is any sequence satisfying Eq.~\eqref{eq:a_n}.
The result in Eq.~\eqref{eq:main} shows that reliable communication over a c-q channel is possible when the transmission rate approaches capacity at the scale slower than $1/\sqrt{n}$. Our proof employs techniques from the error exponent analysis (the LDP regime). For the achievability part, we start from Hayashi's upper bound of the average error for c-q channels \cite{Hay07} followed by an asymptotic expansion of the error-exponent function.  For the converse, we employ a sharp converse bound based on a strong large deviation inequality (Proposition~\ref{prop:strong}). 
This bound is more general than the previous result in Ref.~\cite[Proposition 14]{CHT16b}, since it allows the transmission rates to depend on the blocklength instead of being fixed.
We remark that Altu\u{g} and Wagner's converse proof \cite[Theorem 2.2]{AW14b} is not sufficient for proving Eq.~\eqref{eq:main} because their sphere-packing bound is of a weaker form in general c-q channels \cite[Theorem 6]{CHT16b} (see also \cite{DW14}). Thus, naively following their converse approach will result in a gap between the achievability and converse results (see Remark~\ref{remark1}). 

As a special case of c-q channel coding, we obtain the moderate deviations for binary quantum hypothesis testing (see Theorems \ref{theo:achievability_HT} and \ref{theo:converse_HT}):
\begin{align}
	\lim_{n\to+\infty} \frac{1}{na_n^2} \log \widehat{\alpha}_{\exp\{-n \left[ D(\rho\|\sigma) - a_n \right] \}}\left( \rho^{\otimes n} \| \sigma^{\otimes n} \right) = -\frac{1}{2V(\rho\|\sigma)},
\end{align}
where $\widehat{\alpha}_{\mu}$ denotes the smallest type-I error when the type-II error does not exceed $\mu$; $D(\rho\|\sigma)$ and $V(\rho\|\sigma)$ denote  the relative entropy and relative variance of $\rho$ and $\sigma$, respectively.
The converse part directly follows from the channel coding, and we provide two proofs for the achievability part. The first one comes from Audeneart \textit{et al.}'s error exponent analysis \cite{ANS+08}, while the second one employs a martingale inequality \cite{Sas12}.
We remark that the moderate deviation analysis for classical hypothesis testing was studied by Sason \cite{Sas12}, and by Watanabe and Hayashi \cite{WH14}. Moreover, a recent work by Rouz\'{e} and Datta \cite{RD16} formulated the quantum hypothesis problem into a martingale, which is similar to our approach for proving the achievability. 

Unlike our proof techniques relying on error exponent analysis (the LDP regime), a recent and independent paper \cite{CCT+16b} obtained the same result, but proceeds from the second-order analysis (the CLT regime). Their achievability proof follows from the one-shot capacity by Wang and Renner \cite{WR12}; while the converse part generalizes Polyanskiy and Verd\'u's result \cite{PV10} (which in turn relies on Strassen's Gaussian approximation \cite{Str62}) and a powerful inequality in probability \cite{Roz02} to the quantum scenario. 
We summarize the error behaviors in these three regimes in Table \ref{table:comparison}.

This paper is organized as follows.  We introduce notation and preliminaries in Section \ref{sec:notation}. Section \ref{sec:coding} contains our main result---the moderate deviation analysis for c-q channel coding. In Section \ref{sec:HT}, we present the moderate deviations for quantum hypothesis testing. Lastly, we conclude this paper in Section \ref{sec:conclutions}.

\begin{table}[th!]
\centering
\resizebox{1\columnwidth}{!}{
\begin{tabular}{|l|c|c|c|} 
\toprule
Error Regimes & Concentration Phenomena&  Hypothesis Testing & Channel Coding \\
			\midrule
			\midrule
\multirow{2}{*}{Large Error} & \multirow{2}{*}{CLT: $\Pr\left( S_n \geq \sqrt{n} x \right) \rightarrow 1 - \Phi\left(\frac{x}{\sqrt{v}}\right)$} & \multirow{2}{*}{$ \widehat{\alpha}_{\exp\left\{-n \left[ D - \frac{A}{\sqrt{n}} \right] \right\}} \rightarrow \Phi\left( \frac{A}{\sqrt{V}}\right)$	}& \multirow{2}{*}{$ \epsilon^*\left(n,C- \frac{A}{\sqrt{n}}\right) \rightarrow \Phi\left(\frac{A}{\sqrt{V}}\right)$} \\
&&& \\
		
\hline

\multirow{2}{*}{Medium Error} &  \multirow{2}{*}{MDP: $\Pr\left( S_n \geq n a_n x \right) = \mathrm{e}^{ -\frac{ na_n^2 }{2v} x + o(n a_n^2)}$} 
& \multirow{2}{*}{ $ \widehat{\alpha}_{\exp\{-n \left[ D - a_n  \right] \}} = \mathrm{e}^{ -\frac{n a_n^2  }{2V}+o(n a_n^2) } $}
&  \multirow{2}{*}{$\epsilon^*(n,C-a_n ) = \mathrm{e}^{- \frac{n a_n^2}{2V} + o(n a_n^2) }$} \\	&  &  & \\
	
\hline
	
\multirow{2}{*}{Small Error} & \multirow{2}{*}{LDP: $ \Pr\left( S_n \geq n x \right)= \mathrm{e}^{ -n \Lambda^*(x) + o(n) }$} 
& \multirow{2}{*}{$ \widehat{\alpha}_{\exp\{-n r \}} = \mathrm{e}^{ -n \phi(r) + o(n)}$}  
& 	\multirow{2}{*}{$\epsilon^*(n,R) = \mathrm{e}^{-n E(R) + o(n) }$} \\
	&&& \\
	
\bottomrule
	
\end{tabular}
}
\vspace{5mm}
\caption{ \small This table compares the asymptotic error behaviors of quantum hypothesis testing and classical-quantum channel coding in three error probability regimes: (i) large error (central limit theorem), (ii) medium error (moderate deviation principle), and (iii) small error (large deviation principle).
The quantity $S_n$ denotes the sum of $n$ independent and identically distributed random variables with zero mean and variance $v$. The exponent $\Lambda^*$ is the Legendre-Fenchel transform of the normalized cumulant generating function of $S_n$ \cite{DZ98}.	 The error $\widehat{\alpha}_{\exp\{-nr\}}$ is defined as the minimum type-I error with the type-II error smaller than $\exp\{-nr\}$. The quantities $D$ and $V$ in the hypothesis testing column denote the quantum relative entropy and the relative entropy variance, respectively.	The optimal error probability with blocklength $n$ and rate $R$ is denoted by $\epsilon^*(n,R)$. The quantities $C$ and $V$ in the channel coding column indicate the channel capacity and the channel dispersion, respectively. The sequence $(a_n)_{n\in\mathbb{N}}$ satisfies Eq.~\eqref{eq:a_n}. The quantity $E(R)$ is the reliability function of the channel.
}	\label{table:comparison}

\end{table}

\section{Preliminaries and Notation} \label{sec:notation}

We first introduce necessary notation. Throughout this paper, we consider a Hilbert space $\mathcal{H}$ with finite dimension $d$.  The set of density operators (i.e.~positive semi-definite operators with unit trace) and non-singular density operators on $\mathcal{H}$ are defined by $\mathcal{S(H)}$ and $\mathcal{S}_{>0}(\mathcal{H})$, respectively. The identity operator on $\mathcal{H}$ is denoted by $\mathds{1}_\mathcal{H}$, or simply $\mathds{1}$ if there is no possibility of confusion.  We use $\Tr\left[\,\cdot\, \right]$ as the trace function.  Let $\mathbb{N}$, $\mathbb{R}$, and $\mathbb{R}_{\geq 0}$ denote the set of integers, real numbers, and non-negative real numbers, respectively. Define $[n] := \{1,2,\ldots, n\}$ for $n\in\mathbb{N}$.

The power of a positive semi-definite operator $A$ is defined as: $A^p= \sum_{i:a_i\neq 0} a_i^p P_i$, where $(a_i)_i$ and $(P_i)_i$ are the eigenvalues and eigenprojections of $A = \sum_{i} a_i P_i$. 
We use $\texttt{supp}(A)$ to denote support of the operator $A$.
We write $A \ll B$ if $\texttt{supp}(A) \subset \texttt{supp}(B)$.

\subsection{Quantum Hypothesis Testing and Channel Coding} \label{ssec:HT}
Consider a binary hypothesis testing problem whose null and alternative hypotheses are $\rho\in\mathcal{S(H)}$ and $\sigma\in\mathcal{S(H)}$, respectively. The \emph{type-I error} and \emph{type-II error} of the hypothesis testing, for an operator $0\leq Q\leq \mathds{1}$, are defined as follows:
\begin{align}
\alpha\left(Q;\rho\right)&:= \Tr\left[ (\mathds{1}-Q) \rho \right],\\
\beta\left(Q;\sigma\right)&:= \Tr\left[ Q \sigma \right].
\end{align}
There is a trade-off relation between these two errors. Thus we can define the minimum type-I error when the type-II error is below $\mu\in(0,1)$ as
\begin{align} \label{eq:alpha}
\widehat{\alpha}_{\mu}\left(\rho\|\sigma\right)
:= \min_{0\leq Q\leq \mathds{1} } \big\{ \alpha\left(Q;\rho\right) : \beta\left(Q;\sigma\right) \leq \mu  \big\}.
\end{align}

Denote by $\mathcal{X}$ a finite input alphabet, and  let $\mathscr{P}(\mathcal{X})$ be the set of probability distributions on $\mathcal{X}$. For a sequence ${\mathbf{x}}^n \in \mathcal{X}^n$, we denote by 
\begin{align} \label{eq:type}
P_{\mathbf{x}^n} (x) := \frac1n \sum_{i=1}^n \mathbf{1}\left\{ x = x_i \right\},
\end{align}
where $x_i$ is the $i$-th element of $\mathbf{x}^n$.

A c-q channel $\mathscr{W}$ maps elements of $\mathcal{X}$ to the density operators in $\mathcal{S}(\mathcal{H})$, i.e.~$\mathscr{W}:x \mapsto W_x$. We denote the image of the channel $\mathscr{W}$ by
\begin{align}
\textsf{im}\left(\mathscr{W}\right) := \left\{ \rho\in\mathcal{S(H)} |\, \exists\, x\in\mathcal{X}: \rho = W_x   \right\},
\end{align}
and its closure by $\overline{\textsf{im}(\mathscr{W})}$.
Without loss of generality, we assume that $\textsf{im}\left(\mathscr{W}\right)$ has full support on the Hilbert space $\mathcal{H}$ throughout this paper.

Let $\mathcal{M}$ be a finite alphabetical set with size $M=|\mathcal{M}|$. An ($n$-block) \emph{encoder} is a map $f_n:\mathcal{M}\to \mathcal{X}^n$ that encodes each message $m\in\mathcal{M}$ to a codeword $\mathbf{x}^n(m) :=   x_1(m) \ldots x_n(m) \in\mathcal{X}^n$. The c-q channel then produces an output state $W_{\mathbf{x}^n(m)}^{\otimes n}$ with the input codeword $\mathbf{x}^n(m)$, where
\begin{align}
W_{\mathbf{x}^n(m)}^{\otimes n} = W_{x_1(m)} \otimes \cdots \otimes W_{x_n(m)} \in \mathcal{S}(\mathcal{H}^{\otimes n}).
\end{align}
The \emph{decoder} is described by a positive operator-valued measurement (POVM) $\Pi_n = \{\Pi_{n,1},\ldots, \Pi_{n,M} \}$ on $\mathcal{H}^{\otimes n}$, where $\Pi_{n,i} \geq 0$ and $\sum_{i=1}^{M} \Pi_{n,i} = \mathds{1}$. The pair $(f_n, \Pi_n) =: \mathcal{C}_n$ is called a \emph{code} with \emph{rate} $R = \frac1n \log |\mathcal{M}|$.  The error probability of  sending a message $m$ with the code $ \mathcal{C}_n$ is $\epsilon_m(\mathscr{W},\mathcal{C}_n) :=  1- \Tr\left(\Pi_{n,m} W_{\mathbf{x}^n(m)}\right)$. We use $\epsilon_\text{max}(\mathscr{W},\mathcal{C}_n) = \max_{m\in\mathcal{M}} \epsilon_m(\mathscr{W},\mathcal{C}_n) $ and $\bar{\epsilon}(\mathscr{W},\mathcal{C}_n) = \frac1M \sum_{m\in\mathcal{M}} \epsilon_m(W,\mathcal{C}_n)$ to denote the \emph{maximal} error probability and the \emph{average} error probability, respectively. 
Denote by $\epsilon^*\left(n, R\right)$ the {smallest} average error probability among all codes $\mathcal{C}_n$ with message size $|\mathcal{M}|=\exp\{nR\}$.


\subsection{Information Quantities} \label{ssec:Info}
For any $\rho, \sigma \in \mathcal{S(H)}$, we define the {quantum relative entropy}, (Petz's) quantum R\'enyi divergence \cite{Pet86}, and the {log-Euclidean R\'enyi divergence} \cite{MO14b,CHT16b}, respectively, as follows: 
\begin{align} 
&D(\rho\|\sigma) :=  \Tr \left[ \rho \left( \log \rho - \log \sigma \right) \right], \label{eq:relative} \\
&D_\alpha(\rho\|\sigma) := \frac{1}{\alpha-1}\log\Tr[\rho^\alpha\sigma^{1-\alpha}],\\
&D^\flat_\alpha(\rho\|\sigma) := \frac{1}{\alpha-1} \log \Tr \left[\mathrm{e}^{\alpha \log \rho + (1-\alpha) \log \sigma}\right].
\end{align}
We define two types of the {quantum relative entropy variances} \cite{TH13, Li14} by
\begin{align}
&V(\rho\|\sigma) :=  \Tr \left[ \rho \left( \log \rho - \log \sigma \right)^2 \right] - D(\rho\|\sigma)^2 \\
&\widetilde{V}(\rho\|\sigma) :=
\int_{0}^1\mathrm{d} t\Tr\left[ \rho^{1-t} (\log \rho - \log \sigma ) \rho^t (\log \rho - \log \sigma ) \right] - D(\rho\|\sigma)^2. \label{eq:relative_V}
\end{align}
It is well-known that both quantities are non-negative, and \begin{align} \label{eq:positive_D}
V(\rho\|\sigma) > 0 \quad \text{implies} \quad D(\rho\|\sigma )> 0.
\end{align}

We define the \emph{conditional quantum relative entropy} 
of two channels $\bar{\mathscr{W}},\mathscr{W}$ and $P\in\mathscr{P}(\mathcal{X})$ to be
\begin{align}
D\left( \bar{\mathscr{W}} \| \mathscr{W} | P \right) :=  \sum_{x\in\mathcal{X}} P(x) D\left( \bar{W}_x \| W_x \right).
\end{align}
Similarly, we define the following conditional entropic quantities for  $\sigma\in\mathcal{S(H)}$ and $P\in\mathscr{P}(\mathcal{X})$:
\begin{align}
D\left( \mathscr{W} \| \sigma | P \right) &:=  \sum_{x\in\mathcal{X}} P(x) D\left( W_x \| \sigma \right), \\
D_\alpha\left( \mathscr{W} \| \sigma | P \right) &:=  \sum_{x\in\mathcal{X}} P(x) D_\alpha\left( W_x \| \sigma \right), \\
V\left( \mathscr{W} \| \sigma | P \right) &:= \sum_{x\in\mathcal{X}} P(x) V\left( W_x \| \sigma \right), \label{eq:V_cond} \\
\widetilde{V}\left( \mathscr{W} \| \sigma | P \right) &:= \sum_{x\in\mathcal{X}} P(x) \widetilde{V}\left( W_x \| \sigma \right).
\end{align}
The \emph{mutual information} of the channel $\mathscr{W}:\mathcal{X}\to\mathcal{S(H)}$ with a prior distribution $P\in\mathscr{P}(\mathcal{X})$ is defined by
\begin{align}
I(P,\mathscr{W}) := 
D \left(P\circ \mathscr{W} \| P\otimes P\mathscr{W} \right) =
D \left( \mathscr{W} \| P \mathscr{W} | P \right),
\end{align}
where $P\circ\mathscr{W}:= \sum_{x\in\mathcal{X}}P(x) |x\rangle\langle x|\otimes W_x$
and $P\mathscr{W}:= \sum_{x\in\mathcal{X}} P(x) W_x$.
Hence, the (classical) \emph{information capacity} of the channel $\mathscr{W}$ is 
\begin{align} \label{eq:capacity}
C_\mathscr{W} := \max_{P\in\mathscr{P}(\mathcal{X})} I(P,\mathscr{W}).
\end{align}
The \emph{conditional information variance} and the \emph{unconditional information variance} of $\mathscr{W}:\mathcal{X}\to\mathcal{S(H)}$ with a prior distribution $P\in\mathscr{P}(\mathcal{X})$ are defined, respectively, by
\begin{align} \label{eq:V}
\begin{split}
V (P,\mathscr{W}) &:=  
V\left(\mathscr{W} \| P\mathscr{W} | P  \right), \\
U(P,\mathscr{W}) &:= V\left(P\circ \mathscr{W} \| P \otimes P\mathscr{W} \right).
\end{split}
\end{align}
It is known that (see e.g.~\cite[Lemma 62]{PPV10}) that $V(P^\star, \mathscr{W}) = U(P^\star,\mathscr{W})$ for every capacity-achieving distribution $P^\star\in\mathscr{P}(\mathcal{X})$, i.e.~$I(P^\star,\mathscr{W}) = C_\mathscr{W}$.
Similarly, we also define the unconditional information variance in terms of $\widetilde{V}(\rho\|\sigma)$:
\begin{align}
\widetilde{V} (P,\mathscr{W}) :=  
\widetilde{V}\left(\mathscr{W} \| P\mathscr{W} | P  \right). \label{eq:Vb}
\end{align}
The \emph{minimal peripheral information variance}  and its variant are defined by
\begin{align}
&V_\mathscr{W} :=  \min_{ P\in\mathscr{P}(\mathcal{X}):\, I(P,\mathscr{W})= C_\mathscr{W}} V(P,\mathscr{W}),  \label{eq:V2}\\
&\widetilde{V}_\mathscr{W} :=  \min_{ P\in\mathscr{P}(\mathcal{X}):\, I(P,\mathscr{W})= C_\mathscr{W}} \widetilde{V}(P,\mathscr{W}). \label{eq:V2b}
\end{align}
Furthermore, one can verify that
\begin{align} \label{eq:positive_C}
V_\mathscr{W} > 0 \quad \text{implies} \quad C_\mathscr{W} >0.
\end{align}


\subsubsection{Auxiliary functions and their properties.}
The auxiliary function of a classical-quantum channel is defined as  \cite{BH98, Hol00, FY05, FY06, HM16}
\[
E_0(s,P) := -\log \Tr\left[ \left( \sum_{x\in\mathcal{X}} P(x) W_x^{1/(1+s)} \right)^{1+s}  \right].
\]
In this paper, we will require three variants of the above auxiliary function:   $\forall s\geq 0$ and $\sigma\in\mathcal{S(H)}$,
\begin{align} 
&\widetilde{E}_0(s,P,\sigma) := sD_{1-s}\left( P\circ \mathscr{W}\| P\otimes \sigma \right)
\label{eq:mod_E0}\\
&E_\text{h}(s,P,\sigma ) := s D_{\frac{1}{1+s}}\left( \mathscr{W}\| \sigma | P \right),   \label{eq:Eh} \\
&\widetilde{E}_\text{h}(s,P,\sigma ) := s D^\flat_{\frac{1}{1+s}}\left( \mathscr{W}\| \sigma | P \right),  \label{eq:Eb}
\end{align}
where $D_\alpha$ and $D_\alpha^\flat$ are the (Petz's) quantum R\'enyi divergence and the {log-Euclidean R\'enyi divergence}, respectively.

The function $\widetilde{E}_0(s,P, \sigma)$ will play a major role in the achievability part of our main result (see Theorem~\ref{theo:achievability} in Section~\ref{sec:coding}). This quantity yields an upper bound to the average error probability (see \cite[Eq.~(9)]{Hay07}):
\begin{align} \label{eq:Hayashi}
	\bar{\epsilon}(\mathscr{W},\mathcal{C}_n) \leq 4  \exp\left\{-n \left[ \max_{0\leq s\leq 1} \max_{P\in\mathscr{P}(\mathcal{X})} \left\{-sR + \widetilde{E}_0(s,P, P\mathscr{W}) \right\} \right] \right\}.
\end{align}
Properties of $E_\text{h}$ and $\widetilde{E}_\text{h}$ will be crucial in the analysis of the converse part of our main result.

\noindent The following proposition summarizes properties of $\widetilde{E}_0(s,P,\sigma)$. We provide the proof in Appendix \ref{app:prop2}.

\begin{prop}[Properties of $\widetilde{E}_0(s,P,\sigma)$] \label{prop:prop2}
Consider a classical-quantum channel $\mathscr{W}:\mathcal{X}\to \mathcal{S(H)}$, a distribution $P\in\mathscr{P}(\mathcal{X})$, and a state $\sigma\in\mathcal{S(H)}$ with $W_x\ll \sigma$ for all $x\in \textnormal{\texttt{supp}}(P)$. Then $\widetilde{E}_0(s,P,\sigma)$ defined in Eq.~\eqref{eq:mod_E0} enjoys the following properties.
\begin{enumerate} [(a)]
\item\label{prop2-a}
$\widetilde{E}_0(s,P,\sigma)$ and its partial derivatives $\partial \widetilde{E}_0(s,P,\sigma)/\partial s$, $\partial^2 \widetilde{E}_0(s,P,\sigma)/\partial s^2$, $\partial^3 \widetilde{E}_0(s,P,\sigma)/\partial s^3$ are all continuous in $(s,P)\in \mathbb{R}_{\geq 0} \times \mathscr{P}(\mathcal{X})$.
		
\item\label{prop2-b}
For every $P\in\mathscr{P}(\mathcal{X})$, the function $\widetilde{E}_0(s,P,\sigma)$ is concave in $s\in\mathbb{R}_{\geq 0}$.
		
\item\label{prop2-c}
For every $P\in\mathscr{P}(\mathcal{X})$,
\begin{align} 
\left.\frac{ \partial \widetilde{E}_0(s,P,\sigma)}{\partial s}\right|_{s = 0} = D\left(P\circ \mathscr{W}\|P\otimes \sigma\right).
\end{align}
		
\item\label{prop2-d}
For every $P\in\mathscr{P}(\mathcal{X})$,
\begin{align} 
\lim_{s\to+\infty} \frac{ \partial \widetilde{E}_0(s,P,\sigma)}{\partial s}\leq 
\frac{ \partial \widetilde{E}_0(s,P,\sigma)}{\partial s} \leq D\left(P\circ \mathscr{W}\|P\otimes \sigma\right), \; \forall s\in\mathbb{R}_{\geq 0}. \label{eq:tE0_I1}
\end{align} 
		
\item\label{prop2-e}
For every $P\in\mathscr{P}(\mathcal{X})$,
\begin{align}
\left.\frac{ \partial^2 \widetilde{E}_0(s,P,\sigma)}{\partial s^2}\right|_{s = 0} = -V\left( P\circ\mathscr{W}\|P\otimes \sigma\right).
\end{align}
\end{enumerate}
\end{prop}

\noindent Properties of $E_{\textnormal{h}}(s,P,\sigma)$ are collected in the following proposition. The proof can be found in Appendix~\ref{app:prop_h}.

\begin{prop}[Properties of ${E}_{\textnormal{h}}(s,P,\sigma)$] \label{prop:prop_h}
Consider a classical-quantum channel $\mathscr{W}:\mathcal{X}\to \mathcal{S(H)}$, a distribution $P\in\mathscr{P}(\mathcal{X})$, and a state $\sigma\in\mathcal{S(H)}$ with $W_x\ll \sigma$ for all $x\in \textnormal{\texttt{supp}}(P)$. Then $E_\textnormal{h}(s,P,\sigma)$ defined in Eq.~\eqref{eq:Eh} enjoys the following properties.
\begin{enumerate}[(a)]
\item\label{prop_h-a}
$E_\textnormal{h}(s,P,\sigma)$ and its partial derivatives $\partial E_\textnormal{h}(s,P,\sigma)/\partial s$, $\partial^2 E_\textnormal{h}(s,P,\sigma)/\partial s^2$, $\partial^3 E_\textnormal{h}(s,P,\sigma)/\partial s^3$  are continuous for $(s,P)\in \mathbb{R}_{\geq 0} \times \mathscr{P}(\mathcal{X})$.
		
\item\label{prop_h-b}
For every $P\in\mathscr{P}(\mathcal{X})$, the function $E_\textnormal{h}(s,P,\sigma)$ is concave in $s$ for all $s\in\mathbb{R}_{\geq 0}$.
		
\item\label{prop_h-c}
For every $P\in\mathscr{P}(\mathcal{X})$,
\begin{align} \label{eq:first_h}
\left.\frac{ \partial E_\textnormal{h}(s,P,\sigma)}{\partial s}\right|_{s = 0} = D\left(\mathscr{W}\|\sigma|P\right).
\end{align}
		
\item\label{prop_h-d}
For every $P\in\mathscr{P}(\mathcal{X})$,
\begin{align} \label{eq:E_h_I}
\lim_{s\to+\infty} \frac{ \partial E_\textnormal{h}(s,P,\sigma)}{\partial s}\leq 
\frac{ \partial E_\textnormal{h}(s,P,\sigma)}{\partial s} \leq D\left(\mathscr{W}\|\sigma|P\right), \; \forall s\in\mathbb{R}_{\geq 0}.  
\end{align}
		
\item\label{prop_h-e}
For every $P\in\mathscr{P}(\mathcal{X})$,
\begin{align} \label{eq:second_h}
\left.\frac{ \partial^2 E_\textnormal{h}(s,P,\sigma)}{\partial s^2}\right|_{s = 0} = -V\left( \mathscr{W}\|\sigma|P\right).
\end{align}
\end{enumerate}
\end{prop}

\noindent Proposition~\ref{prop:prop_b} below lists the properties of $\widetilde{E}_\text{h}$, and the proof is provided in Appendix \ref{app:prop_b}.
\begin{prop}[Properties of $\widetilde{E}_\text{h}(s,P,\sigma)$] \label{prop:prop_b}
Consider a classical-quantum channel $\mathscr{W}:\mathcal{X}\to \mathcal{S(H)}$, a distribution $P\in\mathscr{P}(\mathcal{X})$, and a state $\sigma\in\mathcal{S(H)}$ with $W_x\ll \sigma$ for all $x\in \textnormal{\texttt{supp}}(P)$. Then $\widetilde{E}_\textnormal{h}(s,P,\sigma)$ defined in Eq.~\eqref{eq:Eb} enjoys the following properties.
\begin{enumerate}[(a)]
\item\label{prop_b-a}
$\widetilde{E}_\textnormal{h}(s,P,\sigma)$ and its partial derivatives $\partial \widetilde{E}_\textnormal{h}(s,P,\sigma)/\partial s$, $\partial^2 \widetilde{E}_\textnormal{h}(s,P,\sigma)/\partial s^2$, $\partial^3 \widetilde{E}_\textnormal{h}(s,P,\sigma)/\partial s^3$ are all continuous for $(s,P)\in \mathbb{R}_{\geq 0} \times \mathscr{P}(\mathcal{X})$.

\item\label{prop_b-b}
For every $P\in\mathscr{P}(\mathcal{X})$, the function $\widetilde{E}_\textnormal{h}(s,P,\sigma)$ is concave in $s$ for all $s\in\mathbb{R}_{\geq 0}$.
		
\item\label{prop_b-c}
For every $P\in\mathscr{P}(\mathcal{X})$,
\begin{align} \label{eq:first_b}
\left.\frac{ \partial \widetilde{E}_\textnormal{h}(s,P,\sigma)}{\partial s}\right|_{s = 0} = D\left(\mathscr{W}\|\sigma|P\right).
\end{align}
		
\item\label{prop_b-d}
For every $P\in\mathscr{P}(\mathcal{X})$,
\begin{align} \label{eq:E_b_I}
\lim_{s\to+\infty} \frac{ \partial \widetilde{E}_\textnormal{h}(s,P,\sigma)}{\partial s}\leq 
\frac{ \partial \widetilde{E}_\textnormal{h}(s,P,\sigma)}{\partial s} \leq D\left(\mathscr{W}\|\sigma|P\right), \; \forall s\in\mathbb{R}_{\geq 0}.  
\end{align}
		
\item\label{prop_b-e}
For every $P\in\mathscr{P}(\mathcal{X})$,
\begin{align} \label{eq:second_b}
\left.\frac{ \partial^2 \widetilde{E}_\textnormal{h}(s,P,\sigma)}{\partial s^2}\right|_{s = 0} = -\widetilde{V}\left( \mathscr{W}\|\sigma|P\right).
\end{align}
\end{enumerate}
\end{prop}

\subsubsection{Error Exponents.}
Auxiliary functions allow us to concisely define sphere-packing exponent functions of a classical-quantum channel. We will use notation similar to Refs.~\cite{Win99,CHT16bb, CHT16b}. Define
\begin{align} 
\widetilde{E}_\text{sp}(R,P,\sigma) &:=
\min_{\bar{\mathscr{W}}:\mathcal{X}\to\mathscr{S}_\circ} \left\{
D\left(\bar{\mathscr{W}}\|\mathscr{W}|P\right): D\left(\bar{\mathscr{W}}\|\sigma|P\right) \leq R \right\} \label{eq:sp_2P} \\
&=  \sup_{s\geq 0} \left\{  \widetilde{E}_\text{h}(s,P) - sR \label{eq:sp_2Pa}
\right\}, \\
E_\text{sp}^{(2)} (R,P,\sigma) &:= \sup_{s\geq 0} \left\{ E_\text{h}\left(s,P\right) -sR \right\}, \label{eq:Esp_2} 
\end{align}
for all $R>0$, $P\in\mathscr{P}(\mathcal{X})$, and $\sigma \in\mathcal{S}_{>0}(\mathcal{H})$. The equality in Eq.~(\ref{eq:sp_2Pa}) follows from \cite[Theorem 6]{CHT16b}.
%
From the definitions in Eqs.~\eqref{eq:sp_2P} and \eqref{eq:sp2}, it is not hard to see that \cite{ANS+08}
\begin{align} \label{eq:weak_0}
\widetilde{E}_\text{sp}(R,P,\sigma) =  0, \quad \forall R\geq D\left(\mathscr{W}\|\sigma|P\right).
\end{align}
and
\begin{align} \label{eq:R_inf}
E_\text{sp}^{(2)} (R,P,\sigma) = 
\begin{cases}
+\infty, & R< D_0\left(\mathscr{W}\|\sigma|P\right), \\
0, &R \geq D\left(\mathscr{W}\|\sigma|P\right).
\end{cases}
\end{align}
%

\section{Moderate Deviations for Classical-Quantum Channels} \label{sec:coding}

This section presents our main results---the error performance of classical-quantum channels satisfies the moderate deviation property, Eq.~\eqref{eq:moderate}.  The achievability part is stated in Theorem~\ref{theo:achievability}, and its proof is given in Section~\ref{ssec:proof:achievability}. Our proof strategy employs Hayashi's bound \cite{Hay07} and the properties of the modified auxiliary function (Proposition~\ref{prop:prop2}). Theorem~\ref{theo:converse} contains the converse part, and is proved in Section~\ref{ssec:proof:converse}. The proof involves a weak sphere-packing bound (Proposition~\ref{prop:weak}), a sharp converse lower bound (Proposition~\ref{prop:strong}), and an approximation of the error-exponent function around capacity (Proposition~\ref{prop:spCh}).

Let $(a_n)_{n\in\mathbb{N}}$ be a sequence of real numbers  satisfying
\begin{align} \label{eq:cond}
\begin{split}
&\textnormal{(i)} \;a_n \to 0, \quad\text{as} \quad n\to +\infty,\\
&\textnormal{(ii)} \; a_n \sqrt{n} \to +\infty, \quad\text{as}\quad n\to +\infty.
\end{split}
\end{align}

\begin{theo}[Achievability] \label{theo:achievability}
For any $\mathscr{W}:\mathcal{X}\to\mathcal{S(H)}$ with $V_\mathscr{W}>0$ and any sequence $(a_n)_{n\geq1}$ satisfying Eq.~\eqref{eq:cond}, there exists a sequence of codes $\{\mathcal{C}_n\}_{n\geq 1}$ with rates $R_n = C_\mathscr{W} - a_n$ so that
\begin{align} \label{eq:sup1}
\limsup_{n\to +\infty} \frac{1}{n a_n^2} \log \bar{\epsilon} \left( \mathscr{W}, \mathcal{C}_n \right) \leq - \frac{1}{2V_\mathscr{W}}.
\end{align}
\end{theo}
\noindent The proof is given in Section \ref{ssec:proof:achievability}.

\begin{theo}[Converse] \label{theo:converse}
For any $\mathscr{W}:\mathcal{X}\to\mathcal{S(H)}$ with $V_\mathscr{W}>0$,  any sequence $\{a_n\}_{n\geq 1}$ satisfying Eq.~\eqref{eq:cond}, and any sequence of codes $\left\{ \mathcal{C}_n  \right\}_{n\geq 1}$ with rates $R_n  = C_\mathscr{W} - a_n$, it holds that 
	\begin{align} \label{eq:converse1}
	\liminf_{n\to +\infty} \frac{1}{na_n^2} \log \bar{\epsilon} \left( \mathscr{W}, \mathcal{C}_n \right) \geq - \frac{1}{2V_\mathscr{W}}.
	\end{align}	
\end{theo}
\noindent The proof is given in Section \ref{ssec:proof:converse}.

\begin{remark} \label{remark1}
Altu\u{g} and Wagner \cite{AW14b} proved Theorem~\ref{theo:converse} for discrete classical channels by a weak sphere-packing bound with the expression of $\widetilde{E}_\text{sp}$. Although such a weak sphere-packing bound indeed holds for c-q channels (see Proposition~\ref{prop:weak} and Remark \ref{remark2} in Appendix~\ref{app:weak}),
Proposition~\ref{prop:spCh} in Section~\ref{ssec:proof:converse} shows that it will lead to
\begin{align}
\limsup_{n\to +\infty} \frac{1}{n a_n^2} \log \bar{\epsilon} \left( \mathscr{W}, \mathcal{C}_n \right) \leq - \frac{1}{2 \widetilde{V}_\mathscr{W}},
\end{align}
where $\widetilde{V}_\mathscr{W}$ is defined in Eq.~\eqref{eq:V2b}.
Since $\widetilde{V}(\rho\|\sigma) \leq {V}(\rho\|\sigma)$ \cite[Theorem 1.2]{Bou06}, it holds that $\widetilde{V}_\mathscr{W} \leq {V}_\mathscr{W}$ and the equality happens if and only if the channel reduces to classical. Hence, Altu\u{g} and Wagner's method yields a weaker result in quantum regime; namely, a gap between the achievability and the converse. In Section~\ref{ssec:proof:converse}, we will employ a sharp converse bound from strong large deviation theory to achieve our result, Theorem \ref{theo:converse}.
\end{remark}

\subsection{Proof of Achievability:~Theorem \ref{theo:achievability}} \label{ssec:proof:achievability}
Let $\mathscr{W}:\mathcal{X}\to\mathcal{S(H)}$ satisfy $V_\mathscr{W}>0$. 
Let $\{a_n\}_{n\geq 1}$ be any sequence of real numbers satisfying Eq.~\eqref{eq:cond}. 
Since $V_\mathscr{W}>0$, Eq.~\eqref{eq:positive_C} shows that $C_\mathscr{W}>0$.
Hence, we have $C_\mathscr{W} - a_n >0$, for all sufficiently large $n$. Fix such an integer $n$ onwards, Hayashi's upper bound, Eq.~\eqref{eq:Hayashi}, implies that there exists a code $\mathcal{C}_n$ with $R_n = C_\mathscr{W} - a_n$ so that
\begin{align}
\bar{\epsilon}(\mathscr{W},\mathcal{C}_n) \leq 4 \exp \left( -n \left[ \max_{0\leq s \leq 1} \left\{
\widetilde{E}_0(s,P, P\mathscr{W}) - s R_n
\right\} \right] \right),
\end{align}
for all $P\in\mathscr{P}(\mathcal{X})$.  
In the following, we denote by $\widetilde{E}_0(s,P) := \widetilde{E}_0(s,P,P\mathscr{W})$ for notational convenience.
Simple algebra yields
\begin{align} \label{eq:sup5}
\frac{1}{n a_n^2} \log \bar{\epsilon}(W, \mathcal{C}_n) \leq \frac{\log 4}{na_n^2} - \frac1{a_n^2} \max_{0\leq s\leq 1} \left\{ \widetilde{E}_0(s,P) - s R_n\right\},
\end{align}
for all sufficiently large $n$ and any $P\in\mathscr{P}(\mathcal{X})$.

Let $\widetilde{\mathscr{P}}(\mathcal{X})$ be the set of distributions that achieve the minimum in Eq.~\eqref{eq:V2}, and let  $\widetilde{P}\in\widetilde{\mathscr{P}}(\mathcal{X})$.
Note that Ref.~\cite[Lemma 3]{TT13} implies that $\widetilde{\mathscr{P}}(\mathcal{X})$ is compact.
Applying Taylor's theorem to $\widetilde{E}_0(s,\widetilde{P})$ at $s=0$ together with Proposition \ref{prop:prop2} gives
\begin{align} \label{eq:sup2}
\widetilde{E}_0\left(s, \widetilde{P}\right) = sC_\mathscr{W} - \frac{s^2}{2} V_\mathscr{W} + 
\left. \frac{s^3}{6} \frac{\partial^3 \widetilde{E}_0\left(s,\widetilde{P}\right)}{\partial s^3} \right|_{s = \bar{s}},
\end{align}
for some $\bar{s} \in [0,s]$. Let $s_n = a_n/V_\mathscr{W}$. Then  $s_n\leq 1$ for all sufficiently large $n$ by the assumption in Eq.~\eqref{eq:cond} and $V_\mathscr{W} >0$. For all $s_n\leq 1$, Eq.~\eqref{eq:sup2} yields 
\begin{align} 
\max_{0\leq s \leq 1} \left\{\widetilde{E}_0\left(s,\widetilde{P}\right) - sR_n
\right\}
&\geq \widetilde{E}_0\left(s_n,\widetilde{P}\right) - s_n R_n \\
&= 
\frac{a_n}{V_\mathscr{W}} \left( C_\mathscr{W} - R_n \right)
-\frac{a_n^2}{2 V_\mathscr{W}}  + \frac{a_n^3}{6 V_\mathscr{W}^3}
\left. \frac{\partial^3 \widetilde{E}_0\left(s,\widetilde{P}\right)}{\partial s^3} 	\right|_{s = \bar{s}_n} \\
&= \frac{a_n^2}{2 V_\mathscr{W}}  + \frac{a_n^3}{6 V_\mathscr{W}^3}
\left. \frac{\partial^3 \widetilde{E}_0\left(s,\widetilde{P}\right)}{\partial s^3} 	\right|_{s = \bar{s}_n}, \label{eq:sup3}
\end{align}
where $\bar{s}_n\in[0,s_n]$ and Eq.~\eqref{eq:sup3} holds since {$R_n = C_\mathscr{W} - a_n$}.

Define
\begin{align} \label{eq:Upsilon}
\Upsilon = \max_{ (s,P)\in[0,1]\times \widetilde{\mathscr{P}}(\mathcal{X})} 
\left| \frac{\partial^3 \widetilde{E}_0\left(s,{P}\right)}{\partial s^3} \right|,
\end{align}
which is finite due to the compact set $[0,1]\times \widetilde{\mathscr{P}}(\mathcal{X})$ and item \ref{prop2-a} in Proposition \ref{prop:prop2}.
Therefore, Eq.~\eqref{eq:sup3} implies that
\begin{align}
\max_{0\leq s \leq 1} \left\{ \widetilde{E}_0 \left(s, \widetilde{P}\right) - sR_n
\right\}
&\geq \frac{a_n^2}{2 V_\mathscr{W}}  + \frac{a_n^3}{6 V_\mathscr{W}^3}
\left. \frac{\partial^3 \widetilde{E}_0\left(s,\widetilde{P}\right)}{\partial s^3} 	\right|_{s = \bar{s}_n} \\
&\geq \frac{a_n^2}{2 V_\mathscr{W}}  - \frac{a_n^3}{6 V_\mathscr{W}^3}
\left|\left. \frac{\partial^3 \widetilde{E}_0\left(s,\widetilde{P}\right) }{\partial s^3} 	\right|_{s = \bar{s}_n} \right| \\	
&\geq 
\frac{a_n^2}{2 V_\mathscr{W}}  - \frac{a_n^3}{6 V_\mathscr{W}^3}\Upsilon, \label{eq:sup4}
\end{align}
for all sufficiently large $n$.

Substituting Eq.~\eqref{eq:sup4} into Eq.~\eqref{eq:sup5} gives
\begin{align}
\frac{1}{na_n^2} \log \bar\epsilon (\mathscr{W}, \mathcal{C}_n) \leq \frac{\log 4}{ na_n^2}  - \frac{1}{2V_\mathscr{W}} \left( 1 - \Upsilon \frac{a_n}{3V_\mathscr{W}^2} \right).
\end{align}
Recall Eq.~\eqref{eq:cond} and let $n\to+\infty$, which completes the proof:
\begin{align}
\limsup_{n\to +\infty} \frac{1}{na_n^2} \log \bar\epsilon(\mathscr{W},\mathcal{C}_n) \leq -\frac{1}{2V_\mathscr{W}}.
\end{align}
\qed

\subsection{Proof of Converse:~Theorem \ref{theo:converse}} \label{ssec:proof:converse}

Our strategy consists of the following steps.
First, we claim that it suffices to prove Eq.~\eqref{eq:converse1} for the maximal error probability of any code $\mathcal{C}_n$, i.e.~$\epsilon_{\max}(\mathscr{W}, \mathcal{C}_n)$.
Recall the standard expurgation method (see e.g.~\cite[p.~96]{SGB67}, \cite[Theorem 20]{Bla74}, \cite[p.~395]{Bla87}): by removing half codewords with highest error probability to arrive at $\bar{\epsilon} \left( \mathscr{W}, \mathcal{C}_n \right)  \geq \frac12 {\epsilon}_{\max} \left( \mathscr{W}, \mathcal{C}_n' \right) $ with $|\mathcal{C}_n'| = \ceil{|\mathcal{C}_n|/2} \geq \frac12\exp\{nR_n\} = \exp\{n(R_n - \frac1n\log 2)\}$. Since the induced rate back-off is only $ \frac1n \log 2 = o(a_n)$, one might define another sequence $a_n' := a_n - \frac1n \log 2$  satisfying Eq.~\eqref{eq:cond}. Hence, without of loss generality, we only need to prove the converse part for $\epsilon_{\max}$.

Second, we employ the method of Ref.~\cite[Lemma 16]{CHT16b} to relate the error probability $\epsilon_{\max}$ to the minimum type-I error:
\begin{align}
\frac{\log \epsilon_{\max}(\mathscr{W}, \mathcal{C}_n) }{n a_n^2} &\geq \max_{\sigma^n \in \mathcal{S}(\mathcal{H}^{\otimes n})} \min_{\mathbf{x}^n\in\mathcal{X}^n}
\frac{\log \widehat{\alpha}_{ \exp\{-nR_n\} }(W_{\mathbf{x}^n}^{\otimes n}\| \sigma^n) }{n a_n^2}\\
&\geq  \min_{\mathbf{x}^n\in\mathcal{X}^n}
\frac{\log \widehat{\alpha}_{ \exp\{-nR_n\} }(W_{\mathbf{x}^n}^{\otimes n}\|  (P^\star\mathscr{W})^{\otimes n}   ) }{n a_n^2},
 \label{eq:converse0}
\end{align}
where  $P^\star \in \mathscr{P}(\mathcal{X})$ is an arbitrary capacity-achieving distribution, i.e.~$I(P^\star,\mathscr{W}) = C_\mathscr{W}$.

Third, we divide the set of codewords into two groups. Fix an arbitrary $\eta\in(0,\frac12)$. Let $A := \max_{\rho\in\mathscr{S}_\circ} V(\rho\|P^\star\mathscr{W})$ and let $\xi = \sqrt{2A/\eta}$.
Define:
\begin{align} \label{eq:set}
&\Omega_{\text{good}} := \left\{ \mathbf{x}^n \in \mathcal{X}^n: D(\mathscr{W}\|P^\star \mathscr{W}|P_{\mathbf{x}^n}) > R_n  \right\};\\
&\Omega_{\text{bad}} := \mathcal{X}^n \backslash \Omega_{\text{good}}. 
\end{align}
For the codes in $\Omega_{\text{bad}}$, we employ a weak converse bound in Proposition~\ref{prop:weak}, and apply a sharp converse bound, Proposition~\ref{prop:strong}, for $\Omega_{\text{good}}$. Furthermore, we can assume $a_n>0$ for all sufficiently large $n\in\mathbb{N}$ owing to the assumption $\lim_{n\to+\infty} a_n \sqrt{n}  = +\infty$. Subsequently, we will consider such $n$ onwards.

\begin{proof}[Proof of Theorem \ref{theo:converse}]
We start the proof with the case $\Omega_{\text{bad}}$, and further consider two different cases: 
\begin{align}
&\Omega_{\text{bad}}^{(1)} := \left\{ \mathbf{x}^n \in \mathcal{X}^n: D(\mathscr{W}\|P^\star \mathscr{W}|P_{\mathbf{x}^n}) \leq R_n - \frac{2\xi}{\sqrt{n}}   \right\}; \\
&\Omega_{\text{bad}}^{(2)} := \left\{ \mathbf{x}^n \in \mathcal{X}^n: R_n - \frac{2\xi}{\sqrt{n}}< D(\mathscr{W}\|P^\star \mathscr{W}|P_{\mathbf{x}^n}) \leq R_n    \right\}.
\end{align}
We apply the following weak converse bound with $\sigma = P^\star \mathscr{W}$, whose proof is provided in Appendix \ref{app:weak} to further lower bound the right-hand side of Eq.~\eqref{eq:converse0}.
\begin{prop} 
[A Weak Converse Bound] \label{prop:weak}
Consider a classical-quantum channel $\mathscr{W}:\mathcal{X}\to\mathcal{S(H)}$ with $\mathscr{S}_\circ := \overline{\textnormal{\textsf{im}}(\mathscr{W})}$, an arbitrary rate $R\geq 0$, and $\sigma\in\mathcal{S}_{>0}(\mathcal{H})$.
For any $\eta\in(0,\frac12)$, let $N_0\in\mathbb{N}$ such that for all $n\geq N_0$,
\begin{align} \label{eq:weak_cond}
 \mathrm{e}^{ - \xi \sqrt{ n}  } \leq \frac{\eta}{2},
\end{align}
where $\xi = \sqrt{ 2A / \eta}$ and $A := \max_{ \rho\in\mathscr{S}_\circ  }  V(\rho\|\sigma)$.
Then, it holds that for all $n\geq N_0$,
\begin{align} \label{eq:SP2}
\widehat{\alpha}_{\exp\{-nR\}}\left( W_{\mathbf{x}^n}^{\otimes n}\| \sigma^{\otimes n} \right) 
\geq  f(\eta) \exp\left\{-n \left[ \frac{\widetilde{E}_\textnormal{sp}\left(R- \frac{2\xi}{\sqrt{n}} , P_{\mathbf{x}^n} , \sigma \right)}{1-\eta}  \right] \right\},
\end{align}
where $f(\eta) = \exp\left\{ -\frac{h\left(1-\eta\right)}{1-\eta} \right\}$
and $h(p) := - p \log p - (1-p) \log (1-p)$ is the binary entropy function.
\end{prop}

Let $\eta$ and $\xi$ be defined as above, and let $N_1$ be an integer satisfying Eq.~\eqref{eq:weak_cond}. 
Then Eq.~\eqref{eq:SP2} gives, for all $n\geq N_1$,
\begin{align}\label{eq:wmod1}
\frac{ \log \widehat{\alpha}_{\exp\{-nR_n\}}(W_{\mathbf{x}^n}^{\otimes n}\| (P^\star\mathscr{W})^{\otimes n})  }{ n a_n^2  } 
&\geq - \frac{ \widetilde{E}_\text{sp}\left( R_n - \frac{2\xi}{\sqrt{n}} , P_{\mathbf{x}^n}, P^\star\mathscr{W} \right) }{ a_n^2 (1-\eta) } + \frac{ \log f(\eta)}{ n a_n^2}.
\end{align}
Further, Eq.~\eqref{eq:weak_0} implies that for all $\mathbf{x}^n \in \Omega_{\text{bad}}^{(1)}$,
\begin{align}
\widetilde{E}_\text{sp}\left( R_n - \frac{2\xi}{\sqrt{n}} , P_{\mathbf{x}^n}, P^\star\mathscr{W} \right) = 0. \label{eq:converse3}
\end{align}
Hence, we have for all $\mathbf{x}^n\in\Omega_{\text{bad}}^{(1)}$,
\begin{align}
\frac{ \log \widehat{\alpha}_{\exp\{-nR_n\}}(W_{\mathbf{x}^n}^{\otimes n}\| (P^\star\mathscr{W})^{\otimes n})  }{ n a_n^2  } 
&\geq  \frac{ \log f(\eta)}{ n a_n^2} \\
&\geq - \frac{1}{2V_\mathscr{W}} +  \frac{ \log f(\eta)}{ n a_n^2} , 
\end{align}
where the last inequality follows from $V_\mathscr{W} > 0$. 
Since $f(\eta)<+\infty$, taking the infimum limit of $n\to +\infty$ and using Eq.~\eqref{eq:cond} give, for all $\mathbf{x}^n \in \Omega_{\text{bad}}^{(1)}$,
\begin{align}
\liminf_{n\to+\infty} \frac{ \log \widehat{\alpha}_{\exp\{-nR_n\}}\left(W_{\mathbf{x}^n}^{\otimes n}\| (P^\star\mathscr{W})^{\otimes n}) \right)  }{ n a_n^2  }\geq -\frac{1}{2V_\mathscr{W}}. \label{eq:conv1}
\end{align}

Next, we move on to $\mathbf{x}^n \in \Omega_{\text{bad}}^{(2)}$. In this case, $\widetilde{E}_\text{sp}$ in Eq.~\eqref{eq:wmod1} is not equal to zero for any finite $n$, we employ Eq.~\eqref{eq:spChc2} in Proposition~\ref{prop:spCh} below with $\delta_n = a_n + 2\xi/\sqrt{n}$ and $b_n = a_n$ to arrive at 
\begin{align}
\liminf_{n\to+\infty} \frac{ \log \widehat{\alpha}_{\exp\{-nR_n\}}\left(W_{\mathbf{x}^n}^{\otimes n}\|(P^\star\mathscr{W})^{\otimes n})\right)  }{ n a_n^2  } &\geq
- \lim_{n\to+\infty}  \frac{4\xi^2}{n \left(a_n + \frac{2\xi}{\sqrt{n}}\right)^2} \cdot \frac{1}{ 2 \widetilde{V}_{\mathscr{W}}   (1-\eta)  } \\
&= 0 \\
&\geq -\frac{1}{2 V_\mathscr{W}}, \label{eq:conv2}
\end{align}
where the equality follows since $\lim_{n\to+\infty} n a_n^2 = +\infty$.

In the last case of $\mathbf{x}^n \in \Omega_{\text{good}}$, we employ a tighter bound, Proposition~\ref{prop:strong}, to lower bound the right-hand side of Eq.~\eqref{eq:converse0}.  The proof is delayed to Appendix \ref{app:sharp}.

\begin{prop} [A Sharp Converse Bound] \label{prop:strong}
Consider a classical-quantum channel $\mathscr{W}:\mathcal{X}\to\mathcal{S(H)}$ and a state $\sigma\in\mathcal{S(H)}$. Suppose the sequence $\mathbf{x}^n\in\mathcal{X}^n$ satisfies 	\begin{align} \label{eq:cond_V2}
\nu \leq V\left( \mathscr{W} \| \sigma | P_{\mathbf{x}^n} \right) < +\infty
\end{align}
for some $\nu > 0$, and suppose the sequence of rates $(R_n)_{n\in\mathbb{N}}$ satisfies\footnote{Note that $D_0(\mathscr{W}\|\sigma|P) = D(\mathscr{W}\|\sigma|P)$ implies $W_x = \sigma $ for all $x\in\textsf{supp}(P)$ \cite[Collorary~4.1]{Tom16}. This further gives $V(\mathscr{W}\|\sigma|P) = 0$. However, the assumption in Eq.~\eqref{eq:cond_V2} ensures that $\liminf_{n\in\mathbb{N}} D(\mathscr{W}\|\sigma|P_{\mathbf{x}^n}) - D_0(\mathscr{W}\|\sigma|P_{\mathbf{x}^n}) > 0$. Hence, the intervals $[D_0(\mathscr{W}\|\sigma|P_{\mathbf{x}^n}),D(\mathscr{W}\|\sigma|P_{\mathbf{x}^n})]$ for all $\mathbf{x}^n$ satisfying Eq.~\eqref{eq:cond_V2} are not measure zero.} $ D_0(\mathscr{W}\|\sigma|P_{\mathbf{x}^n})  < R_n  < D(\mathscr{W}\|\sigma|P_{\mathbf{x}^n})$. Then, there exists an $N_0\in\mathbb{N}$ such that, for all $n\geq N_0$,
\begin{align}  \label{eq:SP3}
\widehat{\alpha}_{\exp\{-nR_n\}}(W_{\mathbf{x}^n}^{\otimes n}\| \sigma^{\otimes n}) 
\geq 
\frac{A}{s_n^\star\sqrt{n}}	\exp\left\{ 	-n  E_\textnormal{sp}^{(2)} \left(R_n - c_n ,P_{\mathbf{x}^n}, \sigma \right) \right\},	
\end{align}
where $c_n = \frac{K\log n}{n}$ and $A,K>0$ are finite constants independent of the sequence $\mathbf{x}^n$, and 
\begin{align}
s_n^\star := \argmax_{s\geq 0} \left\{ E_\textnormal{h}(s,P_{\mathbf{x}^n}, \sigma) - s R_n
\right\}.
\end{align}
\end{prop}
Before applying Proposition~\ref{prop:strong}, we verify that the condition, Eq.~\eqref{eq:cond_V2}, is satisfied.
Define
\begin{align}
v(\delta) := \min_{P\in\mathscr{P}(\mathcal{X})} \left\{ V\left(\mathscr{W}\| P^\star \mathscr{W}| P \right) : D(\mathscr{W}\|P^\star\mathscr{W}|P) \geq C_\mathscr{W}-\delta   \right\}.
\end{align}
Note that the map $\delta\mapsto v(\delta)$ is monotone decreasing and continuous at $0$ from above, i.e.~$\lim_{\delta \downarrow 0} v(\delta ) =  v(0) = V_\mathscr{W}$ \cite[Lemma 22]{TV15}.
For any $\kappa\in(0,1)$, we can choose a sufficiently small $\gamma>0$ independent of the sequence $\mathbf{x}^n$ such that $v(\gamma) \geq (1-\kappa) V_\mathscr{W} =: \nu >0$.  
Further, let $N_2\in\mathbb{N}$ such that $a_n \leq \gamma$ for all $n\geq N_2$.
Then, one finds, for all $ \mathbf{x}^n \in \Omega_{\text{good}}$ and $n\geq N_2$,
\begin{align}\label{eq_condV}
V\left( \mathscr{W}\| P^\star \mathscr{W} | P_{\mathbf{x}^n} \right) \geq v(\gamma) \geq  \nu >0.
\end{align}
Moreover, since $V_\mathscr{W}>0$ implies that $C_\mathscr{W} = \max_{P\in\mathscr{P}(\mathcal{X})} D(\mathscr{W}\|P^\star \mathscr{W} |P) > \max_{P\in\mathscr{P}} D_0(\mathscr{W}\|P^\star \mathscr{W}|P)$, one can choose a sufficiently large $n$, say $N_3\in\mathbb{N}$, such that $R_n > D_0(\mathscr{W}\|P^\star \mathscr{W}|P_{\mathbf{x}^n})$ for all $n\geq N_3$.
Now, we have for all $\mathbf{x}^n\in\Omega_\text{good}$ and $n\geq \max_\{N_2, N_3\}$ that
\begin{align}
\max_{P\in\mathscr{P}(\mathcal{X})} D_0(\mathscr{W}\|P^\star \mathscr{W}|P) &< R_n < D(\mathscr{W}\|P^\star\mathscr{W}| P_{\mathbf{x}^n}); \\
0&<\nu\leq V(\mathscr{W}\|P^\star\mathscr{W}|P_{\mathbf{x}^n}).
\end{align}
Together with Eqs.~\eqref{eq:converse0} and \eqref{eq_condV} and letting $\sigma = P^\star \mathscr{W}$, Proposition~\ref{prop:strong} yields, for all $\mathbf{x}^n\in\Omega_{\text{good}}$ and all sufficiently large $n$, say $n\geq N_4\in\mathbb{N}$,
\begin{align}
\frac{ \log \widehat{\alpha}_{\exp\{-nR_n\}}\left(W_{\mathbf{x}^n}^{\otimes n}\|(P^\star \mathscr{W})^{\otimes n} \right)  }{ n a_n^2  }
&\geq 
-\frac{E_\textnormal{sp}^{(2)} \left(R_n - c_n ,P_{\mathbf{x}^n}, P^\star \mathscr{W} \right)}{a_n^2} - \frac{ \log s_n^\star\sqrt{n}  }{n a_n^2} + \frac{\log A}{n a_n^2}. \label{eq:strong2}
\end{align}
Recall Eq.~\eqref{eq:spChc3} in Proposition~\ref{prop:spCh} below with $b_n = 0$ and $\delta_n = a_n + c_n$ that $\limsup_{n\to+\infty} \frac{s_n^\star}{a_n+c_n}\leq \frac{1}{V_\mathscr{W}}$. 
Hence, one can fix an arbitrary $\zeta>0$ and there exists an $N_5\in\mathbb{N}$ such that $\frac{s_n^\star\sqrt{n}}{(a_n+c_n)\sqrt{n}} \leq \frac{1}{V_\mathscr{W}}+\zeta$ for all $n\geq N_5$. This then leads to for all sufficiently large $n\geq \max\{N_2,N_3,N_4,N_5\}$ and all $\mathbf{x}^n\in\Omega_{\text{good}}$,
\begin{align}
\frac{ \log \widehat{\alpha}_{\exp\{-nR_n\}}\left(W_{\mathbf{x}^n}^{\otimes n}\|(P^\star \mathscr{W})^{\otimes n} \right)  }{ n a_n^2  }
&\geq 
-\frac{E_\textnormal{sp}^{(2)} \left(R_n - c_n ,P_{\mathbf{x}^n}, P^\star \mathscr{W} \right)}{a_n^2} - \frac{ \log (a_n+c_n)\sqrt{n}  }{n a_n^2} + \frac{\log \frac{A}{ \frac{1}{V_\mathscr{W}}+\zeta } }{n a_n^2}. \label{eq:strong3}
\end{align}
Taking $n\to+\infty$, the second and the third terms on the right-hand side of Eq.~\eqref{eq:strong3} vanish since $c_n = K \frac{\log n}{n} = o(a_n)$ and the assumption $\lim_{n\to+\infty} a_n \sqrt{n} = +\infty$.

Next, we apply Eq.~\eqref{eq:spChc1} in Proposition~\ref{prop:spCh} again
to bound the error-exponent function $E_\text{sp}^{(2)}$ in Eq.~\eqref{eq:strong2}: for all $\mathbf{x}^n \in \Omega^{(3)}$
\begin{align}
\liminf_{n \to +\infty} \frac{ \log \widehat{\alpha}_{\exp\{-nR_n\}}\left(W_{\mathbf{x}^n}^{\otimes n}\| (P^\star \mathscr{W})^{\otimes n} \right)  }{ n a_n^2  } &\geq -\limsup_{n \to +\infty} \frac{E_\textnormal{sp}^{(2)} \left(C_\mathscr{W} - \delta_n ,P_{\mathbf{x}^n}, P^\star \mathscr{W} \right)}{a_n^2}   \\
&= -\limsup_{n \to +\infty} \frac{E_\textnormal{sp}^{(2)} \left(C_\mathscr{W} - \delta_n ,P_{\mathbf{x}^n}, P^\star \mathscr{W} \right)}{\delta_n^2}   \\
&\geq -\frac{1}{2V_\mathscr{W}}. \label{eq:conv3}
\end{align}
Finally, combining Eqs.~\eqref{eq:converse0}, \eqref{eq:conv1}, \eqref{eq:conv2} and \eqref{eq:conv3}
concludes the desired Eq.~\eqref{eq:converse1}.

\begin{prop}[Error Exponent around Capacity] \label{prop:spCh}
	Let $(b_n)_{n\in\mathbb{N}}$ be a sequence of real numbers with $\lim_{n\to+\infty}  b_n = 0$ and let  $(\delta_n)_{n\in\mathbb{N}}$ be a sequence of positive numbers with $\lim_{n\to+\infty} \delta_n = 0$.
	Suppose the sequence of distributions $(P_n)_{n\in\mathbb{N}}$ satisfies
	\begin{align}
	C_{\mathscr{W}} - \delta_n < D(\mathscr{W}\| P^\star \mathscr{W}| P_n) \leq 	C_{\mathscr{W}} - b_n. \label{eq:spChc}
	\end{align}
	The following hold:
\begin{align}
\limsup_{n\to+\infty} \frac{ E_\textnormal{sp}^{(2)} \left(C_\mathscr{W} -\delta_n, P_n, P^\star \mathscr{W} \right) }{ \delta_n^2 } &\leq \limsup_{n\to+\infty} \frac{(\delta_n - b_n )^2}{2 V_\mathscr{W} \delta_n^2}; \label{eq:spChc1} \\
\limsup_{n\to+\infty} \frac{ \widetilde{E}_\textnormal{sp} \left(C_\mathscr{W} -\delta_n, P_n, P^\star \mathscr{W} \right) }{ \delta_n^2 } &\leq \limsup_{n\to+\infty} \frac{(\delta_n - b_n )^2}{2 \widetilde{V}_\mathscr{W} \delta_n^2}; \label{eq:spChc2} \\
\limsup_{n\to+\infty} \frac{s_n^\star}{\delta_n} &\leq \frac{1}{V_\mathscr{W}}, \label{eq:spChc3}
\end{align}
where 
\begin{align}
s_n^\star := \argmax_{s\geq 0} \left\{ E_\textnormal{h}(s,P_n, P^\star\mathscr{W}) - s \left( C_\mathscr{W} - \delta_n  \right)
\right\}.
\end{align}
\end{prop}
\noindent The proof of Proposition~\ref{prop:spCh} is provided in Appendix \ref{proof:spCh}.

\end{proof}

\section{Moderate Deviations for Quantum Hypothesis Testing} \label{sec:HT}
In this section, we show that a special case of channel coding yields the moderate deviation result for quantum hypothesis testing. The achievability part is given in Theorem \ref{theo:achievability_HT}. In Section \ref{ssec:achievability_HT}, we provide two proofs. The first proof follows the idea of asymptotic expansions in Theorem \ref{theo:achievability}; however, we will employ Audenaet \textit{et al.}'s quantum Hoeffding bound \cite{ANS+08}, instead of Hayashi's inequality \cite{Hay07}. The second proof relies on a martingale inequality \cite{Sas12}. The converse part and its proof are provided in Theorem \ref{theo:converse_HT} and Section \ref{ssec:converse_HT}, respectively.

\begin{theo}[Achievability] \label{theo:achievability_HT}
Let $\rho,\sigma \in\mathcal{S(H)}$ be the density operators with finite relative variance $V:=V\left( \rho\|\sigma \right)>0$. For any sequence of real numbers $(a_n)_{n\in\mathbb{N}}$ satisfying Eq.~\eqref{eq:cond},	there exists a sequence $r_n:= D\left(\rho\|\sigma\right) - a_n$ such that
\begin{align}
\limsup_{n\to+\infty} \frac{1}{na_n^2} \log \widehat{\alpha}_{\exp\{-n r_n \}}\left( \rho^{\otimes n} \| \sigma^{\otimes n} \right) \leq -\frac{1}{2V}.
\end{align}
\end{theo}

\begin{theo}[Converse] \label{theo:converse_HT}
Let $\rho,\sigma \in\mathcal{S(H)}$ be the density operators with non-zero and finite relative variance $V:=V\left( \rho\|\sigma \right)>0$. For any sequence of real numbers $\{a_n\}_{n\in\mathbb{N}}$ satisfying	Eq.~\eqref{eq:cond},	there exists a sequence $r_n:= D\left(\rho\|\sigma\right) - a_n$ such that
\begin{align}
\liminf_{n\to+\infty} \frac{1}{na_n^2} \log \widehat{\alpha}_{\exp\{-n r_n \}}\left( \rho^{\otimes n} \| \sigma^{\otimes n} \right) \geq -\frac{1}{2V}.
\end{align}
\end{theo}

\subsection{Proof of Achievability: Theorem \ref{theo:achievability_HT}} \label{ssec:achievability_HT}
In this section, we present two proofs for Theorem \ref{theo:achievability_HT}.
The first one relies on the quantum Hoeffding bound \cite{ANS+08} and the Taylor's expansion of the exponent function $E_\text{h}$.

\begin{proof}[The first proof of Theorem \ref{theo:achievability_HT}]

Recall the following achievability of the quantum Hoeffding bound:
\begin{lemm}[Theorem 5, Section 5.5 of \cite{ANS+08}] \label{lemm:Audenaert}
Let $\rho,\sigma \in \mathcal{S(H)}$. For any $r\geq 0$ and any $n\in\mathbb{N}$, we have	
\begin{align} \label{eq:Audenaert}
\widehat{\alpha}_{\exp\{ -nr \}} \left( \rho^{\otimes n} \| \sigma^{\otimes n} \right)
\leq \exp\left\{  -n \left[ \sup_{ 0<\alpha\leq 1 } \left\{ \frac{\alpha-1}{\alpha} \left( r - D_\alpha\left( \rho\|\sigma \right) \right) \right\} \right]  \right\}.
\end{align}
\end{lemm}

Since $ D(\rho\|\sigma ) > 0$ (due to Eq.~\eqref{eq:positive_D}), we have
\begin{align}
r_n:=D(\rho\|\sigma) - a_n >0
\end{align}
for all sufficiently large $n$. Choose such $n$ onwards, then Eq.~\eqref{eq:Audenaert} implies that:
\begin{align}
\frac{1}{n a_n^2} \log \widehat{\alpha}_{\exp\{ -nr_n \}}\left( \rho^{\otimes n} \| \sigma^{\otimes n } \right) 
&\leq - \frac{1}{a_n^2} \sup_{ 0<\alpha\leq 1 } \left\{ \frac{\alpha-1}{\alpha} \left( r_n - D_\alpha\left( \rho\|\sigma \right) \right) \right\} \\
&= - \frac{1}{a_n^2} \sup_{ s\geq 0 } \left\{ E_\text{h}(s) - s r_n  \right\}, \label{eq:H4}
\end{align}
where we substitute $s = \frac{1-\alpha}{\alpha}$ and let 
\begin{align} 
E_\text{h}(s) :=  
s  D_{\frac{1}{1+s}} \left( \rho\|\sigma \right).
\end{align}

Taylor's theorem followed by simple calculation yields
\begin{align} \label{eq:H1}
E_\text{h}(s) = s D(\rho\|\sigma) - \frac{s^2}{2} V + \left.\frac{s^3}{6} \frac{\partial^3 E_\text{h}(s)}{\partial s^3}\right|_{s = \bar{s}}
\end{align}
for some $\bar{s}\in[0,s]$ and all $s\geq 0$. The above equation is also a simple consequence of items \ref{prop_h-c} and \ref{prop_h-e} in Proposition \ref{prop:prop_h}.
Now let $s_n = a_n/V$, for all $n\in\mathbb{N}$.
Then for all sufficiently large $n$ and for some $\bar{s}_n \in [0, s_n]$,
Eq.~\eqref{eq:H1} yields 
\begin{align}
\sup_{s\geq 0}\left\{ E_\text{h}(s) - sr_n  \right\} &\geq  E_h(s_n) - s_n r_n \\
&= \frac{a_n}{V}\left( D(\rho\|\sigma) - r_n \right) - \frac{a_n^2}{2V} + \left.\frac{a_n^3}{6V^3} \frac{\partial^3 E_\text{h}(s)}{\partial s^3}\right|_{s = \bar{s}_n} \\
&= \frac{a_n^2}{2V} + \left.\frac{a_n^3}{6V^3} \frac{\partial^3 E_\text{h}(s)}{\partial s^3}\right|_{s = \bar{s}_n}, \label{eq:H2}
\end{align}
where we substitute $r_n = D(\rho\|\sigma ) - a_n$ in Eq.~\eqref{eq:H2}.

Define
\begin{align} \label{eq:H5}
\Upsilon := \max_{s\in[0,1]} \left| \frac{\partial^3 E_\text{h}(s)}{\partial s^3}\right|,
\end{align}
which is finite.
Therefore, Eq.~\eqref{eq:H2} leads to
\begin{align}
\sup_{s\geq 0}\left\{ E_\text{h}(s) - sr_n  \right\} 
&\geq  \frac{a_n^2}{2V} + \left.\frac{a_n^3}{6V^3} \frac{\partial^3 E_\text{h}(s)}{\partial s^3}\right|_{s = \bar{s}_n} \\
&\geq \frac{a_n^2}{2V} - \frac{a_n^3}{6V^3} \Upsilon \label{eq:H3}
\end{align}
for all sufficiently large $n$. Substituting Eq.~\eqref{eq:H3} into Eq.~\eqref{eq:H4} yields
\begin{align}
\frac{1}{n a_n^2} \log \widehat{\alpha}_{\exp\left\{-n r_n\right\}}\left(\rho\|\sigma\right) \leq  - \frac{1}{2V} \left( 1 - \Upsilon \frac{a_n}{3 V^2} \right),
\end{align}
which implies the desired achievability part:
\begin{align}
\limsup_{n\to+\infty} \frac{1}{n a_n^2} \log \widehat{\alpha}_{\exp\left\{-n r_n\right\}}\left(\rho\|\sigma\right) \leq -\frac{1}{2V}.
\end{align}

\end{proof}

In the following, we give an alternative proof of Theorem \ref{theo:achievability_HT} by employing a martingale inequality \cite{Sas12}.

\begin{proof}[The second proof of Theorem \ref{theo:achievability_HT}]
We follow the idea in Ref.~\cite{Li14} to write the eigendecomposition of $\rho^{\otimes n}$ and $\sigma^{\otimes n}$, respectively, as
\begin{align} \label{eq:eigen}
\rho^{\otimes n} = \sum_{x^n} \lambda^n(x^n) \proj{f_{x^n}^n}; \quad \sigma^{\otimes n} = \sum_{y^n} \gamma^n(y^n) \proj{g_{y^n}^n}, 
\end{align}
where $x^n := x_1x_2\ldots x_n$; $y^n:=y_1y_2\ldots y_n$; $\lambda^n(x^n) = \prod_{i=1}^n \lambda(x_i)$; $\mu^n (y^n) = \prod_{i=1}^n \mu(y_i)$; $|f_{x^n}^n\rangle = |f_{x_1}\rangle \otimes |f_{x_2}\rangle \otimes \cdots \otimes |f_{x_n}\rangle$; and $|g_{y^n}^n\rangle = |g_{y_1}\rangle \otimes |g_{y_2}\rangle \otimes \cdots \otimes |g_{y_n}\rangle$. 
Further, we define a pair of random variables $(X,Y)$ via the Nussbaum-Szko{\l}a mapping \cite{NS09}, i.e.~$P_{X,Y}(x,y) = \lambda(x)|\gamma_{xy}|^2$, where $\gamma_{xy} := \langle g_y | f_x \rangle \in\mathbb{C}$. It is well-known that
\begin{align}
D(\rho\|\sigma) &= 
D(\lambda(X)\|\mu(Y)) = 
\mathbb{E}_{(X,Y)}\left[ \log \frac{\lambda(X)}{\mu(Y)}\right], \label{eq:NSD}\\
V(\rho\|\sigma) &= 
V(\lambda(X)\|\mu(Y)) =
\mathrm{Var}_{(X,Y)}\left[ \log \frac{\lambda(X)}{\mu(Y)}\right]. \label{eq:NSV}
\end{align}
Let $T_n := \exp\left\{ n r_n \right\}$.
For every sequence $x^n$, we define a sub-normalized vector:	 
\begin{align}
|\xi_{x_n}^n\rangle := \sum_{y^n: \lambda^n(x^n)/\mu^n(y^n) \geq T_n} \gamma_{x^n y^n}^n |g_{y^n}^n\rangle
\end{align}
with $\gamma_{x^n y^n}^n = \prod_{i=1}^n \gamma_{x_iy_i}$ and $\sum_{x} |\gamma_{xy}|^2 = \sum_{y} |\gamma_{xy}|^2 = 1$.
Applying the Gram-Schmidt orthonormalization process on $\left\{ |\xi_{x^n}^n\rangle \right\}_{x^n}$ to obtain an orthonormal vectors
\begin{align} \label{eq:HT1}
|\hat{\xi}_{x^n}^n\rangle = \sum_{y^n: \lambda^n(x^n)/\mu^n(y^n) \geq T_n} t_{x^n y^n}^n |g_{y^n}^n\rangle
\end{align}
for some $t_{x^n y^n}^n\in\mathbb{C}$ and 
\begin{align} \label{eq:HT2}
\sum_{y^n: \lambda^n(x^n)/\mu^n(y^n) \geq T^n} |t_{x^n y^n}^n|^2 = 1.
\end{align}
We define a test of the hypotheses by
\begin{align}
Q_n := \sum_{x^n} |\hat{\xi}_{x^n}^n\rangle \langle \hat{\xi}_{x^n}^n|.
\end{align}
Then, it suffices to show $\beta\left(Q_n; \sigma^{\otimes n}\right) \leq \exp\{-n r_n\}$ and 
\begin{align}
\lim_{n\to+\infty}\frac{1}{na_n^2} \log 	\alpha\left( Q_n; \rho^{\otimes n} \right) \leq - \frac{1}{2V}
\end{align}
to complete the proof. The former follows Eqs.~\eqref{eq:eigen}, \eqref{eq:HT1}, and \eqref{eq:HT2}:
\begin{align}
\beta\left(Q_n; \sigma^{\otimes n}\right) 
&= \sum_{x^n} \Tr\left[ \sigma^{\otimes n}  |\hat{\xi}_{x^n}^n\rangle \langle \hat{\xi}_{x^n}^n| \right] \notag\\
&= \sum_{x^n} \sum_{y^n: \lambda^n(x^n)/\mu^n(y^n) \geq T^n} |t_{x^n y^n}^n|^2 \mu^n(y^n) \notag\\
&\leq \sum_{x^n} \frac{\lambda^n(x^n)}{T_n}=  \frac{1}{T_n} = \exp\{-n r_n\}.
\end{align}	
Likewise, since $\frac{|\xi_{x^n}\rangle\langle \xi_{x^n}|}{|\langle \xi_{x^n}|\xi_{x^n}\rangle|^2  }\leq Q_n$, one can verify that
\begin{align}
\alpha\left(Q_n; \rho^{\otimes n}\right) &\leq 
1-\sum_{x^n} \lambda^n(x^n) \langle \xi_{x^n}^n| \xi_{x^n}^n \rangle\\
&= \Pr\left\{ \frac{\lambda^n(X^n)}{\mu^n(Y^n)} < T_n \right\} \\
&= \Pr\left\{ \log \frac{\lambda^n(X^n)}{\mu^n(Y^n)} < n r_n \right\}. \label{eq:HT3}
\end{align}

Next, we adopt Sason's approach \cite{Sas12} to construct a martingale sequence $\left\{U_k, \mathfrak{M}_k  \right\}_{k=0}^n$, where	$\mathfrak{M}_k$ denotes the sigma-algebra formed by $(X_l,Y_l)_{l=1}^k$;	$\mathfrak{M}_0\subseteq \mathfrak{M}_1 \subseteq \ldots \subseteq \mathfrak{M}_n$ is the filtration; and 
\begin{align}
U_k &:= \mathbb{E}_{(X^n,Y^n)} \left[ \left.\log \frac{\lambda^n(X^n)}{\mu^n(Y^n)}\right| \mathfrak{M}_k \right]\\
&= \sum_{i=1}^k \log \frac{\lambda(X_i)}{\mu(Y_i)} + \sum_{i=k+1}^n \mathbb{E}_{X^n}\left[ \log \frac{\lambda(X_i)}{\mu(Y_i)}\right] \\
&= \sum_{i=1}^k \log \frac{\lambda(X_i)}{\mu(Y_i)} + (n-k) D(\lambda(X)\|\mu(Y)).
\end{align}
In particular, we have
\begin{align}
U_0 = n D\left(\lambda(X)\|\mu(Y)\right);\,
U_n = \log \frac{\lambda(X^n)}{\mu(Y^n)} = \sum_{i=1}^n\log \frac{\lambda(X_i)}{\mu(Y_i)}. \notag
\end{align}
Hence, it can be verified that:
\begin{align*}
&U_k - U_{k-1} = \log \frac{\lambda(X_k)}{\mu(Y_k)} - D(\lambda(X)\|\mu(Y)); \\
&\mathbb{E}_{X^n} \left[ \left.U_k - U_{k-1} \right| \mathfrak{M}_{k-1} \right] = 0;\\
&\mathbb{E}_{X^n} \left[ \left.\left(U_k - U_{k-1}\right)^2 \right| \mathfrak{M}_{k-1} \right] 
= V\left( \lambda(X) \| \mu(Y) \right) = V.
\end{align*}
Let
\begin{align}
b := \max_{(x,y):x=y} \left| \log \frac{\lambda(x)}{\mu(y)} - D(\lambda(X)\|\mu(Y))\right|,
\end{align}
which is a finite number due to the assumption of the finite-dimensional Hilbert space.
	Then, we have $| U_k - U_{k-1} | \leq b$ almost surely for every $k\in[n]$.
	Equipped with the notation above, Eq.~\eqref{eq:HT3} can be expressed as:
	\begin{align}
	\alpha\left(Q_n;\rho^{\otimes n}\right) 
	&= \Pr\left\{  U_n - U_0 \leq -n a_n \right\}. \label{eq:HT4}
	\end{align}
	
	In the following, we borrow the idea from Sason \cite{Sas12} to employ a martingale inequality to upper bound Eq.~\eqref{eq:HT4}.
	

	\begin{theo}
		[Refined Azuma's Inequality {\cite[Theorem 2]{Sas12}}] \label{theo:Bennet}
		Let $\left(X_k\right)_{k=1}^n$ be a martingale with respect to the filtration $\left( \mathfrak{M}_k \right)_{k=0}^n$ such that the following requirements are satisfied almost surely:
		\textnormal{(i)} $\mathbb{E}\left[ X_k | \mathfrak{M}_{k-1} \right] = 0$;
		\textnormal{(ii)} $\mathbb{E}\left[ X_k^2 | \mathfrak{M}_{k-1} \right] \leq v$;
		\textnormal{(iii)} $\|X_k\|_{\infty} \leq b_k$.
		For any $x\geq0$,
		\begin{align} 
		&\Pr\left\{ \sum_{k=1}^n X_k \geq x n \right\} 
		= \Pr\left\{ \sum_{k=1}^n X_k \leq -x n \right\} 
		\notag  \\
		&\leq 2\exp\left\{ -n h\left(\left.\frac{bx+v}{b^2+v} \right\| \frac{v}{b^2+v} \right)\right\}, \label{eq:Bennett}		
		\end{align}
		where 	
		$h(p\|q) := p\log\frac{p}{q} + (1-p)\log \frac{1-p}{1-q}$.
	\end{theo}

Apply Theorem \ref{theo:Bennet} to Eq.~\eqref{eq:HT4} with $x = a_n$, $X_k = U_k-U_{k-1}$ for ever $k\in[n]$:
\begin{align}
\alpha\left(Q_n; \rho^{\otimes n}\right)  
\leq 2\exp \left\{ -n h\left(\left.\frac{ba_n+V}{b^2+V} \right\| \frac{V}{b^2+V} \right) \right\}. \label{eq:NPa3}
\end{align} 
By  using a scalar inequality \cite[Lemma 1]{Sas12}:
\begin{align}
(1+u) \log (1+u) \geq u + \frac{u^2}{2} - \frac{u^3}{6}, \quad u\geq 0,
\end{align}
and the definition of $h(\cdot\|\cdot)$ in Theorem \ref{theo:Bennet},
Eq.~\eqref{eq:NPa3} leads to
\begin{align}
\alpha\left(Q_n; \rho^{\otimes n}\right)   
&\leq 2\exp\left\{ -n \left[ \frac{a_n^2 }{2V } \left( 1 - \frac{a_n b}{3 V(1+V/b^2)} \right) \right]
\right\}. \label{eq:NPa4}
\end{align}
Finally, recall that $\lim_{n\to+\infty} a_n = 0$ in Eq.~\eqref{eq:cond}, then 
\begin{align*}
\limsup_{n\to+\infty} \frac{1}{n a_n^2} \log {\alpha}_{n}\left( \eta_{n} \right) \leq -\frac{1}{2V}.
\end{align*}
	
\end{proof}

\subsection{Proof of Converse: Theorem \ref{theo:converse_HT}} \label{ssec:converse_HT}

The converse part is a direct consequence of  the sharp converse Hoeffding bound, Theorem \ref{prop:strong}.

Let $\mathcal{X} = \{x\}$ and $W_x = \rho$. We apply Theorem \ref{prop:strong} with $r= r_n$ to obtain
\begin{align}
\widehat{\alpha}_{\exp\left\{-n r_n\right\}} \left( \rho^{\otimes n} \| \sigma^{\otimes n} \right) 
&\geq 
\frac{A}{s_n^\star\sqrt{n}}	\exp\left\{ 	-n  \left[ \sup_{0<\alpha\leq 1} \frac{\alpha-1}{\alpha} \left(r_n - c_n - D_\alpha\left(\rho\|\sigma\right)\right) \right] \right\},\label{eq:converse_HT1}
\end{align}
for sufficiently large $n\in\mathbb{N}$ and some constant $A>0$. Here
\begin{align}
s_n^\star := \argmax_{s\geq 0} \left\{   s D_{\frac{1}{1+s}}(\rho\|\sigma) - s r_n  \right\}.
\end{align}
Now let 
\begin{align}
\delta_n := a_n + c_n, \quad \forall n\in\mathbb{N},
\end{align}
and invoke Proposition \ref{prop:spCh} with $W_x = \rho$, $P(x) = 1$, and substitute $P^\star \mathscr{W}$ with $\sigma$ to obtain
\begin{align}
\limsup_{n\to+\infty} \frac{ \sup_{s\geq 0} \left\{  - s \left( D\left(\rho\|\sigma\right) - \delta_n\right) + sD_{\frac{1}{1+s}}(\rho\|\sigma) \right\} }{\delta_n^2} \leq \frac{1}{2V}. \label{eq:converse_HT6}
\end{align}
Moreover, Eq.~\eqref{eq:spChc3} in Proposition \ref{prop:spCh} gives that $\lim_{n\to+\infty}\frac{s_n^\star}{\delta_n} = 1/V$.
Combining Eqs.~\eqref{eq:converse_HT1} and \eqref{eq:converse_HT6} concludes our claim:
\begin{align}
\liminf_{n\to+\infty} \frac{ \log \widehat{\alpha}_{\exp\left\{-nr_n\right\}} \left( \rho^{\otimes n} \| \sigma^{\otimes n} \right)  }{n \delta_n^2}
&\geq -\frac{1}{2V}.
\end{align}

\section{Conclusion} \label{sec:conclutions}
A practical question in quantum information theory is that---is it possible for a reliable communication through a c-q channel when the transmission rate approaches capacity in blocklength?
In this paper, we propose a moderate deviation analysis for c-q channel and thus give an affirmative answer.
Moreover, we also establish the moderate deviations for quantum hypothesis testing.

Our proof strategy is based on a strong large deviation theory \cite{CHT16bb, CHT16b} and the study of the asymptotic behaviour of the error exponent function. 
As a result, we successfully bridge the connection between small error regime and the medium error regime.
On the other hand, the recent work from the authors \cite{CCT+16b} also obtains the moderate deviation result via the techniques in the non-vanishing error regime.
It is remarkable that both methods from different regimes arrive at the same place, and hence both this work along with Ref.~\cite{CCT+16b} illuminate the whole picture of the three regimes in quantum information theory.
\qed

\section*{Acknowledgements}
MH is supported by an ARC Future Fellowship under Grant FT140100574.  We would like to thank Vincent Tan for introducing us to Altu\u{g} and Wagner's work of moderate deviation analysis. We also thank Christopher Chubb and Marco Tomamichel for the insightful discussions and the useful comments.

\appendix

\section{Properties of Auxiliary Functions} \label{app:aux}
This section contains proofs of Propositions \ref{prop:prop2} and \ref{prop:prop_h}. Most results follow from properties of Petz quantum R\'enyi divergence \cite{Pet86} (see also \cite{MO15, LT15, Tom16}).

\subsection{Proof of Proposition \ref{prop:prop2}} \label{app:prop2}	
	
\begin{prop4}[Properties of $\widetilde{E}_0(s,P,\sigma)$]
	For any classical-quantum channel $\mathscr{W}:\mathcal{X}\to\mathcal{S(H)}$, the modified auxiliary function $\widetilde{E}_0(s,P,\sigma)$ admits the following properties.
	\begin{enumerate} [(a)]
		\item\label{prop2-aa}
		$\widetilde{E}_0(s,P,\sigma)$ and its partial derivatives $\partial \widetilde{E}_0(s,P,\sigma)/\partial s$, $\partial^2 \widetilde{E}_0(s,P,\sigma)/\partial s^2$, $\partial^3 \widetilde{E}_0(s,P,\sigma)/\partial s^3$ are all continuous in $(s,P)\in \mathbb{R}_{\geq 0} \times \mathscr{P}(\mathcal{X})$.
		
		\item\label{prop2-bb}
		For every $P\in\mathscr{P}(\mathcal{X})$, the function $\widetilde{E}_0(s,P,\sigma)$ is concave in $s\in\mathbb{R}_{\geq 0}$.
		
		\item\label{prop2-cc}
		For every $P\in\mathscr{P}(\mathcal{X})$,
		\begin{align} 
		\left.\frac{ \partial \widetilde{E}_0(s,P,\sigma)}{\partial s}\right|_{s = 0} =  D(P\circ\mathscr{W}\|P\otimes\sigma)
		\end{align}
		
		\item\label{prop2-dd}
		For every $P\in\mathscr{P}(\mathcal{X})$,
		\begin{align} 
		\lim_{s\to+\infty} \frac{ \partial \widetilde{E}_0(s,P)}{\partial s}\leq 
		\frac{ \partial \widetilde{E}_0(s,P)}{\partial s} \leq D(P\circ\mathscr{W}\|P\otimes\sigma), \; \forall s\in\mathbb{R}_{\geq 0}. \label{eq:tE0_I11}
		\end{align} 
		
		\item\label{prop2-ee}
		For every $P\in\mathscr{P}(\mathcal{X})$,
		\begin{align}
		\left.\frac{ \partial^2 \widetilde{E}_0(s,P)}{\partial s^2}\right|_{s = 0} = -V(P\circ\mathscr{W}\|P\otimes\sigma).
		\end{align}
	\end{enumerate}
\end{prop4}	

\begin{proof}[Proof of Proposition~\ref{prop:prop2}]	
	~\\
\begin{itemize}
\item[(\ref{prop:prop2}-\ref{prop2-aa})] The continuity can be proved by the standard approach of functional calculus (see e.g.~\cite[Lemma III.1]{MO15} and \cite[Section 4.2]{LT15}). Let $\widetilde{F}(s) := \sum_{x\in\mathcal{X}} P(x) \Tr\left[ W_x^{1-s} (\sigma)^s
\right]$. Direct calculation shows that 
\begin{align}
\frac{\partial \widetilde{E}_0(s,P,\sigma)}{\partial s} &= -\frac{\widetilde{F}'(s)}{\widetilde{F}(s)}, \label{eq:tE01} \\
\frac{\partial^2 \widetilde{E}_0(s,P,\sigma)}{\partial s^2} &= -\frac{\widetilde{F}''(s)}{\widetilde{F}(s)} + \left(  \frac{\partial \widetilde{E}_0(s,P,\sigma)}{\partial s} \right)^2, \label{eq:tE02} \\
\frac{ \partial^3 \widetilde{E}_0(s,P,\sigma)}{\partial s^3} &= -\frac{\widetilde{F}'''(s,P)}{\widetilde{F}(s,P)} + 3 \frac{ \partial \widetilde{E}_0(s,P,\sigma)}{\partial s} \frac{ \partial^2 \widetilde{E}_0(s,P,\sigma)}{\partial s^2} - \left( \frac{ \partial \widetilde{E}_0(s,P,\sigma)}{\partial s} \right)^3, \label{eq:tE03}		
\end{align}
and
\begin{align}
\widetilde{F}'(s) &= \sum_{x\in\mathcal{X}} P(x) \Tr \left[ - W_x^{1-s}  {\log} (W_x) (\sigma)^s + W_x^{1-s}   (\sigma)^s  {\log} (\sigma) \right], \label{eq:tF1}\\
\begin{split} \label{eq:tF2}
\widetilde{F}''(s) &= \sum_{x\in\mathcal{X}} P(x) \Tr \left[ W_x^{1-s}   {\log}^2 (W_x)  (\sigma)^s  -W_x^{1-s}  {\log} (W_x)  (\sigma)^s  {\log} (\sigma) \right.\\ &
\left. \quad - W_x^{1-s}  {\log} (W_x) (\sigma)^s  {\log}  (\sigma) + W_x^{1-s}   (\sigma)^s  {\log}^2 (\sigma )
\right],
\end{split} \\
\begin{split} \label{eq:tF3}
\widetilde{F}'''(s) &= \sum_{x\in\mathcal{X}} P(x) \Tr \left[ -W_x^{1-s}  {\log}^3 (W_x)     (\sigma)^s + W_x^{1-s}  {\log}^2 (W_x) (\sigma)^s  {\log} (\sigma) \right.\\ 
& \left. \quad +2 W_x^{1-s}  {\log}^2 (W_x)  (\sigma)^s  {\log} (\sigma)
-2W_x^{1-s}  {\log} (W_x) (\sigma)^s  {\log}^2 (\sigma)
\right.\\ &
\left.\quad -W_x^{1-s}  {\log} (W_x)  (\sigma)^s  {\log}^2 (\sigma) + W_x^{1-s}  (\sigma)^s {\log}^3 (\sigma )
\right].
\end{split}
\end{align}
Since the matrix power function is continuous (with respect to the strong topology; see e.g.~\cite[Theorem 1.19]{Hig08}), we conclude the continuity of the partial derivatives Eqs.~\eqref{eq:tE01}-\eqref{eq:tE03} in item \ref{prop2-aa}.		
		
\item[(\ref{prop:prop2}-\ref{prop2-bb})] 
The claim follows from the concavity of the map $s\mapsto sD_{1-s}(\,\cdot\,\|\,\cdot\,)$ (see e.g.~\cite[Lemma III.11]{MO14b}).
		
\item[(\ref{prop:prop2}-\ref{prop2-cc})] The results can be derived from evaluating Eqs.~\eqref{eq:tE01}, \eqref{eq:tE02}, \eqref{eq:tF1}, and \eqref{eq:tF2} at $s=0$.  We provide an alternative proof here. One can verify
\begin{align}
\left.\frac{\partial \widetilde{E}_0(s,P,\sigma)}{\partial s}\right|_{s=0} &=  \left.D_{1-s} \left( P\circ \mathscr{W} \| P \otimes \sigma \right) - s  D_{1-s}'\left( P\circ \mathscr{W}  \| P \otimes \sigma \right)\right|_{s=0} \\
&= \left. D_{1-s}\left( P\circ \mathscr{W}  \| P \otimes \sigma \right) \right|_{s=0} \\
&= D(P\circ\mathscr{W}\|P\otimes \sigma). \label{eq:tE_I}
\end{align}
		
\item[(\ref{prop:prop2}-\ref{prop2-dd})] The concavity of the map $s\mapsto \widetilde{E}(s,P,\sigma)$ in item \ref{prop2-bb} ensures that  $\partial \widetilde{E}(s,P,\sigma)/\partial s$ is non-increasing in $s$. Along with Eq.~\eqref{eq:tE_I}, we conclude Eq.~\eqref{eq:tE0_I1}.

\item[(\ref{prop:prop2}-\ref{prop2-ee})] Following from item \ref{prop2-cc}, one obtain
\begin{align}
\left.\frac{\partial^2 \widetilde{E}_0(s,P,\sigma)}{\partial s^2}\right|_{s=0} &= \left. -2 D_{1-s}' \left( P\circ \mathscr{W}  \| P \otimes \sigma \right) + s D_{1-s}''\left( P\circ \mathscr{W}  \| P \otimes \sigma \right)\right|_{s=0} \\
&= \left. -2 D_{1-s}' \left( P\circ \mathscr{W}  \| P \otimes \sigma \right)\right|_{s=0} \\
&= - V(P\circ\mathscr{W}\|P\otimes \sigma), \label{eq:tE_V2}
\end{align}
where the last equality \eqref{eq:tE_V2} follows from the fact $D_{1/1+s}'(\cdot\|\cdot)|_{s=0} = V (\cdot\|\cdot)/2$ \cite[Theorem 2]{LT15}.

	\end{itemize}
	\end{proof}	

\subsection{Proof of Proposition \ref{prop:prop_h}} \label{app:prop_h}
\begin{prop5}[Properties of $E_\text{h}(s,P,\sigma)$] 
	Consider a classical-quantum channel $\mathscr{W}:\mathcal{X}\to \mathcal{S(H)}$, a distribution $P\in\mathscr{P}(\mathcal{X})$, and a state $\sigma\in\mathcal{S(H)}$ with $W_x\ll \sigma$ for all $x\in \textnormal{\texttt{supp}}(P)$. Then $E_\textnormal{h}(s,P,\sigma)$ defined in Eq.~\eqref{eq:Eh} enjoys the following properties.
	\begin{enumerate}[(a)]
		\item\label{prop_h-aa}
		The partial derivatives $\partial E_\textnormal{h}(s,P,\sigma)/\partial s$, $\partial^2 E_\textnormal{h}(s,P,\sigma)/\partial s^2$, $\partial^3 E_\textnormal{h}(s,P,\sigma)/\partial s^3$, and $E_\textnormal{h}(s,P)$ are all continuous for $(s,P)\in \mathbb{R}_{\geq 0} \times \mathscr{P}(\mathcal{X})$.
		
		\item\label{prop_h-bb}
		For every $P\in\mathscr{P}(\mathcal{X})$, the function $E_\textnormal{h}(s,P,\sigma)$ is concave in $s$ for all $s\in\mathbb{R}_{\geq 0}$.
		
		\item\label{prop_h-cc}
		For every $P\in\mathscr{P}(\mathcal{X})$,
		\begin{align} \label{eq:first_hh}
		\left.\frac{ \partial E_\textnormal{h}(s,P,\sigma)}{\partial s}\right|_{s = 0} = D\left(\mathscr{W}\|\sigma|P\right).
		\end{align}
		
		\item\label{prop_h-dd}
		For every $P\in\mathscr{P}(\mathcal{X})$,
		\begin{align} \label{eq:E_h_II}
		\lim_{s\to+\infty} \frac{ \partial E_\textnormal{h}(s,P,\sigma)}{\partial s}\leq 
		\frac{ \partial E_\textnormal{h}(s,P,\sigma)}{\partial s} \leq D\left(\mathscr{W}\|\sigma|P\right), \; \forall s\in\mathbb{R}_{\geq 0}.  
		\end{align}
		
		\item\label{prop_h-ee}
		For every $P\in\mathscr{P}(\mathcal{X})$,
		\begin{align} \label{eq:second_hh}
		\left.\frac{ \partial^2 E_\textnormal{h}(s,P,\sigma)}{\partial s^2}\right|_{s = 0} = -V\left( \mathscr{W}\|\sigma|P\right).
		\end{align}
	\end{enumerate}
\end{prop5}
\begin{proof}[Proof Proposition~\ref{prop:prop_h}]
	~\\
\begin{itemize}
\item[(\ref{prop:prop_h}-\ref{prop_h-aa})]		
Direct calculation yields that
\begin{align}
\frac{\partial E_\text{h} (s,P,\sigma)}{\partial s} &= 
D_{\frac{1}{1+s}} \left( \mathscr{W} \| \sigma |P \right) - \frac{s}{(1+s)^2} D_{\frac{1}{1+s}}'\left( \mathscr{W}\|\sigma|P \right) \label{eq:tE_1h} \\
\frac{\partial^2 E_\text{h}(s,P,\sigma) }{\partial s^2}&= 
-\frac{2}{(1+s)^3} D_{\frac{1}{1+s}}' \left(\mathscr{W} \| \sigma |P \right) + \frac{s}{(1+s)^4} D_{\frac{1}{1+s}}'' \left( \mathscr{W}\| \sigma |P \right) \label{eq:tE_2h} \\
\frac{\partial^3 E_\text{h}(s,P,\sigma) }{\partial s^3}&= 
\frac{6}{(1+s)^4} D_{\frac{1}{1+s}}' \left( \mathscr{W} \| \sigma |P \right) + \frac{3-3s}{(1+s)^5} D_{\frac{1}{1+s}}'' \left( \mathscr{W} \| \sigma |P \right) \notag\\
&\quad - \frac{s}{(1+s)^6} D_{\frac{1}{1+s}}''' \left( \mathscr{W} \| \sigma |P \right).
\label{eq:tE_3h}
\end{align}
From Eqs.~\eqref{eq:tE_1h}-\eqref{eq:tE_3h} and the fact that $D_{1/(1+s)}\left(\mathscr{W}\|\sigma|P\right)$, $D_{1/(1+s)}'\left(\mathscr{W}\|\sigma|P\right)$, $D_{1/(1+s)}''\left(\mathscr{W}\|\sigma|P\right)$, and $D_{1/(1+s)}'''\left(\mathscr{W}\|\sigma|P\right)$ are continuous for $(s,P)\in \mathbb{R}_{\geq 0}\times \mathscr{P}(\mathcal{X})$, we deduce the continuity property in item \ref{prop_h-aa}.

\item[(\ref{prop:prop_h}-\ref{prop_h-bb})]	
The proof strategy follows closely with \cite[Appendix B]{MO14b}. Let $\psi(\alpha) = \sum_{x\in\mathcal{X}} P(x) \log \Tr\left[ W_x^{\alpha} \sigma^{1-\alpha} \right]$.
Since $\alpha \mapsto \psi(\alpha)$ is convex for all $\alpha\in(0,1]$ \cite[Lemma III.11]{MO14b}, it can be written as the supremum of affine functions, i.e.~
\begin{align}
\psi(\alpha)  = \sup_{i\in\mathcal{I}} \left\{ c_i \alpha + d_i  \right\}
\end{align}
for some index set $\mathcal{I}$.
Hence,
\begin{align} \label{eq:E_h_convex}
- E_\text{h}(s,P,\sigma) = (1+s) \psi\left( \frac{1}{1+s} \right)  = \sup_{i\in\mathcal{I}} \left\{ c_i + d_i (1+s)\right\}.
\end{align}
The right-hand side of Eq.~\eqref{eq:E_h_convex}, in turn, implies that the map $s\mapsto E_\text{h}(s,P,\sigma)$ is convex for all $s\in\mathbb{R}_{\geq 0}$.
		
\item[(\ref{prop:prop_h}-\ref{prop_h-cc})]	 From Eqs.~\eqref{eq:tE_1h} and \eqref{eq:tE_2h}, one finds
\begin{align}
\left.\frac{\partial E_\text{h} (s,P,\sigma)}{\partial s}\right|_{s=0} 
		&= D\left(\mathscr{W}\|\sigma | P\right). \label{eq:tE_I_h}
\end{align}
		
\item[(\ref{prop:prop_h}-\ref{prop_h-dd})]	 The concavity of the map $s\mapsto E_\text{h}(s,P,\sigma)$ in item \ref{prop_h-bb} ensures that  $\partial {E}_\text{h}(s,P,\sigma)/\partial s$ is non-increasing in $s$. Along with Eq.~\eqref{eq:tE_I_h} in item \ref{prop_h-c}, we conclude Eq.~\eqref{eq:E_h_I}.

\item[(\ref{prop:prop_h}-\ref{prop_h-ee})]	Applying $D_{1/1+s}'(\cdot\|\cdot)|_{s=0} = V (\cdot\|\cdot)/2$ \cite[Theorem 2]{LT15}, it holds that 
\begin{align}
\left.\frac{\partial^2 E_\text{h}(s,P,\sigma) }{\partial s^2}\right|_{s=0} 
&= - V\left(\mathscr{W}\|\sigma|P\right). \label{eq:tE_V2_h}
\end{align}
\end{itemize}	
\end{proof}

\subsection{Proof of Proposition \ref{prop:prop_b}} \label{app:prop_b}
\begin{prop7}[Properties of $\widetilde{E}_\text{h}(s,P,\sigma)$]
	Consider a classical-quantum channel $\mathscr{W}:\mathcal{X}\to \mathcal{S(H)}$, a distribution $P\in\mathscr{P}(\mathcal{X})$, and a state $\sigma\in\mathcal{S(H)}$ with $W_x\ll \sigma$ for all $x\in \textnormal{\texttt{supp}}(P)$. Then $\widetilde{E}_\textnormal{h}(s,P,\sigma)$ defined in Eq.~\eqref{eq:Eb} enjoys the following properties.
	\begin{enumerate}[(a)]
		\item\label{prop_b-aa}
		The partial derivatives $\partial \widetilde{E}_\textnormal{h}(s,P,\sigma)/\partial s$, $\partial^2 \widetilde{E}_\textnormal{h}(s,P,\sigma)/\partial s^2$, $\partial^3 \widetilde{E}_\textnormal{h}(s,P,\sigma)/\partial s^3$, and $\widetilde{E}_\textnormal{h}(s,P,\sigma)$ are all continuous for $(s,P)\in \mathbb{R}_{\geq 0} \times \mathscr{P}(\mathcal{X})$.
		
		\item\label{prop_b-bb}
		For every $P\in\mathscr{P}(\mathcal{X})$, the function $\widetilde{E}_\textnormal{h}(s,P,\sigma)$ is concave in $s$ for all $s\in\mathbb{R}_{\geq 0}$.
		
		\item\label{prop_b-cc}
		For every $P\in\mathscr{P}(\mathcal{X})$,
		\begin{align} \label{eq:first_bb}
		\left.\frac{ \partial \widetilde{E}_\textnormal{h}(s,P,\sigma)}{\partial s}\right|_{s = 0} = D\left(\mathscr{W}\|\sigma|P\right).
		\end{align}
		
		\item\label{prop_b-dd}
		For every $P\in\mathscr{P}(\mathcal{X})$,
		\begin{align} \label{eq:E_b_II}
		\lim_{s\to+\infty} \frac{ \partial \widetilde{E}_\textnormal{h}(s,P,\sigma)}{\partial s}\leq 
		\frac{ \partial \widetilde{E}_\textnormal{h}(s,P,\sigma)}{\partial s} \leq D\left(\mathscr{W}\|\sigma|P\right), \; \forall s\in\mathbb{R}_{\geq 0}.  
		\end{align}
		
		\item\label{prop_b-ee}
		For every $P\in\mathscr{P}(\mathcal{X})$,
		\begin{align} \label{eq:second_bb}
		\left.\frac{ \partial^2 \widetilde{E}_\textnormal{h}(s,P,\sigma)}{\partial s^2}\right|_{s = 0} = -\widetilde{V}\left( \mathscr{W}\|\sigma|P\right).
		\end{align}
	\end{enumerate}
\end{prop7}

\begin{proof}[Proof of Proposition~\ref{prop:prop_b}]
	This proof follows similarly from Proposition \ref{prop:prop_h}.
	\begin{itemize}
		\item[(\ref{prop:prop_b}-\ref{prop_b-aa})]		
		Direct calculation yields that
		\begin{align}
			\frac{\partial \widetilde{E}_\text{h} (s,P,\sigma)}{\partial s} &= 
			\widetilde{D}_{\frac{1}{1+s}} \left( \mathscr{W} \| \sigma |P \right) - \frac{s}{(1+s)^2} \widetilde{D}_{\frac{1}{1+s}}'\left( \mathscr{W}\|\sigma|P \right) \label{eq:tE_1b} \\
			\frac{\partial^2 \widetilde{E}_\text{h}(s,P,\sigma) }{\partial s^2}&= 
			-\frac{2}{(1+s)^3} \widetilde{D}_{\frac{1}{1+s}}' \left(\mathscr{W} \| \sigma |P \right) + \frac{s}{(1+s)^4} \widetilde{D}_{\frac{1}{1+s}}'' \left( \mathscr{W}\| \sigma |P \right) \label{eq:tE_2b} \\
			\frac{\partial^3 \widetilde{E}_\text{h}(s,P,\sigma) }{\partial s^3}&= 
			\frac{6}{(1+s)^4} \widetilde{D}_{\frac{1}{1+s}}' \left( \mathscr{W} \| \sigma |P \right) + \frac{3-3s}{(1+s)^5} \widetilde{D}_{\frac{1}{1+s}}'' \left( \mathscr{W} \| \sigma |P \right) \notag\\
			&\quad - \frac{s}{(1+s)^6} \widetilde{D}_{\frac{1}{1+s}}''' \left( \mathscr{W} \| \sigma |P \right).
			\label{eq:tE_3b}
		\end{align}
		From Eqs.~\eqref{eq:tE_1b}-\eqref{eq:tE_3b} and the fact that $\widetilde{D}_{1/(1+s)}\left(\mathscr{W}\|\sigma|P\right)$,
		$\widetilde{D}_{1/(1+s)}'\left(\mathscr{W}\|\sigma|P\right)$,
		$\widetilde{D}_{1/(1+s)}''\left(\mathscr{W}\|\sigma|P\right)$, and $D_{1/(1+s)}'''\left(\mathscr{W}\|\sigma|P\right)$ are continuous for $(s,P)\in \mathbb{R}_{\geq 0}\times \mathscr{P}(\mathcal{X})$, we deduce the continuity property in item \ref{prop_b-aa}.

		\item[(\ref{prop:prop_b}-\ref{prop_b-bb})]	
		The proof strategy follows closely with \cite[Appendix B]{MO14b}. Let 
		\begin{align}
		\tilde{\psi}(\alpha) = \sum_{x\in\mathcal{X}} P(x) \log \Tr\left[ \mathrm{e}^{ \alpha \log W_x+(1-\alpha)\log \sigma} \right].
		\end{align}
		Since $\alpha \mapsto \tilde{\psi}(\alpha)$ is convex for all $\alpha\in(0,1]$ \cite[Lemma III.11]{MO14b}, it can be written as the supremum of affine functions, i.e.~
		\begin{align}
			\tilde{\psi}(\alpha)  = \sup_{i\in\mathcal{I}} \left\{ c_i \alpha + d_i  \right\}
		\end{align}
		for some index set $\mathcal{I}$.
		Hence,
		\begin{align} \label{eq:E_b_convex}
			- \widetilde{E}_\text{h}(s,P,\sigma) = (1+s) \tilde{\psi}\left( \frac{1}{1+s} \right)  = \sup_{i\in\mathcal{I}} \left\{ c_i + d_i (1+s)\right\}.
		\end{align}
		The right-hand side of Eq.~\eqref{eq:E_b_convex}, in turn, implies that the map $s\mapsto \widetilde{E}_\text{h}(s,P,\sigma)$ is convex for all $s\in\mathbb{R}_{\geq 0}$.
		
		\item[(\ref{prop:prop_b}-\ref{prop_b-cc})]	 From Eqs.~\eqref{eq:tE_1b} and \eqref{eq:tE_2b} and recalling \cite[Lemma III.4]{MO14b}, one finds
		\begin{align}
			\left.\frac{\partial \widetilde{E}_\text{h} (s,P,\sigma)}{\partial s}\right|_{s=0} 
			&= D\left(\mathscr{W}\|\sigma | P\right). \label{eq:tE_I_b}
		\end{align}
		
		\item[(\ref{prop:prop_b}-\ref{prop_b-dd})]	 The concavity of the map $s\mapsto E_\text{h}(s,P)$ in item \ref{prop_b-bb} ensures that  $\partial {E}_\text{h}(s,P)/\partial s$ is non-increasing in $s$. Along with Eq.~\eqref{eq:tE_I_b} in item \ref{prop_b-c}, we conclude Eq.~\eqref{eq:E_b_I}.
		
		\item[(\ref{prop:prop_b}-\ref{prop_h-ee})]	Following similar steps in \cite[Proposition 4]{LT15}, it can be verifies that
		\begin{align} \label{eq:tEb1}
		\left.\widetilde{D}'_{\alpha}(\rho\|\sigma)\right|_{\alpha=1} = \lim_{\alpha\uparrow 1} \frac12 \frac{\mathrm{d}^2}{\mathrm{d}\alpha^2} \log f(\alpha)
		= \frac{f(1) f''(1) - (f'(1))^2}{2 (f(1))^2},
		\end{align}
		where $f(\alpha) := \Tr\left[ \mathrm{e}^{\alpha\log \rho + (1-\alpha)\sigma}  \right]$.
		Further, the Fr\'echet derivative of the exponential (see e.g.~\cite[Example X.4.2]{Bha97}) gives
		\begin{align}
		&f'(\alpha) = \Tr\left[ \mathrm{e}^{\alpha\log \rho + (1-\alpha)\sigma} \left( \log \rho - \log \sigma \right) \right], \\
		&f''(\alpha) = \int_{0}^1 \mathrm{d}t \Tr\left[ \mathrm{e}^{t(\alpha\log \rho + (1-\alpha)\sigma)} \left( \log \rho - \log \sigma \right) \mathrm{e}^{(1-t)(\alpha\log \rho + (1-\alpha)\sigma)  } \left( \log \rho - \log \sigma \right) \right],
		\end{align}
		Therefore, Eq.~\eqref{eq:tEb1} equals
		\begin{align}
		\left.\widetilde{D}'_{\alpha}(\rho\|\sigma)\right|_{\alpha=1} &= \frac12 \left(  
		\int_{0}^1\mathrm{d} t\Tr\left[ \rho^{1-t} (\log \rho - \log \sigma ) \rho^t (\log \rho - \log \sigma ) \right] - D(\rho\|\sigma)^2
		\right) \\
		&= \frac12\widetilde{V}(\rho\|\sigma).
		\end{align}
		Finally, combining with Eq.~\eqref{eq:tE_2b} yields
		\begin{align}
			\left.\frac{\partial^2 \widetilde{E}_\text{h}(s,P,\sigma) }{\partial s^2}\right|_{s=0} 
			&= - \widetilde{V}\left(\mathscr{W}\|\sigma|P\right). \label{eq:tE_V2_b}
		\end{align}
	\end{itemize}
\end{proof}

\section{A Weak Converse Bound: Proof of Proposition \ref{prop:weak}} \label{app:weak}

\begin{prop2}
	[Weak Converse Bound with Polynomial Prefactors]
	Consider a classical-quantum channel $\mathscr{W}:\mathcal{X}\to\mathcal{S(H)}$ with $\mathscr{S}_\circ := \overline{\textnormal{\textsf{im}}(\mathscr{W})}$, an arbitrary rate $R\geq 0$, and $\sigma\in\mathcal{S}_{>0}(\mathcal{H})$.
	For any $\eta\in(0,\frac12)$ and $c>0$, let $N_0\in\mathbb{N}$ such that for all $n\geq N_0$,
	\begin{align}  \label{eq:weak2}
	c\cdot \mathrm{e}^{ - \xi \sqrt{ n}  } \leq \frac{\eta}{2},
	\end{align}
	where $\xi = \sqrt{ 2A / \eta}$ and $A := \max_{ \rho\in\mathscr{S}_\circ  }  V(\rho\|\sigma)$.
	Then, it holds that for all $n\geq N_0$,
	\begin{align}
	\widehat{\alpha}_{c\exp\{-nR\}}\left( W_{\mathbf{x}^n}^{\otimes n}\| \sigma^{\otimes n} \right) 
	\geq 
	f(\eta)
	\exp\left\{-n \left[ \frac{\widetilde{E}_\textnormal{sp}\left(R- \frac{2\xi}{\sqrt{n}} , P_{\mathbf{x}^n} , \sigma \right)}{1-\eta}  \right] \right\},
	\end{align}
	where $f(\eta) = \exp\left\{ -\frac{h\left(1-\eta\right)}{1-\eta} \right\}$
	and $h(p) := - p \log p - (1-p) \log (1-p)$ is the binary entropy function.
\end{prop2}

\begin{remark} \label{remark2}
	Consider a constant composition code with common type $P_{\mathbf{x}^n}$ on a finite input alphabet $\mathcal{X}$.
	Recall the definition of the weak sphere-packing exponent  \cite{Win99,CHT16b}: 
	\begin{align} \label{eq:weakSP}
	\widetilde{E}_\text{sp}(R,P_{\mathbf{x}^n}) :=
	\min_{\bar{\mathscr{W}}:\mathcal{X}\to\mathcal{S(H)}} \left\{
	D\left(\bar{\mathscr{W}}\|\mathscr{W}|P_{\mathbf{x}^n}\right): I(P_{\mathbf{x}^n},\bar{\mathscr{W}}) \leq R \right\}.
	\end{align}
	 Proposition \ref{prop:weak}, along with \cite[Lemma 11]{CHT16b}, establishes a weak sphere-packing bound with polynomial prefactors, which generalizes Altu{\u{g}} and Wagner's result \cite[Lemma 3]{AW14b} to c-q channels: for any $\eta\in(0,\frac12)$ and for all sufficiently large $n$ such that Eq.~\eqref{eq:weak2} holds, we have
	\begin{align}
	\epsilon_{\max}(\mathscr{W},P_{\mathbf{x}^n}) 
	&\geq 
	\max_{\sigma\in\mathcal{S(H)}}\widehat{\alpha}_{\exp\{-nR\}}\left( W_{\mathbf{x}^n}^{\otimes n} \| \sigma^{\otimes n} \right) \\
	&\geq \widehat{\alpha}_{\exp\{-nR\}}\left( W_{\mathbf{x}^n}^{\otimes n} \| (\sigma^\star)^{\otimes n} \right) \\
	&\geq  
	f(\eta)
	\exp\left\{-n \left[ \frac{\widetilde{E}_\textnormal{sp}\left(R- \frac{2\xi}{\sqrt{n}} , P_{\mathbf{x}^n} \right)}{1-\eta}  \right] \right\}, \label{eq:weak3}
	\end{align}
	where $\sigma^\star := P_{\mathbf{x}^n} \bar{\mathscr{W}}^\star$ and $\bar{\mathscr{W}}^\star$ is an arbitrary minimizer in Eq.~\eqref{eq:weakSP}.
	Moreover, Eq.~\eqref{eq:weak3} improves the prefactor of Winter's weak sphere-packing bound \cite{Win99} from the order of subexponential to polynomial.
\end{remark}

\begin{proof}[Proof of Proposition~\ref{prop:weak}]
Consider an arbitrary sequence $\mathbf{x}^n\in\mathcal{X}^n$ and a test $Q_n$ on $\mathcal{H}^{\otimes n}$. For two c-q channels $\bar{\mathscr{W}},\mathscr{W}: \mathcal{X} \to \mathscr{S}_\circ$, the data-processing inequality implies that
\begin{align}
D\left( \bar{W}_{\mathbf{x}^n}^{\otimes n} \| {W}_{\mathbf{x}^n}^{\otimes n} \right) 
&\geq \left[1-\alpha(Q_n;\bar {W}_{\mathbf{x}^n}^{\otimes n}) \right] \log \frac{1-\alpha(Q_n;\bar{W}_{\mathbf{x}^n}^{\otimes n})}{1-\alpha(Q_n;{W}_{\mathbf{x}^n}^{\otimes n})} + \alpha(Q_n;\bar{W}_{\mathbf{x}^n}^{\otimes n}) \log \frac{\alpha(Q_n;\bar{W}_{\mathbf{x}^n}^{\otimes n})}{\alpha(Q_n;{W}_{\mathbf{x}^n}^{\otimes n})}\\
&= - h\left(  \alpha(Q_n;\bar{W}_{\mathbf{x}^n}^{\otimes n}) \right) - \alpha(Q_n;\bar{W}_{\mathbf{x}^n}^{\otimes n}) \log \alpha(Q_n; {W}_{\mathbf{x}^n}^{\otimes n})  \notag\\
&\quad\quad - \left[1-\alpha(Q_n;\bar{W}_{\mathbf{x}^n}^{\otimes n})\right] \log \left(1-\alpha(Q_n;{W}_{\mathbf{x}^n}^{\otimes n})\right) \label{eq:change4} \\
&\geq  - \alpha(Q_n;\bar{W}_{\mathbf{x}^n}^{\otimes n}) \log \alpha(Q_n;{W}_{\mathbf{x}^n}^{\otimes n}) - h\left(  \alpha(Q_n;\bar{W}_{\mathbf{x}^n}^{\otimes n}) \right), \label{eq:change5}
\end{align}
where the last inequality \eqref{eq:change5} follows since the third term in \eqref{eq:change4} is non-negative. Continuing from Eq.~\eqref{eq:change5}, we have
\begin{align}
\alpha(Q_n;{W}_{\mathbf{x}^n}^{\otimes n})
&\geq \exp\left\{ - \frac{	D\left( \left.\bar W^{\otimes n}_{\mathbf{x}^n}\right\| W^{\otimes n}_{\mathbf{x}^n} \right) + h\left(  \alpha(Q_n;\bar{W}_{\mathbf{x}^n}^{\otimes n}) \right)}{\alpha(Q_n;\bar{W}_{\mathbf{x}^n}^{\otimes n}) } \right\} \\
&=\exp\left\{ - \frac{	n\,D\left( \left.\left. \bar{\mathscr{W}}\right\| \mathscr{W} \right| P_{\mathbf{x}^n}\right) + h\left( \alpha(Q_n;\bar{W}_{\mathbf{x}^n}^{\otimes n}) \right)}{ \alpha(Q_n;\bar{W}_{\mathbf{x}^n}^{\otimes n}) } \right\}, \label{eq:change2}
\end{align}
where Eq.~\eqref{eq:change2} follows from the additivity of the relative entropy and the empirical distribution $P_{\mathbf{x}^n}$.

The next step is to replace $\alpha(Q_n;{W}_{\mathbf{x}^n}^{\otimes n})$ with a lower bound that does not depend on the dummy channel $\bar W$, provided that $\bar W$ satisfies certain conditions. This can be done using Proposition~\ref{prop:sc}, Wolfowitz's strong converse bound. We delay its proof in Appendix \ref{app:sc}.


\begin{prop}
	[Wolfowitz's Strong Converse] \label{prop:sc}
	Let $\mathscr{S}_\circ\subseteq \mathcal{S(H)}$ be closed and let
	$\bar{\mathscr{W}} : \mathcal{X}\to \mathscr{S}_\circ$ be an arbitrary classical-quantum channel.
	Consider the binary hypothesis testing:
	\begin{align}
	&\mathsf{H}_0:   \bar{W}_{\mathbf{x}^n}^{\otimes n},\\
	&\mathsf{H}_1:   \sigma^{\otimes n},
	\end{align}
	where $\mathbf{x}^n\in\mathcal{X}^n$ and $\sigma\in\mathcal{S}_{>0}(\mathcal{H})$.
	For any test $Q_n$ such that $\beta(Q_n; \sigma^{\otimes n}) \leq \mathrm{e}^{-nR}$ and $D\left(\bar W_{\mathbf{x}^n} \| \sigma | P_{\mathbf{x}^n}\right) \leq R - 2\kappa$, it holds that
	\begin{align}
	\alpha\left( Q_n; \bar{W}_{\mathbf{x}^n}^{\otimes n} \right) > 1 - \frac{A}{n\kappa^2} - \mathrm{e}^{-n\kappa  },
	\end{align}
	where $A := \max_{ \rho\in\mathscr{S}_\circ} V\left( \rho \| \sigma \right)$.
\end{prop}

Fix $0<\eta < \frac12$, and let $\xi^2 := \frac{2A}{\eta}$. Note that $\xi^2$ is finite because $A<+\infty$. For all $n\geq N_0$, we have
\begin{align}
	c \cdot \mathrm{e}^{-\xi\sqrt{n}} &\leq \frac{\eta}{2} \label{eq:cond2}
\end{align}
by assumption in Proposition \ref{prop:weak}.
Choose $\kappa = \xi/\sqrt{n}$. For any $\bar{\mathscr{W}} :\mathcal{X}\to\mathscr{S}_\circ$ with $D\left( \bar{\mathscr{W}} \| \sigma | P_{\mathbf{x}^n} \right) \leq R - \frac{2\xi}{\sqrt{n}}$ and any test $Q_n$ such that $\beta(Q_n;\sigma^{\otimes n})\leq \mathrm{e}^{-nR}$, Proposition \ref{prop:sc} gives a lower bound to the type-I error:
\begin{align}
	\alpha(Q_n; \bar{W}_{\mathbf{x}^n}^{\otimes n}) &\geq 1 - \frac{A}{n\kappa^2} -  \mathrm{e}^{-n\kappa} 
	\geq 1 - \eta. \label{eq:lowerr}
\end{align}
Hence, combining Eqs.~\eqref{eq:change2} and \eqref{eq:lowerr} yields that, for any $\beta(Q_n;\sigma^{\otimes n}) \leq c\mathrm{e}^{-nR}$,
\begin{align}
\alpha(Q_n;{W}_{\mathbf{x}^n}^{\otimes n})
&\geq \max_{\bar{\mathscr{W}}: D\left(\bar{\mathscr{W}}\|\sigma| P_{\mathbf{x}^n}\right)\leq R- \frac{2\xi}{\sqrt{n}}} \exp\left\{ - \frac{	n\,D\left( \left.\left.\bar{\mathscr{W}}\right\| \mathscr{W} \right| P_{\mathbf{x}^n}\right) + h\left(1-\eta\right) }{ 1-\eta } \right\}, \\
& =\exp\left\{ -\frac{h\left(1-\eta\right)}{1-\eta} \right\} \exp\left\{ - \frac{	n\,\widetilde{E}_{\text{sp}}\left(R-\frac{2\xi}{\sqrt{n}},P_{\mathbf{x}^n},\sigma \right) }{ 1-\eta } \right\},
\label{eq:change12}
\end{align}
which concludes Proposition \ref{prop:weak}.

\end{proof}

\subsection{Proof of Wolfowitz's Strong Converse: Proposition \ref{prop:sc}} \label{app:sc}

To prove our claim, we first introduce notation for generalized divergences.
For any $\rho,\sigma\in\mathcal{S}(\mathcal{H})$, and $\gamma > 0$, define the \emph{hockey-stick divergence} by
\begin{align} \label{eq:hockey}
\mathcal{D}_\gamma (\rho \| \sigma) := \Tr \left[ \left(\rho - \gamma \sigma \right)_+\right],
\end{align}
where $A_+ := A\{A\geq 0\}$ denotes the self-adjoint matrix contributed only by its positive part.
This divergence satisfies the data-processing inequality (DPI):
\begin{align}
\Tr \left[ \left(\rho - \gamma \varrho \right)_+\right]
\geq \Tr \left[ \left( \mathcal{N}(\rho) - \gamma \mathcal{N}(\varrho) \right)_+\right],
\end{align}
for any completely positive and trace-preserving map $\mathcal{N}: \mathcal{S}(\mathcal{H}_\text{in}) \to \mathcal{S}(\mathcal{H}_\text{out})$ \cite[Lemma 4]{SW12}. Let 
\begin{align}
\rho_p &:= p |0\rangle\langle 0| + (1-p) |1\rangle\langle 1|,\quad \text{and}\quad
\sigma_q := q |0\rangle\langle 0| + (1-q) |1\rangle\langle 1|,
\end{align}
for $0\leq p,q\leq 1$ and some orthonormal basis $\left\{ |0\rangle, |1\rangle \right\}$, and define
\begin{align} \label{eq:hockey_c}
\mathrm{d}_\gamma \left( p\|q \right) := \mathcal{D}_\gamma\left( \rho_p \| \sigma_q \right).
\end{align}
Note that the quantity $\mathrm{d}_\gamma\left( p\|q \right)$ is independent of the choice of the basis $\left\{ |0\rangle, |1\rangle \right\}$.
Now we are ready to prove Proposition \ref{prop:sc}.

\begin{proof}[Proof of Proposition \ref{prop:sc}]
	Fix an arbitrary test $Q_n$ on $\mathcal{H}^{\otimes n}$.
	For notational convenience, we shorthand $\rho^n = \bar W_{\mathbf{x}^n}^{\otimes n}$, $\tau^n = \sigma^{\otimes n}$, $\alpha = \alpha(Q_n; \rho^n)$ and $\beta = \beta = (Q_n; \tau^n)$.
	Further, we assume $\beta(Q_n;\tau^{ n}) \leq  \mathrm{e}^{ -nR }$.
	From the definition of the classical divergence, Eqs.~\eqref{eq:hockey} and \eqref{eq:hockey_c}, and any $\gamma >0$, we find
	\begin{align}
	\mathrm{d}_\gamma (1- \alpha \| \beta)  &= \left( 1-\alpha  - \gamma \beta\right)_+ + \left(\alpha  - \gamma \left[ 1-\beta \right] \right)_+ \\
	&\geq 1 - \alpha  - \gamma \beta\\
	&\geq 1 - \alpha  - \gamma \mathrm{e}^{-nR}.  \label{eq:sc10}
	\end{align}
	On the other hand, DPI and the measurement map $\Tr[Q_n(\cdot)]|0\rangle\langle0| + (1-\Tr[Q_n(\cdot)])|1\rangle\langle 1|$ imply that
	\begin{align}
	 \mathcal{D}_\gamma \left( \rho^n \| \tau^n \right) \geq 
	\mathrm{d}_\gamma \left(\Tr[Q_n \rho^n] \| \Tr[Q_n\tau^n]\right) = \mathrm{d}_\gamma (1- \alpha \| \beta). \label{eq:sc9}
	\end{align}
	Hence, Eqs.~\eqref{eq:sc10} and \eqref{eq:sc9} lead to
	\begin{align}\label{eq:sc1}
	\alpha &\geq 1 - \mathcal{D}_\gamma \left( \rho^n \| \tau^n \right) - \gamma \mathrm{e}^{-nR}.
	\end{align}
	Since
	\begin{align}
	\mathcal{D}_\gamma \left( \rho^n \| \tau^n \right) 
	&= \Tr\left[ \left\{\rho^n - \gamma \tau^n \geq 0 \right\} \left( \rho^n - \gamma \tau^n \right) \right] \\
	&\leq \Tr\left[ \left\{ \rho^n - \gamma \tau^n \geq 0 \right\}  \rho^n \right], \label{eq:sc2}
	\end{align}
	continuing from Eq.~\eqref{eq:sc1} gives
	\begin{align}
	\alpha \geq 1 -\Tr\left[ \left\{ \rho^n - \gamma \tau^n \geq 0 \right\}  \rho_{^n} \right] - \gamma \mathrm{e}^{-nR}. \label{eq:sc3}
	\end{align}
	
	Next, invoking Lemma \ref{lemm:Chebyshev} below, for all $\log \gamma >  D\left( \rho^n\| \tau^n \right)$, we have
	\begin{align}
	\alpha &\geq 1 - \frac{ V\left(\rho^n\|\tau^n \right)
	}{ \left[\log \gamma - D\left(\rho^n\|\tau^n  \right) \right]^2}   - \gamma \mathrm{e}^{-nR}  \\
	&= 1 - \frac{ V\left( \bar{\mathscr{W}}\|\sigma | P_{\mathbf{x}^n} \right)
	}{ n \left[ \frac{\log \gamma}{n}  - D\left( \bar{\mathscr{W}}\|\sigma | P_{\mathbf{x}^n} \right) \right]^2}   - \gamma \mathrm{e}^{-nR}  \label{eq:sc3new}
	\end{align}
	Finally, recall $D \left( \bar{\mathscr{W}}\|\sigma | P_{\mathbf{x}^n} \right) \leq R - 2\kappa$
	and $A := \max_{\rho\in\mathscr{S}_\circ} V\left( \rho \| \sigma \right)$ and choose $\log \gamma =  n D \left( \bar{\mathscr{W}}\|\sigma | P_{\mathbf{x}^n} \right) + n\kappa$.
	Then, Eq.~\eqref{eq:sc3new} yields, for any test $Q_n$ and $\beta(Q_n; \sigma^{\otimes n})\leq  \mathrm{e}^{-nR}$,
	\begin{align}
	\alpha\left(Q_n; \bar{\mathscr{W}}_{\mathbf{x}^n}^{\otimes n}\right)
	&\geq 1 - \frac{V\left( \bar{\mathscr{W}}\|\sigma | P_{\mathbf{x}^n} \right)}{n\kappa^2} - \mathrm{e}^{-n\kappa} \\
	&\geq 1 - \frac{ A }{n\kappa^2} - \mathrm{e}^{-n\kappa},
	\end{align}
	which concludes the proof.
	
	\begin{lemm}[Quantum Chebyshev's Inequality {\cite[Lemma 6]{SW12}}] \label{lemm:Chebyshev}
		Let $\rho,\sigma \in \mathcal{S}(\mathcal{H})$ and assume $\log \gamma > D(\rho\|\sigma)$. Then
		\begin{align} \label{eq:Che}
		\Tr\left[ \rho \left\{ \rho - \gamma \sigma\geq 0  \right\} \right] \leq 
		\frac{ V(\rho\|\sigma) }{ \left[  \log \gamma - D(\rho\|\sigma)\right]^2}.
		\end{align}
	\end{lemm}

\end{proof}

\section{A Sharp Converse Bound from Strong Large Deviation}  \label{app:sharp}
In this section, we provide the proof of Proposition \ref{prop:strong}.
Our technique highly relies on a strong large deviation inequality.

\subsection{A Strong Large Deviation Inequality}
Let $\left(X_i\right)_{i\in\mathbb{N}}$ be a sequence of independent, real-valued random variables with probability measures $\left(\mu_i\right)_{i=1}^n$. Let
$Z_n := \sum_{i=1}^n X_i$ and let
 $\Lambda_n(t) :=  \log \mathbb{E}\left[\mathrm{e}^{t Z_n}\right]$. Define the Legendre-Fenchel transform of $\frac1n \Lambda_n(\cdot)$ by:
\begin{align} \label{eq:LF}
\Lambda_n^*(z) := \sup_{t \in \mathbb{R}} \left\{ zt - \frac1n  \Lambda_n(t) \right\}, \quad \forall z\in\mathbb{R}.
\end{align}
Let $\left(T_n\right)_{n\in\mathbb{N}}$ be a bounded sequence of real numbers and $\left(t_n^\star\right)_{n\in\mathbb{N}}$ be a sequence satisfying for all $n\in\mathbb{N}$
\begin{align}
t_n^\star &\in(0,1);\\
	T_n &=\frac1n \Lambda_n'(t_n^\star); \\
	\Lambda_n^*(T_n) &= T_n t_n^\star - \frac1n \Lambda_n (t_n^\star). \label{eq:Rao1}
\end{align}
With these definitions, we can now state the following sharp concentration inequality for $\frac1n Z_n$:

\begin{theo}[Chaganty-Sethuraman's Concentration Inequality {\cite[Theorem 3.3]{CS93}} ] \label{theo:strong}
For any $\eta\in(0,1)$, there exists an $N_0\in\mathbb{N}$ such that, for all $n\geq N_0$,
\begin{align}
\Pr\left\{\frac1n Z_n \geq T_n, \right\} \geq \frac{1-\eta}{  t_n^\star \sqrt{2\pi n m_{2,n}} } \exp\{ -n\Lambda_n^\star(T_n)  \},
\end{align}
where  $m_{2,n} := \frac1n \sum_{i=1}^n \Var_{\tilde{\mu}_{n,i}}\left[{X}_i\right]$, and the measure $\tilde{\mu}_{n,i}$ is defined via
\begin{align}
\frac{\mathrm{d}\tilde{\mu}_{n,i}}{\mathrm{d}\mu_i} (y) := \frac{\mathrm{e}^{y t_n^\star}}{ \mathbb{E}\left[\mathrm{e}^{t_n^\star X_i}\right]}.
\end{align}
\end{theo}

\begin{remark}
	Chaganty and Sethuraman in Ref.~{\cite[Theorem 3.3]{CS93}} considered a more general sequence of random variables $\{Z_n\}_{n\in\mathbb{N}}$, which are not necessarily the sum of random variables.
	They proved Theorem~\ref{theo:strong} provided that the following condition is satisfied: there exists $\delta_0>0$ such that for any $\delta$ and $\lambda$ with $0<\delta<\delta_0<\lambda$,
	$\sup_{\delta<|t|\leq \lambda t_n^\star} \left| \Lambda_n(t_n^\star + \mathrm{i} t) / \Lambda_n(t_n^\star) \right| = o(1/\sqrt{n})$, where the supremum is defined to be $0$ if $\{t:\delta<|t|\leq \lambda t_n^\star\}$ is empty.
	In the case of $Z_n$ being a sum of random variables, $ \Lambda_n(t_n^\star + \mathrm{i} t) / \Lambda_n(t_n^\star) $ is the product of the characteristic functions of $\{X_i\}_{i=1}^n$.
	Since the supremum of a characteristic function on a compact interval not containing $0$ is less than $1$, this condition is thus satisfied.
	
	We note that the lower bound in Theorem~\ref{theo:strong} for the general sequence of random variables $\left(X_i\right)_{i\in\mathbb{N}}$ suffices to establish the converse, Theorem~\ref{theo:converse}.
	We do not particularly consider the case of lattice valued random variables (see e.g.~~{\cite[Theorem 3.5]{CS93}}).
\end{remark}
\subsection{Proof of Proposition~\ref{prop:strong}}
\begin{prop3} [A Sharp Converse Bound] 
Consider a classical-quantum channel $\mathscr{W}:\mathcal{X}\to\mathcal{S(H)}$ and a state $\sigma\in\mathcal{S(H)}$. Suppose the sequence $\mathbf{x}^n\in\mathcal{X}^n$ satisfies	\begin{align} \label{eq:cond_V22}
\nu \leq V\left( \mathscr{W} \| \sigma | P_{\mathbf{x}^n} \right) < +\infty
\end{align}
for some $\nu > 0$, and suppose the sequence of rates $(R_n)_{n\in\mathbb{N}}$ satisfies  $ D_0(\mathscr{W}\|\sigma|P_{\mathbf{x}^n})  < R_n  < D(\mathscr{W}\|\sigma|P_{\mathbf{x}^n})$. Then, there exists an $N_0\in\mathbb{N}$ such that, for all $n\geq N_0$,
\begin{align}  
\widehat{\alpha}_{\exp\{-nR_n\}}(W_{\mathbf{x}^n}^{\otimes n}\| \sigma^{\otimes n}) 
\geq 
\frac{A}{s_n^\star\sqrt{n}}	\exp\left\{ 	-n  E_\textnormal{sp}^{(2)} \left(R_n - c_n ,P_{\mathbf{x}^n}, \sigma \right) \right\},	
\end{align}
where $c_n = \frac{K\log n}{n}$ and $A,K>0$ are finite constants independent of the sequence $\mathbf{x}^n$, and 
\begin{align}
s_n^\star := \argmax_{s\geq 0} \left\{ E_\textnormal{h}(s,P_{\mathbf{x}^n}, \sigma) - s R_n
\right\}. \label{eq:s_n1}
\end{align}
\end{prop3}

\begin{proof}[Proof of Proposition~\ref{prop:strong}]
Let $\rho^n := W_{\mathbf{x}^n}^{\otimes n}$, $\sigma^n := \sigma^{\otimes n}$, 
${p}^n := \bigotimes_{i=1}^n p_{x_i}$, and ${q}^n := \bigotimes_{i=1}^n q_{x_i}$, where $p_{x_i},q_{x_i}$ are Nussbaum-Szko{\l}a distributions \cite{NS09} of $W_{x_i},\sigma$ for every $i\in [n]$. 
Let $\tilde{R}_n := R_n - \gamma_n$, where $\gamma_n := \frac{\log n}{2n} + \frac{x}{n}$ for some $x\in\mathbb{R}$. The choice of $x$ and the rate back-off term $\gamma_n$ will become evident later. 
Let $N_1 \in\mathbb{N}$ such that $\tilde{R}_n \geq D_0(\mathscr{W}\|\sigma|P_{\mathbf{x}^n})$ for all $n\geq N_1$.
Subsequently, we choose such $n\geq N_1$ onwards.	
	
Since $D_\alpha(W_{x_i}\|\sigma) = D_\alpha(p_{x_i}\|q_{x_i})$, for $\alpha\in(0,1]$, we use the notation 
\begin{align}\phi_n(\tilde{R}_n) :=  
E_\text{sp}^{(2)}(\tilde{R}_n, P_{\mathbf{x}^n}, \sigma) =
\sup_{0<\alpha\leq 1} \frac{1-\alpha}{\alpha}\left(  \sum_{x\in\mathcal{X}}P_{\mathbf{x}^n}(x)D_\alpha(p_{x_i}\|q_{x_i}) - \tilde{R}_n \right),
\end{align}
where $P_{\mathbf{x}^n}$ denotes the empirical distribution of $\mathbf{x}^n = x_1,\ldots x_n$.
Moreover, the condition in Eq.~\eqref{eq:cond_V22} implies that $W_x\ll \sigma$, for all $x\in \texttt{supp}(P_{\mathbf{x}^n})$, and thus $p^n \ll q^n$. Without loss of generality,	we let $q_{x_i}(\omega) = 0$, $\omega \not\in \texttt{supp}(p_{x_i})$ since they won't contribute to $\phi_n(\tilde{R}_n)$.

We apply Nagaoka's argument \cite{Nag06}: for any $0\leq Q_n\leq \mathds{1}$, choosing $\delta = \exp\{n\tilde{R}_n-n\phi_n(\tilde{R}_n)\}$ yields:
	\begin{align}  \label{eq:sharp15}
	\alpha\left(Q_n; \rho^n \right) + \delta \beta\left(Q_n; \sigma^n \right)
	\geq  \frac12 \left( \alpha\left(\mathscr{U};p^n \right) + \mathrm{e}^{n\tilde{R}_n-n\phi_n(\tilde{R}_n)} \beta\left( \mathscr{U};q^n\right) \right),
	\end{align}
	where
	\begin{align}
	\alpha\left( {\mathscr{U}};  {p}^n \right) := \sum_{\omega\in {\mathscr{U}}^\mathrm{c}}  {p}^n(\omega); \quad
	\beta\left( {\mathscr{U}}; {q}^n\right) := \sum_{\omega\in {\mathscr{U}}}  {q}^n(\omega),	
	\end{align}
	and  
	\begin{align} \label{eq:sharp10}
	{\mathscr{U}} &:= \left\{  \omega:  p^n(\omega)\mathrm{e}^{ n{\phi}_n\left( {\tilde{R}_n}\right)} >  q^n(\omega) \mathrm{e}^{ n{\tilde{R}_n}}  \right\}.
	\end{align}	
	
In the following, we will employ Theorem \ref{theo:strong}, to further lower bound $\alpha\left( {\mathscr{U}};  {p}^n \right)$ and $\beta\left( {\mathscr{U}}; {q}^n\right)$. Before proceeding, we need to introduce some notation. Define the \emph{tilted distributions}, for every $i\in[n]$ and $t\in[0,1]$, to be
\begin{align}
\hat{q}_{{x_i},t}(\omega) := \frac{  {p}_{x_i}(\omega)^{1-t}  {q}_{x_i}(\omega)^{t} }{ \sum_{\omega\in \texttt{supp}(p_{x_i})}  {p}_{x_i}(\omega)^{1-t}  {q}_{x_i}(\omega)^{t} }, \quad \omega \in \texttt{supp}(p_{x_i}).
\end{align}
Let
\begin{align}
\Lambda_{0,{x_i}} (t) := \log \mathbb{E}_{ {p}_{x_i}} \left[ \mathrm{e}^{t \log \frac{ {q}_{x_i}}{ {p}_{x_i}} } \right], \quad
\Lambda_{1,{x_i}} (t) := \log \mathbb{E}_{ {q}_{x_i}} \left[ \mathrm{e}^{t \log \frac{ {p}_{x_i}}{ {q}_{x_i}} } \right], \label{eq:zero_derivative}	
\end{align}
Since $p^n$ and $q^n$ share the same support, it can be verified that the maps $t\mapsto \Lambda_{j,x_i}(t)$, $j\in\{0,1\}$ are differential  for all $t\in[0,1]$. One can immediately verify the following partial derivatives with respect to~$t$:
\begin{align}
\Lambda'_{0,{x_i}} (t) = \mathbb{E}_{\hat{q}_{{x_i},t}} \left[ \log \frac{ {q}_{x_i}}{ {p}_{x_i}}  \right], &\quad \Lambda'_{1,{x_i}} (t)  = \mathbb{E}_{\hat{q}_{{x_i},1-t}} \left[ \log \frac{ {p}_{x_i}}{ {q}_{x_i}}  \right]; \label{eq:first_derivative}\\
\Lambda''_{0,{x_i}} (t) = \Var_{\hat{q}_{{x_i},t}} \left[ \log \frac{ {q}_{x_i}}{ {p}_{x_i}}  \right], &\quad \Lambda''_{1,{x_i}} (t) = \Var_{\hat{q}_{{x_i},1-t}} \left[ \log \frac{ {p}_{x_i}}{ {q}_{x_i}}  \right]. \label{eq:second_derivatie} 
\end{align}
Note that Eqs.~\eqref{eq:zero_derivative}, \eqref{eq:first_derivative}, and \eqref{eq:second_derivatie} ensure that 
\begin{align} 
\Lambda_{0,x_i}(t) &= \Lambda_{1,x_i} (1-t), \quad
\Lambda_{0,x_i}'(t) = -\Lambda_{1,x_i}' (1-t),\quad
\Lambda_{0,x_i}''(t) = \Lambda_{1,x_i}'' (1-t) \label{eq:sym0}.
\end{align}
With $\Lambda_{j,{x_i}} (t)$ in Eq.~(\ref{eq:zero_derivative}), we can define 
\begin{align}
&{\Lambda}_{j,P_{\mathbf{x}^n}} (t) := \sum_{x\in \mathcal{X} } P_{\mathbf{x}^n} (x) \Lambda_{j,x}(t), \quad\quad\;
j\in\{0,1\}; \label{eq:FL0}\\
	&\Lambda_{j,P_{\mathbf{x}^n}}^*(z) := \sup_{t\in\mathbb{R}} \left\{ tz - {\Lambda}_{j,P_{\mathbf{x}^n}}(t)  \right\},
	\quad j\in\{0,1\}, \label{eq:FL}
	\end{align}
where  $\Lambda_{j,P_{\mathbf{x}^n} }^*(z)$ in Eq.~(\ref{eq:FL}) are the \emph{Legendre-Fenchel transform} of ${\Lambda_{j,P_{\mathbf{x}^n}}(t)}$. The quantities $\Lambda_{j,P_{\mathbf{x}^n} }^*(z)$ would appear in the lower bounds of $\alpha\left( {\mathscr{U}};  {p}^n \right)$ and $\beta\left( {\mathscr{U}}; {q}^n\right)$ obtained by Theorem~\ref{theo:strong} as shown later.

In the following, we will relate the Legendre-Fenchel transform $\Lambda_{j,P_n}^* (z)$ to the desired error-exponent function $\phi_n(\tilde{R}_n)$.	Such a relationship is stated in the following lemma whose proof was presented in \cite{CHT16b}.
\begin{lemm}[{\cite[Lemma 17]{CHT16b}}] \label{lemm:regularity}
The following holds for all sequences $\mathbf{x}^n$ satisfying Eq.~\eqref{eq:cond_V22} and all $r\in\left(D_0(\mathscr{W}\|\sigma|P_{\mathbf{x}^n}), D(\mathscr{W}\|\sigma|P_{\mathbf{x}^n})\right)$:
\begin{enumerate}[(a)]			
\item\label{regularity-a} 
$\Lambda_{0,P_{\mathbf{x}^n}} '' (t) >0$ for all $t\in[0,1]$.
\item\label{regularity-b} $\Lambda^*_{0,P_{\mathbf{x}^n} } \left(  {\phi}_n({r}) - {r} \right) =  {\phi}_n(r)$.
		
\item\label{regularity-c} $\Lambda^*_{1,P_{\mathbf{x}^n} } \left( r -  {\phi}_n({r}) \right) = r$.
		
\item\label{regularity-d} 
Let $s^\star $ be the optimizer of $E_\textnormal{sp}^{(2)}(r,P_{\mathbf{x}^n},\sigma)$, c.f.~\eqref{eq:s_n1}.
The optimizer of $\Lambda^*_{0,P_{\mathbf{x}^n}}(z)$, denoted by $t^\star$, is unique and satisfies $t^\star = \frac{s^\star}{1+s^\star} \in (0,1)$ and ${\Lambda}'_{0,P_{\mathbf{x}^n} }(t^\star) =  {\phi}_n(r) - r$.  
\end{enumerate}
\end{lemm}

Since the item \ref{regularity-d} in Lemma \ref{lemm:regularity} shows that the optimizer $t$ in Eq.~\eqref{eq:FL} always lies in the compact set $[0,1]$, by invoking Eq.~\eqref{eq:sym0} we define the following quantity:
\begin{align}
V_{\min} (\nu) &:= \min_{t \in[0,1], \,P_{\mathbf{x}^n} \in \mathscr{P}_{\nu}(\mathcal{X}) }  {\Lambda}''_{0,P_n}(t),\label{eq_mhvmin} 
\end{align}
where $\mathscr{P}_{\nu}(\mathcal{X}) := \left\{ P_{\mathbf{x}^n} \in\mathscr{P}(\mathcal{X}): \nu \leq V\left(\mathscr{W}\|\sigma|P_{\mathbf{x}^n}\right) < +\infty  \right\}$ is a compact set owing to the continuity of the map $P\mapsto V\left(\mathscr{W}\|\sigma|P\right)$; see Eq.~\eqref{eq:V_cond}.
	
Further,  from the definitions in Eqs,~\eqref{eq:second_derivatie}, $\Lambda_{0,(\cdot)}''(\cdot)$ is continuous functions in $[0,1]\times \mathscr{P}(\mathcal{X})$. The minimization in the above definitions are well-defined and finite. Further, the quantity $V_{\min}(\nu)$ is bounded away from zero owing to item \ref{regularity-a} in Lemma \ref{lemm:regularity}.

	Now, we are ready to derive the lower bounds to $\alpha\left({\mathscr{U}};  p^n \right)$ and $\beta\left({\mathscr{U}}; q^n \right)$.
	Fix an arbitrary $\eta\in(0,1)$.
	Applying Theorem~\ref{theo:strong} to  $X_i = \log  {q}_i - \log  {p}_i$ with probability measure $ {p}_i$, and threshold $ T_n = {\tilde{R}_n} -  {\phi}_n( {\tilde{R}_n}) $ gives, for all sufficiently large $n$, say $n\geq N_2\in\mathbb{N}$,
	\begin{align}
	\alpha\left( {\mathscr{U}};  {p}^n \right) &:= \sum_{\omega\in {\mathscr{U}}^\mathrm{c}}  {p}^n(\omega)  \\
	&= \Pr\left\{ \frac1n\sum_{i=1}^n Z_i \geq  \tilde{R}_n - {\phi}_n( {
		\tilde{R}_n})  \right\} \\	
	&\geq \frac{1-\eta}{t_n^\star\sqrt{2\pi n V_{\min}(\nu)}} \exp\left\{  -n \Lambda^*_{0,P_{\mathbf{x}^n} } \left(  {\phi}_n( {\tilde{R}_n}) -  {\tilde{R}_n} \right)  \right\}, \label{eq:sharp12mh}
	\end{align}
	where
	\begin{align}
	t_n^\star := \argmax_{t\in\mathbb{R}} \left\{ tz_n   - \Lambda_{0,P_{\mathbf{x}^n}}(t)   \right\}
	\end{align}
	Similarly, applying again Theorem~\ref{theo:strong} to $X_i = \log  {p}_i - \log  {q}_i$ with probability measure $=  {q}_i$, and threshold $ {\phi}_n( {\tilde{R}_n}) -  {\tilde{R}_n}$ yields, for all sufficiently large $n$, say $n\geq N_3\in\mathbb{N}$,
	\begin{align}
	\beta\left( {\mathscr{U}}; {q}^n\right) &:= \sum_{\omega\in {\mathscr{U}}}  {q}^n(\omega) \\
	&= \Pr\left\{ \frac1n\sum_{i=1}^n Z_i \geq  {\phi}_n( {\tilde{R}_n}) -  {\tilde{R}_n}  \right\} \\
	&\geq \frac{1-\eta}{(1-t_n^\star)\sqrt{2\pi n V_{\min}(\nu)}} \exp\left\{  -n \Lambda^*_{1,P_{\mathbf{x}^n} } \left(  {\tilde{R}_n} -  {\phi}_n( {\tilde{R}_n}) \right)  \right\} \label{eq:sharp14mh}\\
	&\geq \frac{1-\eta}{\sqrt{2\pi n V_{\min}(\nu)}} \exp\left\{  -n \Lambda^*_{1,P_{\mathbf{x}^n} } \left(  {\tilde{R}_n} -  {\phi}_n( {\tilde{R}_n}) \right)  \right\},	
	\label{eq:sharp13mh}
	\end{align}
where the term $1-t_n^\star$ in Eq~\eqref{eq:sharp14mh} comes from the symmetry in Eq.~\eqref{eq:sym0}, and 	the last inequality \eqref{eq:sharp13mh} follows from $t_n^\star\in(0,1)$ in item \ref{regularity-d} of Lemma \ref{lemm:regularity}.

Continuing from Eq.~\eqref{eq:sharp12mh} and item \ref{regularity-b} in Lemma \ref{lemm:regularity} gives
\begin{equation}
\alpha\left( {\mathscr{U}};  {p}^n \right)\geq \frac{1-\eta}{t_n^\star\sqrt{2\pi n V_{\min}(\nu)}} \exp\{-n   {\phi}_n( {\tilde{R}_n})\}. \label{eq:sharp12}
\end{equation}
Eq.~\eqref{eq:sharp13mh} together with item \ref{regularity-c} in Lemma \ref{lemm:regularity} yields
\begin{equation}
\beta\left( {\mathscr{U}}; {q}^n\right) \geq \frac{1-\eta}{\sqrt{2\pi n V_{\min}(\nu)}} \exp\{-n \tilde{R}_n\} = 2\exp\{-nR_n\}, \label{eq:sharp13}
\end{equation}
where we choose $x = \log 2\sqrt{2\pi V_{\min}(\nu)} - \log (1-\eta)$ in the rate back-off $\gamma_n = \frac{\log n}{2n} + \frac{x}{n}$. Thus we can bound the left-hand side of Eq.~(\ref{eq:sharp15}) from below. If for any test $0\leq Q_n \leq \mathbb{1}$ such that 
\begin{align}
\beta(Q_n;\sigma^n) \leq \exp\{-nR_n\}, \label{eq:sharp_beta}
\end{align}
holds, it implies that
\begin{align}
\alpha(Q_n;\rho^n) \geq \frac{1-\eta}{t_n^\star 2 \sqrt{2\pi n V_{\min}(\nu)}} \exp\{-n   {\phi}_n( {\tilde{R}_n})\}. \label{eq:sharp_alpha}
\end{align}
Finally, let $A := (1-\eta)/(2 \sqrt{2\pi V_{\min}(\nu)})$ and choose a constant $K>0$ such that for all $n\geq N_0 := \max\{N_1,N_2,N_3\}$,
\begin{align}
\gamma_n = \frac{\log n}{2n} + \frac{\log 2\sqrt{2\pi V_{\min}(\nu)} - \log (1-\eta) }{n} \leq \frac{K\log n}{n} =: c_n. \label{eq:g(n)}
\end{align}
Since the map $r\mapsto \phi_n(r)$ is monotonically decreasing \cite[Section 5]{ANS+08}, Eqs.~\eqref{eq:sharp_beta}, \eqref{eq:sharp_alpha}, and \eqref{eq:g(n)} conclude our result: for all $n\geq N_0$,
\begin{align}
\widehat{\alpha}_{\exp\{ -nR \}  } \left( \rho^n\| \sigma^n \right) &\geq \frac{A}{t_n^\star\sqrt{n}} \exp\left\{ -n E_\text{sp}^{(2)}\left(  R_n - c_n  , P_{\mathbf{x}^n}, \sigma\right)  \right\} \\
&\geq \frac{A}{s_n^\star\sqrt{n}} \exp\left\{ -n E_\text{sp}^{(2)}\left(  R_n - c_n  , P_{\mathbf{x}^n}, \sigma\right)  \right\},
\end{align}	
where the last inequality follows from item \ref{regularity-d} in Lemma~\ref{lemm:regularity}: $t_n^\star = s_n^\star/(1+s_n^\star) \in(0,1)$.
\end{proof}

\section{Proof of Proposition \ref{prop:spCh}} \label{proof:spCh}
\begin{prop6}[Error Exponent around Capacity] 
Let $(b_n)_{n\in\mathbb{N}}$ be a sequence of real numbers with $\lim_{n\to+\infty}  b_n = 0$ and let  $(\delta_n)_{n\in\mathbb{N}}$ be a sequence of positive numbers with $\lim_{n\to+\infty} \delta_n = 0$.
Suppose the sequence of distributions $(P_n)_{n\in\mathbb{N}}$ satisfies
\begin{align}
C_{\mathscr{W}} - \delta_n < D(\mathscr{W}\| P^\star \mathscr{W}| P_n) \leq 	C_{\mathscr{W}} - b_n. \label{eq:spChcc}
\end{align}
The following holds:
\begin{align}
\limsup_{n\to+\infty} \frac{ E_\textnormal{sp}^{(2)} \left(C_\mathscr{W} -\delta_n, P_n, P^\star \mathscr{W} \right) }{ \delta_n^2 } &\leq \limsup_{n\to+\infty} \frac{(\delta_n - b_n )^2}{2 V_\mathscr{W} \delta_n^2}; \label{eq:spChc11} \\
\limsup_{n\to+\infty} \frac{ \widetilde{E}_\textnormal{sp} \left(C_\mathscr{W} -\delta_n, P_n, P^\star \mathscr{W} \right) }{ \delta_n^2 } &\leq \limsup_{n\to+\infty} \frac{(\delta_n - b_n )^2}{2 \widetilde{V}_\mathscr{W} \delta_n^2}; \label{eq:spChc22} \\
\limsup_{n\to+\infty} \frac{s_n^\star}{\delta_n} &\leq \frac{1}{V_\mathscr{W}}, \label{eq:spChc33}
\end{align}
where 
\begin{align}
s_n^\star := \argmax_{s\geq 0} \left\{ E_\textnormal{h}(s,P_n, P^\star\mathscr{W}) - s \left( C_\mathscr{W} - \delta_n  \right)
\right\}.
\end{align}
\end{prop6}


\begin{proof}[Proof of Proposition \ref{prop:spCh}]
We only prove Eqs.~\eqref{eq:spChc11} and \eqref{eq:spChc33}, since Eq.~\eqref{eq:spChc22} follows from the same argument and Proposition~\ref{prop:prop_b}.

Recall the error-exponent function $E_\text{sp}^{(2)}$:
\begin{align}\label{eq_D1}
E_\textnormal{sp}^{(2)} \left(C_\mathscr{W} -\delta_n, P, P^\star \mathscr{W} \right) &= \sup_{s\geq 0} \left\{  -s\left( C_\mathscr{W}-\delta_n \right) + E_\textnormal{h}(s,P,P^\star \mathscr{W})\right\}.
\end{align}
In the following, we fix $\sigma = P^\star \mathscr{W}$ in the definition of $E_\text{h}$ (Eq.~\eqref{eq:Eh}) and denote by
\begin{align}
E_\textnormal{h}(s,P) :=
E_\textnormal{h}(s,P, P^\star\mathscr{W}) =
 s D_{\frac{1}{1+s}}\left( \mathscr{W}\| P^\star \mathscr{W} | P \right).
\end{align}
for notational convenience.
We define a \emph{critical rate} for a c-q channel $\mathscr{W}$ to be
\begin{align} \label{eq:critical}
r_\textnormal{cr} := \max_{P\in\mathscr{P}(\mathcal{X})}\left.\frac{\partial E_\text{h}(s,P) }{\partial s}\right|_{s=1}.
\end{align}
Let $N_0$ be the smallest integer such that $C_\mathscr{W} - \delta_n > r_\text{cr}$, $\forall n\geq N_0$. 
Since the map $r \mapsto E_\text{sp}^{(2)} (r,\cdot,\cdot) $ is non-increasing \cite[Section 5]{ANS+08}, the maximization over $s$ in Eq.~(\ref{eq_D1}) can be restricted to the set $[0,1]$ for any rate  above $r_\text{cr}$, i.e.,
\begin{align} \label{eq:sp2}
 E_\textnormal{sp}^{(2)} \left(C_\mathscr{W} -\delta_n, P_n, P^\star \mathscr{W} \right)
=  \max_{0\leq s\leq 1} \left\{  -s\left( C_\mathscr{W}-\delta_n \right) + E_\textnormal{h}(s,P_n)        \right\}.
\end{align}
For every $n\in\mathbb{N}$, let $s_n^\star$ attain the maxima in Eq.~\eqref{eq:sp2} at a rate of $C_\mathscr{W} - \delta_n \geq 0$.
In the following lemma, we  discuss the asymptotic behavior of $\{s_n^\star\}_{n\in\mathbb{N}}$.
\begin{lemm} \label{lemm:s_n}
Let $s_n^\star$ attain the maxima in Eq.~\eqref{eq:sp2} and $P_n$ satisfy Eq.~\eqref{eq:spChcc}. We have
	\begin{enumerate}[(a)]
		\item\label{s_n-a} The limit point of $\{P_n\}_{n\in\mathbb{N}}$ is capacity achieving.
		\item\label{s_n-b} $s_n^\star>0$ for all $n\in\mathbb{N}$ and $\lim_{n\to +\infty} s_n^\star = 0$.
	\end{enumerate}
\end{lemm}
\begin{proof}[Proof of Lemma \ref{lemm:s_n}]
Let $\{P_{n_k}\}_{k\geq 1}$ and $\{s_{n_k}^\star\}_{k\geq 1}$ be arbitrary subsequences.
Since $\mathscr{P}(\mathcal{X})$ and $[0,1]$ are compact, we may assume that
\begin{align} 
\lim_{k\to+\infty} P_{n_k} = P_o, \quad \lim_{k\to\infty} s_{n_k}^\star = s_o,
\end{align}
for some $P_o\in\mathscr{P}(\mathcal{X})$ and $s_o \in [0,1]$.
\begin{itemize} 
\item[(\ref{lemm:s_n}-\ref{s_n-a})]
Let $k\to+\infty$. Eq.~\eqref{eq:spChcc} implies that
\begin{align}
D(\mathscr{W}\| P^\star \mathscr{W}| P_o) = C_\mathscr{W},
\end{align}
which guarantees that $P_o$ is capacity-achieving by the dual representation of the information radius, see e.g.~\cite{SW01}, \cite[Theorem 2]{TT13}.

%
		
\item[(\ref{lemm:s_n}-\ref{s_n-b})]
One can observe from Eq.~\eqref{eq:sp2} that $s_n^\star = 0$ if and only if $C_\mathscr{W} - \delta_n \geq D(\mathscr{W}\|P^\star\mathscr{W}|P_n)$.
However, this violates the assumption in Eq.~\eqref{eq:spChcc}. Hence, we have $s_n^\star > 0$ for all $n\in\mathbb{N}$.

Since $P_o$ is capacity achieving, the uniqueness of the divergence center implies that $P_o \mathscr{W} = P^\star \mathscr{W}$. 
Item \ref{prop_h-c} in Proposition \ref{prop:prop_h} shows that
\begin{align} \label{eq:s_n6}
\left.\frac{\partial^2 E_\text{h}\left(s,P_o\right)}{\partial s^2}\right|_{s=0} = 
- V\left(\mathscr{W}\|P^\star \mathscr{W} |P_o\right) = 
- V(P_o,\mathscr{W}) \leq -V_\mathscr{W} < 0,
\end{align}
where the last inequality follows since $V_\mathscr{W} > 0$.
Then, Eq.~\eqref{eq:s_n6} implies that the first-order derivative $\partial E_\text{h}\left(s,P_o\right)/{\partial s}$ is strictly decreasing around $s=0$. 
Moreover, item \ref{prop_h-d} in Proposition~\ref{prop:prop_h} gives 
\begin{align} \label{eq:s_n2}
\left.\frac{\partial E_\text{h}\left(s,P_o\right)}{\partial s}\right|_{s=s_o} \leq D\left(\mathscr{W}\|P^\star \mathscr{W}| P_o\right) 
= C_\mathscr{W},
\end{align}
This, together with items \ref{prop_h-b} and \ref{prop_h-c} in Proposition \ref{prop:prop_h}, shows that the first inequality in Eq.~\eqref{eq:s_n2} becomes an equality if and only if $s_o=0$. Since the subsequence was arbitrary, item \ref{s_n-b} is shown.

\end{itemize}
\end{proof}

Now we are ready to prove this proposition.
We start with proving Eq.~\eqref{eq:spChc33}.
Since $s\mapsto E_\text{h}(s,\cdot)$ is concave from item \ref{prop_h-b} in Proposition \ref{prop:prop_h}, the maximizer $s_n^\star$ must satisfy
\begin{align} \label{eq:s_n3}
\left.\frac{\partial E_\text{h}(s,P_{n_k})}{\partial s}\right|_{s=s_{n_k}^\star} = C_\mathscr{W} - \delta_{n_k}.
\end{align}
Further, item \ref{prop_h-c} in Proposition \ref{prop:prop_h} gives
\begin{align}
\left.\frac{ \partial E_\text{h}\left(s,P_{n_k}^\star\right)}{\partial s}\right|_{s=0} = D\left(\mathscr{W} \| P^\star \mathscr{W} | P_{n_k}^\star \right).
\end{align}
The mean value theorem states that there exists a number $\hat{s}_{n_k} \in \left(0, s_{n_k}^\star\right)$, for each $k\geq \mathbb{N}$,  such that
\begin{align}
-\left.\frac{ \partial^2 E_\text{h}\left(s,P_{n_k}\right)}{\partial s^2}\right|_{s=\hat{s}_{n_k}}
&= \frac{ D\left( \mathscr{W} \| P^\star \mathscr{W} | P_{n_k} \right) - C_\mathscr{W} + \delta_{n_k}}{ s_{n_k}^\star} \\ 
&\leq \frac{\delta_{n_k}}{s_{n_k}^\star}, \label{eq:s_n4}
\end{align}
where the last inequality is again due to $D\left( \mathscr{W}\| P^\star \mathscr{W}| P_{n_k}^\star \right) \leq C_\mathscr{W}$.
When $k$ approaches infinity, items \ref{prop_h-a} and \ref{prop_h-e} in Proposition \ref{prop:prop_h}  give
\begin{align} \label{eq:s_n5}
\lim_{k\to+\infty} \left.\frac{\partial^2 E_\text{h}\left(s,P_{n_k}\right)}{\partial s^2}\right|_{s = \hat{s}_{n_k} }
= \left.\frac{\partial^2 E_\text{h}\left(s,P_{o}\right)}{\partial s^2}\right|_{s = 0 }
= -V(P_o,\mathscr{W}) \leq -V_\mathscr{W}.
\end{align}
Combining Eqs.~\eqref{eq:s_n4} and \eqref{eq:s_n5} leads to
\begin{align} \label{eq:s_nc}
\limsup_{k\to+\infty} \frac{ s_{n_k}^\star }{\delta_{n_k}} \leq \frac{1}{V_\mathscr{W}}.
\end{align}
Since the subsequence was arbitrary, the above result establishes Eq.~\eqref{eq:spChc33}.

Next, for any sufficiently large $n\geq N_0$, we apply Taylor's theorem to the map $s_n^\star \mapsto E_\text{h} \left( s_n^\star, P_n \right)$ at the original point to obtain
\begin{align}
& E_\textnormal{sp}^{(2)} \left(C_\mathscr{W} -\delta_n, P_n, P^\star \mathscr{W}\right)
\notag \\ 
&= -s_n^\star\left( C_\mathscr{W} - \delta_n \right) + E_\text{h}\left( s_n^\star, P_n \right) \\
&= s_n^\star \left(\delta_n + D(\mathscr{W}\|P^\star W|P_n)-  C_\mathscr{W}  \right) - \frac{(s_n^\star)^2}{2} V\left( P_n, \mathscr{W} \right) + \frac{(s_n^\star)^3}{6} \left.\frac{ \partial^3 E_\text{h}(s,P_n)}{ \partial s^3}\right|_{s=\bar{s}_n} 
\label{eq_D16}
\end{align}
for some $\bar{s}_n\in\left[0,s_n^\star\right]$.
Let
\begin{align} \label{eq:Upsilon2}
\Upsilon = \max_{ (s,P)\in[0,1]\times \mathscr{P}(\mathcal{X})} 
\left| \frac{\partial^3 {E}_\text{h}\left(s,{P}\right)}{\partial s^3} \right|.
\end{align}
Continuing from Eq.~\eqref{eq_D16} gives
\begin{align}
E_\textnormal{sp}^{(2)} \left(C_\mathscr{W} -\delta_n, P_n, P^\star \mathscr{W} \right)
&\leq s_n^\star  (\delta_n-b_n)  - \frac{(s_n^\star)^2}{2} V\left( P_n, \mathscr{W} \right) + \frac{(s_n^\star)^3 \Upsilon}{6} \\
& \leq\sup_{s\geq 0} \left\{ s (\delta_n-b_n) - \frac{s^2}{2}V(P_n, \mathscr{W})   \right\} + \frac{(s_n^\star)^3 \Upsilon}{6} \\
&= \frac{(\delta_n-b_n)^2}{2V(P_n, \mathscr{W})} + \frac{(s_n^\star)^3 \Upsilon}{6}, \label{eq:temp1}
\end{align}
where the first line follows from the assumption $D\left( \mathscr{W} \| P^\star \mathscr{W} | P_{n} \right) \leq C_\mathscr{W} - b_n$ in Eq.~\eqref{eq:spChcc} and Eq.~\eqref{eq:Upsilon2}.
Finally, Eq.~\eqref{eq:temp1}, along with item \ref{s_n-b} in Lemma \ref{lemm:s_n} and Eq.~\eqref{eq:s_nc}, implies that
\begin{align}
\limsup_{n\to+\infty} \frac{ E_\textnormal{sp}^{(2)} \left(C_\mathscr{W} -\delta_n, P_n, P^\star \mathscr{W} \right) }{\delta_n^2}
&\leq \limsup_{n\to+\infty} \frac{(\delta_n-b_n)^2}{2 V(P_n ,\mathscr{W}) \delta_n^2} 
\\
&\leq \limsup_{n\to+\infty} \frac{(\delta_n-b_n)^2}{2V_\mathscr{W} \delta_n^2},
\end{align}
where the last inequality follows from the continuity of $V(\,\cdot\,,\mathscr{W})$ on $\mathscr{P}(\mathcal{X})$ (Eq.~\eqref{eq:V}); the fact that $\{P_n\}_{n\in\mathbb{N}}$ is capacity achieving (item \ref{s_n-a} in Lemma \ref{lemm:s_n}); and the definition of $V_{\mathscr{W}}$ in Eq.~\eqref{eq:V2}.
\end{proof}



\begin{thebibliography}{10}
	\providecommand{\url}[1]{#1}
	\csname url@samestyle\endcsname
	\providecommand{\newblock}{\relax}
	\providecommand{\bibinfo}[2]{#2}
	\providecommand{\BIBentrySTDinterwordspacing}{\spaceskip=0pt\relax}
	\providecommand{\BIBentryALTinterwordstretchfactor}{4}
	\providecommand{\BIBentryALTinterwordspacing}{\spaceskip=\fontdimen2\font plus
		\BIBentryALTinterwordstretchfactor\fontdimen3\font minus
		\fontdimen4\font\relax}
	\providecommand{\BIBforeignlanguage}[2]{{%
			\expandafter\ifx\csname l@#1\endcsname\relax
			\typeout{** WARNING: IEEEtran.bst: No hyphenation pattern has been}%
			\typeout{** loaded for the language `#1'. Using the pattern for}%
			\typeout{** the default language instead.}%
			\else
			\language=\csname l@#1\endcsname
			\fi
			#2}}
	\providecommand{\BIBdecl}{\relax}
	\BIBdecl
	
	\normalsize
	
	\bibitem{Str62}
	V.~Strassen, ``Asymptotische absch{\"{a}}tzungen in {Shannon's}
	informationstheorie,'' \emph{Transactions of the Third Prague Conference on
		Information Theory}, pp. 689--723, 1962.
	
	\bibitem{Sha48}
	C.~E. Shannon, ``A mathematical theory of communication,'' \emph{The Bell
		System Technical Journal}, vol.~27, pp. 379--423, 1948.
	
	\bibitem{PPV10}
	Y.~Polyanskiy, H.~V. Poor, and S.~Verdu, ``Channel coding rate in the finite
	blocklength regime,''
	\href{http://dx.doi.org/10.1109/tit.2010.2043769}{\emph{{IEEE} Transactions on Information
			Theory}}, \href{http://dx.doi.org/10.1109/tit.2010.2043769}{vol.~56},
	\href{http://dx.doi.org/10.1109/tit.2010.2043769}{no.~5},
	\href{http://dx.doi.org/10.1109/tit.2010.2043769}{pp. 2307--2359},
	\href{http://dx.doi.org/10.1109/tit.2010.2043769}{May 2010}.
	
	\bibitem{Hay09b}
	M.~Hayashi, ``Information spectrum approach to second-order coding rate in
	channel coding,''
	\href{http://dx.doi.org/10.1109/tit.2009.2030478}{\emph{IEEE Transactions on
			Information Theory}},
	\href{http://dx.doi.org/10.1109/tit.2009.2030478}{vol.~55},
	\href{http://dx.doi.org/10.1109/tit.2009.2030478}{no.~11},
	\href{http://dx.doi.org/10.1109/tit.2009.2030478}{pp. 4947--4966},
	\href{http://dx.doi.org/10.1109/tit.2009.2030478}{Nov 2009}.
	
	\bibitem{TH13}
	M.~Tomamichel and M.~Hayashi, ``A hierarchy of information quantities for
	finite block length analysis of quantum tasks,''
	\href{http://dx.doi.org/10.1109/tit.2013.2276628}{\emph{{IEEE} Transactions
			on Information Theory}},
	\href{http://dx.doi.org/10.1109/tit.2013.2276628}{vol.~59},
	\href{http://dx.doi.org/10.1109/tit.2013.2276628}{no.~11},
	\href{http://dx.doi.org/10.1109/tit.2013.2276628}{pp. 7693--7710},
	\href{http://dx.doi.org/10.1109/tit.2013.2276628}{Nov 2013}.
	
	\bibitem{Li14}
	K.~Li, ``Second-order asymptotics for quantum hypothesis testing,''
	\href{http://dx.doi.org/10.1214/13-aos1185}{\emph{The Annals of Statistics}},
	\href{http://dx.doi.org/10.1214/13-aos1185}{vol.~42},
	\href{http://dx.doi.org/10.1214/13-aos1185}{no.~1},
	\href{http://dx.doi.org/10.1214/13-aos1185}{pp. 171--189},
	\href{http://dx.doi.org/10.1214/13-aos1185}{Feb 2014}.
	
	\bibitem{TV15}
	M.~Tomamichel and V.~Y.~F. Tan, ``Second-order asymptotics for the classical
	capacity of image-additive quantum channels,''
	\href{http://dx.doi.org/10.1007/s00220-015-2382-0}{\emph{Communications in
			Mathematical Physics}},
	\href{http://dx.doi.org/10.1007/s00220-015-2382-0}{vol. 338},
	\href{http://dx.doi.org/10.1007/s00220-015-2382-0}{no.~1},
	\href{http://dx.doi.org/10.1007/s00220-015-2382-0}{pp. 103--137},
	\href{http://dx.doi.org/10.1007/s00220-015-2382-0}{May 2015}.
	
	\bibitem{TBR16}
	M.~Tomamichel, M.~Berta, and J.~M. Renes, ``Quantum coding with finite
	resources,'' \href{http://dx.doi.org/10.1038/ncomms11419}{\emph{Nature
			Communications}}, \href{http://dx.doi.org/10.1038/ncomms11419}{vol.~7},
	\href{http://dx.doi.org/10.1038/ncomms11419}{p. 11419},
	\href{http://dx.doi.org/10.1038/ncomms11419}{May 2016}.
	
	\bibitem{TT13}
	M.~Tomamichel and V.~Y.~F. Tan, ``A tight upper bound for the third-order
	asymptotics for most discrete memoryless channels,''
	\href{http://dx.doi.org/10.1109/tit.2013.2276077}{\emph{IEEE Transactions on
			Information Theory}},
	\href{http://dx.doi.org/10.1109/tit.2013.2276077}{vol.~59},
	\href{http://dx.doi.org/10.1109/tit.2013.2276077}{no.~11},
	\href{http://dx.doi.org/10.1109/tit.2013.2276077}{pp. 7041--7051},
	\href{http://dx.doi.org/10.1109/tit.2013.2276077}{Nov 2013}.
	
	\bibitem{Tan14}
	V.~Y.~F. Tan, ``Asymptotic estimates in information theory with non-vanishing
	error probabilities,''
	\href{http://dx.doi.org/10.1561/0100000086}{\emph{Foundations and
			Trends{\textregistered} in Communications and Information Theory}},
	\href{http://dx.doi.org/10.1561/0100000086}{vol.~10},
	\href{http://dx.doi.org/10.1561/0100000086}{no.~4},
	\href{http://dx.doi.org/10.1561/0100000086}{pp. 1--184},
	\href{http://dx.doi.org/10.1561/0100000086}{2014}.
	
	\bibitem{TT15}
	V.~Y.~F. Tan and M.~Tomamichel, ``The third-order term in the normal
	approximation for the {AWGN} channel,''
	\href{http://dx.doi.org/10.1109/tit.2015.2411256}{\emph{IEEE Transactions on
			Information Theory}},
	\href{http://dx.doi.org/10.1109/tit.2015.2411256}{vol.~61},
	\href{http://dx.doi.org/10.1109/tit.2015.2411256}{no.~5},
	\href{http://dx.doi.org/10.1109/tit.2015.2411256}{pp. 2430--2438},
	\href{http://dx.doi.org/10.1109/tit.2015.2411256}{May 2015}.
	
	\bibitem{Sha59}
	C.~E. Shannon, ``Probability of error for optimal codes in a {Gaussian}
	channel,''
	\href{http://dx.doi.org/10.1002/j.1538-7305.1959.tb03905.x}{\emph{Bell System
			Technical Journal}},
	\href{http://dx.doi.org/10.1002/j.1538-7305.1959.tb03905.x}{vol.~38},
	\href{http://dx.doi.org/10.1002/j.1538-7305.1959.tb03905.x}{no.~3},
	\href{http://dx.doi.org/10.1002/j.1538-7305.1959.tb03905.x}{pp. 611--656},
	\href{http://dx.doi.org/10.1002/j.1538-7305.1959.tb03905.x}{May 1959}.
	
	\bibitem{Gal68}
	R.~Gallager, \href{http://as.wiley.com/WileyCDA/WileyTitle/productCd-0471290483.html}{\emph{Information Theory and Reliable Communication}. Wiley, 1968}.
	
	\bibitem{Fan61}
	R.~M. Fano, \emph{Transmission of Information, A Statistical Theory of
		Communications}.\hskip 1em plus 0.5em minus 0.4em\relax The MIT Press, 1961.
	
	\bibitem{Bla87}
	R.~E. Blahut, \emph{Principles and practice of information theory}.\hskip 1em
	plus 0.5em minus 0.4em\relax Addison-Wesley, 1987.
	
	\bibitem{HHH07}
	E.~A. Haroutunian, M.~E. Haroutunian, and A.~N. Harutyunyan, ``Reliability
	criteria in information theory and in statistical hypothesis testing,''
	\href{http://dx.doi.org/10.1561/0100000008}{\emph{Foundations and
			Trends{\textregistered} in Communications and Information Theory}},
	\href{http://dx.doi.org/10.1561/0100000008}{vol.~4},
	\href{http://dx.doi.org/10.1561/0100000008}{no. 2--3},
	\href{http://dx.doi.org/10.1561/0100000008}{pp. 97--263},
	\href{http://dx.doi.org/10.1561/0100000008}{2007}.
	
	\bibitem{CK11}
	I.~Csisz{\'a}r and J.~K{\"o}rner, \emph{Information Theory: Coding Theorems for
		Discrete Memoryless Systems}.\hskip 1em plus 0.5em minus 0.4em\relax
	Cambridge University Press ({CUP}), 2011.
	
	\bibitem{CHT16b}
	H.-C. Cheng, M.-H. Hsieh, and M.~Tomamichel, ``Quantum Sphere-Packing Bounds with Polynomial Prefactors,'' \href{http://arxiv.org.abs/1704.05703}{\texttt{arXiv:1704.05703 [quant-ph]}}.
	
	\bibitem{DZ98}
	A.~Dembo and O.~Zeitouni, \emph{Large Deviations Techniques and
		Applications}.\hskip 1em plus 0.5em minus 0.4em\relax Springer, 1998.
	
	\bibitem{AW10}
	Y.~Altu{\u{g}} and A.~B. Wagner, ``Moderate deviation analysis of channel
	coding: Discrete memoryless case,'' in \href{http://dx.doi.org/10.1109/isit.2010.5513319}{\emph{2010 {IEEE} International
		Symposium on Information Theory}. Jun 2010}.
	
	\bibitem{AW14b}
	------, ``Moderate deviations in channel coding,''
	\href{http://dx.doi.org/10.1109/tit.2014.2323418}{\emph{IEEE Transactions on
			Information Theory}},
	\href{http://dx.doi.org/10.1109/tit.2014.2323418}{vol.~60},
	\href{http://dx.doi.org/10.1109/tit.2014.2323418}{no.~8},
	\href{http://dx.doi.org/10.1109/tit.2014.2323418}{pp. 4417--4426},
	\href{http://dx.doi.org/10.1109/tit.2014.2323418}{Aug 2014}.
	
	\bibitem{PV10}
	Y.~Polyanskiy and S.~Verdu, ``Channel dispersion and moderate deviations limits
	for memoryless channels,'' in \href{http://dx.doi.org/10.1109/tit.2010.2043769}{\emph{2010 48th Annual Allerton Conference on
		Communication, Control, and Computing}. Sep
	2010}.

	\bibitem{Win99}
	A.~Winter, ``Coding theorems of quantum information theory,''  \href{http://arxiv.org/abs/quant-ph/9907077}{\textit{PhD
			Thesis, Universit{\"{a}}t Bielefeld}, 1999}. 
		
	\bibitem{Sas12}
	I.~Sason, ``Moderate deviations analysis of binary hypothesis testing,'' in
	\href{http://dx.doi.org/10.1109/isit.2012.6284675}{\emph{2012 {IEEE} International Symposium on Information Theory
		Proceedings}. Jul 2012}.

	\bibitem{WH14}
	S.~Watanabe and M.~Hayashi, ``Finite-length analysis on tail probability for Markov chain and application to simple hypothesis testing,"
	\href{http://arxiv.org/abs/1401.3801}{\texttt{arXiv:1401.3801}}.
	
	\bibitem{RD16}
	C.~Rouz\'{e} and N.~Datta, 
	``Finite blocklength and moderate deviation analysis of hypothesis testing of correlated quantum states and application to classical-quantum channels with memory,"
	\href{http://arxiv.org/abs/1612.01464}{\texttt{arXiv:1612.01464}}.	
	
	\bibitem{Hay07}
	M.~Hayashi, ``Error exponent in asymmetric quantum hypothesis testing and its
	application to classical-quantum channel coding,''
	\href{http://dx.doi.org/10.1103/physreva.76.062301}{\emph{Physical Review
			A}}, \href{http://dx.doi.org/10.1103/physreva.76.062301}{vol.~76},
	\href{http://dx.doi.org/10.1103/physreva.76.062301}{no.~6},
	\href{http://dx.doi.org/10.1103/physreva.76.062301}{dec 2007}.
	
	\bibitem{CHT16bb}
	H.-C.~Cheng, M.-H.~Hsieh and M.~Tomamichel, ``Sphere-Packing Bound for Symmetric Classical-Quantum Channels,'' \href{http://arxiv.org/abs/1701.02957}{\texttt{arXiv:1701.02957}}.	
	
	
	\bibitem{DW14}
	M.~Dalai and A.~Winter, ``Constant compositions in the sphere packing bound for
	classical-quantum channels,'' in \href{http://dx.doi.org/10.1109/isit.2014.6874813}{\emph{2014 {IEEE} International Symposium on
		Information Theory}. Jun 2014}.
	
	\bibitem{ANS+08}
	K.~M.~R. Audenaert, M.~Nussbaum, A.~Szko{\l}a, and F.~Verstraete, ``Asymptotic
	error rates in quantum hypothesis testing,''
	\href{http://dx.doi.org/10.1007/s00220-008-0417-5}{\emph{Communications in
			Mathematical Physics}},
	\href{http://dx.doi.org/10.1007/s00220-008-0417-5}{vol. 279},
	\href{http://dx.doi.org/10.1007/s00220-008-0417-5}{no.~1},
	\href{http://dx.doi.org/10.1007/s00220-008-0417-5}{pp. 251--283},
	\href{http://dx.doi.org/10.1007/s00220-008-0417-5}{Feb 2008}.
	
	
	\bibitem{CCT+16b}
	C.~T. Chubb, V.~Y.~F. Tan, and M.~Tomamichel, ``Moderate deviation analysis for
	classical communication over quantum channels,'' \href{http://arxiv.org/abs/1701.03114}{\texttt{arXiv:1701.03114 [quant-ph]}}.	
	
	\bibitem{WR12}
	L.~Wang and R.~Renner, ``One-shot classical-quantum capacity and hypothesis
	testing,''
	\href{http://dx.doi.org/10.1103/physrevlett.108.200501}{\emph{Physical Review
			Letters}}, \href{http://dx.doi.org/10.1103/physrevlett.108.200501}{vol. 108},
	\href{http://dx.doi.org/10.1103/physrevlett.108.200501}{no.~20},
	\href{http://dx.doi.org/10.1103/physrevlett.108.200501}{May 2012}.
	
	\bibitem{Roz02}
	L.~V. Rozovsky, ``Estimate from below for large-deviation probabilities of a
	sum of independent random variables with finite variances,''
	\href{http://dx.doi.org/10.1023/A:1014589618720}{\emph{Journal of
			Mathematical Sciences}},
	\href{http://dx.doi.org/10.1023/A:1014589618720}{vol. 109},
	\href{http://dx.doi.org/10.1023/A:1014589618720}{pp. 2192--2209},
	\href{http://dx.doi.org/10.1023/A:1014589618720}{2002}.
	
	\bibitem{NH07}
	H.~Nagaoka and M.~Hayashi, ``An information-spectrum approach to classical and
	quantum hypothesis testing for simple hypotheses,''
	\href{http://dx.doi.org/10.1109/tit.2006.889463}{\emph{IEEE Transactions on
			Information Theory}},
	\href{http://dx.doi.org/10.1109/tit.2006.889463}{vol.~53},
	\href{http://dx.doi.org/10.1109/tit.2006.889463}{no.~2},
	\href{http://dx.doi.org/10.1109/tit.2006.889463}{pp. 534--549},
	\href{http://dx.doi.org/10.1109/tit.2006.889463}{Feb 2007}.
	
	\bibitem{BH98}
	M.~V. Burnashev and A.~S. Holevo, ``On the reliability function for a quantum
	communication channel,'' \emph{Problems of information transmission},
	vol.~34, no.~2, pp. 97--107, 1998.
	
	\bibitem{Hol00}
	A.~Holevo, ``Reliability function of general classical-quantum channel,''
	\href{http://dx.doi.org/10.1109/18.868501}{\emph{{IEEE} Transaction on
			Information Theory}}, \href{http://dx.doi.org/10.1109/18.868501}{vol.~46},
	\href{http://dx.doi.org/10.1109/18.868501}{no.~6},
	\href{http://dx.doi.org/10.1109/18.868501}{pp. 2256--2261},
	\href{http://dx.doi.org/10.1109/18.868501}{2000}.
	
	\bibitem{FY05}
	J.~I.~Fujii, R.~Nakamoto, and K.~Yanagi,``Remarks on concavity of the auxiliary function appearing in quantum reliability function in classical-quantum channels,'' in
	\href{http://dx.doi.org/isit.2005.1523466}{\emph{2005 {IEEE} International
			Symposium on Information Theory}. Jun 2005}.

	
	\bibitem{FY06}
	J.~I.~Fujii, R.~Nakamoto, and K.~Yanagi, ``Concavity of the auxiliary function appearing in quantum reliability function,''
	\href{http://dx.doi.org/10.1109/tit.2006.876248}{\emph{{IEEE} Transaction on
			Information Theory}}, \href{http://dx.doi.org/10.1109/10.1109/tit.2006.876248}{vol.~52},
	\href{http://dx.doi.org/10.1109/10.1109/tit.2006.876248}{no.~7},
	\href{http://dx.doi.org/10.1109/10.1109/tit.2006.876248}{pp. 3310--3313},
	\href{http://dx.doi.org/10.1109/10.1109/tit.2006.876248}{2006}.
	
	
	
	\bibitem{HM16}
	H.-C. Cheng and M.-H. Hsieh, ``Concavity of the auxiliary function for
	classical-quantum channels,''
	\href{http://dx.doi.org/10.1109/TIT.2016.2598835}{\emph{{IEEE} Transactions
			on Information Theory}},
	\href{http://dx.doi.org/10.1109/TIT.2016.2598835}{vol.~62},
	\href{http://dx.doi.org/10.1109/TIT.2016.2598835}{no.~10},
	\href{http://dx.doi.org/10.1109/TIT.2016.2598835}{pp. 5960 -- 5965},
	\href{http://dx.doi.org/10.1109/TIT.2016.2598835}{2016}.
	
	\bibitem{Win99}
	A.~Winter, ``Coding theorems of quantum information theory,'' 1999, (PhD
		Thesis, Universit{\"{a}}t Bielefeld). \href{http://arxiv.org/abs/quant-ph/9907077}{\texttt{arXiv:quant-ph/9907077}}.
	
	\bibitem{SGB67}
	C.~Shannon, R.~Gallager, and E.~Berlekamp, ``Lower bounds to error probability
	for coding on discrete memoryless channels. {I},''
	\href{http://dx.doi.org/10.1016/s0019-9958(67)90052-6}{\emph{Information and
			Control}}, \href{http://dx.doi.org/10.1016/s0019-9958(67)90052-6}{vol.~10},
	\href{http://dx.doi.org/10.1016/s0019-9958(67)90052-6}{no.~1},
	\href{http://dx.doi.org/10.1016/s0019-9958(67)90052-6}{pp. 65--103},
	\href{http://dx.doi.org/10.1016/s0019-9958(67)90052-6}{Jan 1967}.
	
	\bibitem{Bou06}
	J.-C.~Bourin, ``Matrix versions of some classical inequalities,"
	\href{http://dx.doi.org/10.1016/j.laa.2006.01.002}{Linear Algebra and its Applications},
	\href{http://dx.doi.org/10.1016/j.laa.2006.01.002}{vol.~446},	
	\href{http://dx.doi.org/10.1016/j.laa.2006.01.002}{no.~2--3},		
	\href{http://dx.doi.org/10.1016/j.laa.2006.01.002}{pp.~890--907},			
	\href{http://dx.doi.org/10.1016/j.laa.2006.01.002}{July 2006}.
	
%
	
%
	\bibitem{Pet86}
	D.~Petz, ``Quasi-entropies for finite quantum systems,''
	\href{http://dx.doi.org/10.1016/0034-4877(86)90067-4}{\emph{Reports on
			Mathematical Physics}},
	\href{http://dx.doi.org/10.1016/0034-4877(86)90067-4}{vol.~23},
	\href{http://dx.doi.org/10.1016/0034-4877(86)90067-4}{no.~1},
	\href{http://dx.doi.org/10.1016/0034-4877(86)90067-4}{pp. 57--65},
	\href{http://dx.doi.org/10.1016/0034-4877(86)90067-4}{Feb 1986}.
	
	\bibitem{MO15}
	M.~Mosonyi and T.~Ogawa, ``Two approaches to obtain the strong converse
	exponent of quantum hypothesis testing for general sequences of quantum
	states,'' \href{http://dx.doi.org/10.1109/tit.2015.2489259}{\emph{{IEEE}
			Transactions on Information Theory}},
	\href{http://dx.doi.org/10.1109/tit.2015.2489259}{vol.~61},
	\href{http://dx.doi.org/10.1109/tit.2015.2489259}{no.~12},
	\href{http://dx.doi.org/10.1109/tit.2015.2489259}{pp. 6975--6994},
	\href{http://dx.doi.org/10.1109/tit.2015.2489259}{Dec 2015}.
	
	\bibitem{LT15}
	S.~M. Lin and M.~Tomamichel, ``Investigating properties of a family of quantum
	r{\'{e}}nyi divergences,''
	\href{http://dx.doi.org/10.1007/s11128-015-0935-y}{\emph{Quantum Information
			Processing}}, \href{http://dx.doi.org/10.1007/s11128-015-0935-y}{vol.~14},
	\href{http://dx.doi.org/10.1007/s11128-015-0935-y}{no.~4},
	\href{http://dx.doi.org/10.1007/s11128-015-0935-y}{pp. 1501--1512},
	\href{http://dx.doi.org/10.1007/s11128-015-0935-y}{Feb 2015}.
	
	\bibitem{Tom16}
	M.~Tomamichel, \emph{Quantum Information Processing with Finite
		Resources}.\hskip 1em plus 0.5em minus 0.4em\relax Springer International
	Publishing, 2016.
	
	\bibitem{Hig08}
	N.~J. Higham, \href{http://dx.doi.org/10.1137/1.9780898717778}{\emph{Functions of Matrices: Theory and Computation}.} \href{http://dx.doi.org/10.1137/1.9780898717778}{
	{SIAM}, Jan 2008}.
	
	\bibitem{MO14b}
	M.~Mosonyi and T.~Ogawa, ``Strong converse exponent for classical-quantum
	channel coding,'' 2014 \href{http://arxiv.org/abs/1409.3562}{\texttt{arXiv:1409.3562 [quant-ph]}}.
	
	\bibitem{Bha97}
	R.~Bhatia, \emph{Matrix Analysis}, 
	\href{http://dx.doi.org/10.1007/978-1-4612-0653-8}{Springer New York},
	\href{http://dx.doi.org/10.1007/978-1-4612-0653-8}{1997}.
	
	\bibitem{Bla74}
	R.~E. Blahut, ``Hypothesis testing and information theory,''
	\href{http://dx.doi.org/10.1109/tit.1974.1055254}{\emph{{IEEE} Transaction on
			Information Theory}},
	\href{http://dx.doi.org/10.1109/tit.1974.1055254}{vol.~20},
	\href{http://dx.doi.org/10.1109/tit.1974.1055254}{no.~4},
	\href{http://dx.doi.org/10.1109/tit.1974.1055254}{pp. 405--417},
	\href{http://dx.doi.org/10.1109/tit.1974.1055254}{Jul 1974}.
	
	\bibitem{SW12}
	N.~Sharma and N.~A. Warsi, ``Fundamental bound on the reliability of quantum
	information transmission,''
	\href{http://dx.doi.org/10.1103/physrevlett.110.080501}{\emph{Physical Review
			Letters}}, \href{http://dx.doi.org/10.1103/physrevlett.110.080501}{vol. 110},
	\href{http://dx.doi.org/10.1103/physrevlett.110.080501}{no.~8},
	\href{http://dx.doi.org/10.1103/physrevlett.110.080501}{Feb 2013}.
	
	\bibitem{CS93}
	N.~R.~Chaganty and J.~Sethuraman, ``Strong Large Deviation and Local Limit Theorems,''
	\href{http://dx.doi.org/10.1214/aop/1176989136}{\emph{The Annals of Probability}},
	\href{http://dx.doi.org/10.1214/aop/1176989136}{vol.~21},
	\href{http://dx.doi.org/10.1214/aop/1176989136}{no.~3},
	\href{http://dx.doi.org/10.1214/aop/1176989136}{pp.~1671--1690},
	\href{http://dx.doi.org/10.1214/aop/1176989136}{1993}.	
	
	\bibitem{Sib69}
	R.~Sibson, ``Information radius,''
	\href{http://dx.doi.org/10.1007/BF00537520}{\emph{{Zeitschrift f\"ur
				Wahrscheinlichkeitstheorie und Verwandte Gebiete}}},
	\href{http://dx.doi.org/10.1007/BF00537520}{vol.~14},
	\href{http://dx.doi.org/10.1007/BF00537520}{no.~2},
	\href{http://dx.doi.org/10.1007/BF00537520}{pp. 149--160},
	\href{http://dx.doi.org/10.1007/BF00537520}{1969}.
	
	\bibitem{Csi95}
	I.~{Csisz\'ar}, ``Generalized cutoff rates and {R\'enyi's} information
	measures,'' \href{http://dx.doi.org/10.1109/18.370121}{\emph{IEEE
			Transactions on Information Theory}},
	\href{http://dx.doi.org/10.1109/18.370121}{vol.~41},
	\href{http://dx.doi.org/10.1109/18.370121}{no.~1},
	\href{http://dx.doi.org/10.1109/18.370121}{pp. 26--34},
	\href{http://dx.doi.org/10.1109/18.370121}{1995}.
	
	\bibitem{MH11}
	M.~Mosonyi and F.~Hiai, ``On the quantum {R{\'e}nyi} relative entropies and
	related capacity formulas,''
	\href{http://dx.doi.org/10.1109/tit.2011.2110050}{\emph{IEEE Transactions on
			Information Theory}},
	\href{http://dx.doi.org/10.1109/tit.2011.2110050}{vol.~57},
	\href{http://dx.doi.org/10.1109/tit.2011.2110050}{no.~4},
	\href{http://dx.doi.org/10.1109/tit.2011.2110050}{pp. 2474--2487},
	\href{http://dx.doi.org/10.1109/tit.2011.2110050}{Apr 2011}.
	
	\bibitem{WWY14}
	M.~M. Wilde, A.~Winter, and D.~Yang, ``Strong converse for the classical
	capacity of entanglement-breaking and {Hadamard} channels via a sandwiched
	{R{\'{e}}nyi} relative entropy,''
	\href{http://dx.doi.org/10.1007/s00220-014-2122-x}{\emph{Communications in
			Mathematical Physics}},
	\href{http://dx.doi.org/10.1007/s00220-014-2122-x}{vol. 331},
	\href{http://dx.doi.org/10.1007/s00220-014-2122-x}{no.~2},
	\href{http://dx.doi.org/10.1007/s00220-014-2122-x}{pp. 593--622},
	\href{http://dx.doi.org/10.1007/s00220-014-2122-x}{Jul 2014}.
	
	\bibitem{HP14}
	F.~Hiai and D.~Petz, \emph{Introduction to Matrix Analysis and
		Applications}.\hskip 1em plus 0.5em minus 0.4em\relax Springer International
	Publishing, 2014.
	
	\bibitem{NS09}
	M.~Nussbaum and A.~Szko{\l}a, ``The {Chernoff} lower bound for symmetric
	quantum hypothesis testing,''
	\href{http://dx.doi.org/10.1214/08-aos593}{\emph{Annals of Statistics}},
	\href{http://dx.doi.org/10.1214/08-aos593}{vol.~37},
	\href{http://dx.doi.org/10.1214/08-aos593}{no.~2},
	\href{http://dx.doi.org/10.1214/08-aos593}{pp. 1040--1057},
	\href{http://dx.doi.org/10.1214/08-aos593}{apr 2009}.
	
	\bibitem{Nag06}
	H.~Nagaoka, ``The converse part of the theorem for quantum {Hoeffding} bound,'' 2006
	\href{http://arxiv.org/abs/quant-ph/0611289}{\texttt{arXiv:quant-ph/0611289}}.
	
	\bibitem{AW14}
	Y.~Altu{\u{g}} and A.~B. Wagner, ``Refinement of the sphere-packing bound:
	Asymmetric channels,''
	\href{http://dx.doi.org/10.1109/tit.2014.2299275}{\emph{{IEEE} Transactions
			on Information Theory}},
	\href{http://dx.doi.org/10.1109/tit.2014.2299275}{vol.~60},
	\href{http://dx.doi.org/10.1109/tit.2014.2299275}{no.~3},
	\href{http://dx.doi.org/10.1109/tit.2014.2299275}{pp. 1592--1614},
	\href{http://dx.doi.org/10.1109/tit.2014.2299275}{Mar 2014}.
	
	\bibitem{SW01}
	B.~Schumacher and M.~D. Westmoreland, ``Optimal signal ensembles,''
	\href{http://dx.doi.org/10.1103/physreva.63.022308}{\emph{Physical Review
			A}}, \href{http://dx.doi.org/10.1103/physreva.63.022308}{vol.~63},
	\href{http://dx.doi.org/10.1103/physreva.63.022308}{no.~2},
	\href{http://dx.doi.org/10.1103/physreva.63.022308}{Jan 2001}.

	\bibitem{WWD16}
	X.~Wang, W.~Xie, and R.~Duan, ``Semidefinite programming strong converse bounds for classical capacity," \href{http://arxiv.org/abs/1610.06381}{\texttt{arXiv:1610.06381 [quant-ph]}}.

	
\end{thebibliography}

\end{document}